\documentclass[acmsmall]{acmart}
\AtBeginDocument{%
  }

\setcopyright{acmlicensed}
\copyrightyear{2018}
\acmYear{2018}
\acmDOI{XXXXXXX.XXXXXXX}

\acmJournal{JDS}
\acmVolume{37}
\acmNumber{4}
\acmArticle{111}
\acmMonth{8}

\usepackage{float}
\usepackage{natbib}
\usepackage{flushend}
\usepackage[normalem]{ulem}
\usepackage{algorithmic}
\numberwithin{equation}{section}
\usepackage{url}
\usepackage{graphicx}
\usepackage{subfig}

\usepackage{verbatim}
\usepackage{microtype}
\usepackage{xspace}
\usepackage{amsmath}
\usepackage{amsthm}

\usepackage{amssymb}

\usepackage{newtxmath}
\usepackage{algorithmic}
\usepackage[ruled, vlined, linesnumbered]{algorithm2e}
\usepackage{booktabs}

\usepackage{wrapfig}
\usepackage{microtype}
\usepackage{xcolor}
\usepackage{color}
\usepackage{multirow}
\usepackage{balance}
\usepackage{nowidow}
\usepackage{import}
\usepackage{mathtools}
\usepackage{listings}
\usepackage{array}
\usepackage{moresize}
\usepackage{hhline}
\usepackage{adjustbox}
\usepackage{color, colortbl}
\usepackage{pifont}

\newcommand{\zqzuo}[1]{\textcolor{blue}{\textit{[TODO: #1]}}}

\newcommand{\MyPara}[1]{\vspace{3pt}\noindent\textbf{\textit{#1}}\hspace{0.5em}}

\newcommand{\eg}{\hbox{\emph{e.g.}}\xspace}
\newcommand{\ie}{\hbox{\emph{i.e.}}\xspace}

\theoremstyle{remark}

\newtheorem{theorem}{Theorem}

\usepackage{cleveref}

\newcommand*{\tool}{BigDataflow\xspace}
\newcommand*{\toolwho}{BigDataflow-whole\xspace}
\newcommand*{\toolinc}{BigDataflow-incremental\xspace}
\newcommand*{\toolchianina}{Chianina\xspace}
\newcommand*{\toolnaive}{BigDataflow-classic\xspace}

\newcommand{\change}[1]{\textcolor{black}{#1}}

\newcommand*{\toolincre}{BigDataflow-incremental\xspace}

\definecolor{Gray}{gray}{0.6}
\definecolor{backcolour}{rgb}{0.95,0.95,0.92}
\definecolor{codegreen}{rgb}{0,0.6,0}

\renewcommand{\cref}{\S \ref }

\lstdefinestyle{myStyle}{
    backgroundcolor=\color{backcolour},   
    commentstyle=\color{codegreen},
    basicstyle=\ttfamily\scriptsize,
    keywordstyle=\bfseries\color{blue},
    breakatwhitespace=false,         
    breaklines=true,                 
    keepspaces=true,                 
    numbers=left,       
    numbersep=2pt,                  
    showspaces=false,                
    showstringspaces=false,
    showtabs=false,                  
    tabsize=2,
    xleftmargin=1em,
}
\begin{document}

\title{Scaling Inter-procedural Dataflow Analysis on the Cloud}
\titlenote{\change{This submission is an extended version of a conference paper \cite{fse-23-bigdataflow} previously published in FSE'2023. Compared to the conference version, this manuscript additionally designed and implemented an incremental dataflow analysis framework, named \toolinc. To be specific, two incremental analysis algorithms are proposed, which are carefully elaborated in \S\ref{subsec:incremental-naive} and \S\ref{subsec:incremental-opt}. The formal correctness proofs are also given in \S\ref{subsec:proof-incre}. We conducted the empirical evaluations on \toolinc. Section \ref{subsec:eval-increment} discusses the detailed results. Moreover, we added the clarifications about the implementation and programming model of \toolinc in Section \ref{sec:implementation} and \ref{sec:model}, respectively. More related work about incremental static analysis is also presented in \S\ref{subsec:related-incremental}.}}

\author{Zewen Sun}
\authornote{Both authors contributed equally. Zewen Sun partially worked on this project while she interned at Alibaba Cloud Computing.}
\email{sunzew@smail.nju.edu.cn}
\author{Yujin Zhang}
\authornotemark[2]
\email{yujinz@smail.nju.edu.cn}
\author{Duanchen Xu}
\email{mf1933108@smail.nju.edu.cn}
\author{Yiyu Zhang}
\email{zhangyy0721@smail.nju.edu.cn}
\author{Yun Qi}
\email{mf20330058@smail.nju.edu.cn}
\author{Yueyang Wang}
\email{181860105@smail.nju.edu.cn}
\author{Yi Li}
\email{yili_cs@smail.nju.edu.cn}
\author{Zhaokang Wang}
\email{wang.zk@foxmail.com}
\author{Yue Li}
\email{yueli@nju.edu.cn}
\author{Xuandong Li}
\email{lxd@nju.edu.cn}
\author{Zhiqiang Zuo}
\authornote{Corresponding author.}
\email{zqzuo@nju.edu.cn}
\affiliation{%
  \institution{State Key Laboratory for Novel Software Technology at Nanjing University}
  \country{China}
}
    
    \author{Qingda Lu}
    \affiliation{%
      \institution{Alibaba Cloud Computing}
      \country{United States}
      }
    \email{qingda.lu@alibaba-inc.com}

    \author{Wenwen Peng}
    \affiliation{%
      \institution{Alibaba Cloud Computing}
      \country{China}
      }
    \email{wenwen.pww@alibaba-inc.com}

    \author{Shengjian Guo}
    \affiliation{%
      \institution{Baidu Research}
      \country{United States}
      }
    \email{guosj@vt.edu}
	

    \renewcommand{\shortauthors}{Sun et al.}

\begin{abstract}
	\change{Apart from forming the backbone of compiler optimization, static dataflow analysis has been widely applied in a vast variety of applications, such as bug detection, privacy analysis, program comprehension, etc. Despite its importance, performing interprocedural dataflow analysis on large-scale programs is well known to be challenging. In this paper, we propose a novel distributed analysis framework supporting the general interprocedural dataflow analysis. Inspired by large-scale graph processing, we devise dedicated distributed worklist algorithms for both whole-program analysis and incremental analysis. We implement these algorithms and develop a distributed framework called BigDataflow running on a large-scale cluster. The experimental results validate the promising performance of BigDataflow -- BigDataflow can finish analyzing the program of millions lines of code in minutes. Compared with the state-of-the-art, BigDataflow achieves much more analysis efficiency.}
\end{abstract}


\maketitle









\section{Introduction}
\label{sec:intro}

Dataflow analysis is a technique for statically gathering program information at program points along the program's control flow. 
Besides forming the backbone of compiler optimization, it has been adopted in many other significant application areas, including bug detection~\cite{sui-2016-svf-cc,Sadowski:2018-cacm-google}, security vulnerability discovery~\cite{hybridroid-ase-2016}, privacy analysis~\cite{Arzt:2014-pldi}, program testing/debugging~\cite{debugging-dataflow-pldi94,dataflowtesting-acms}, etc. 
In a dataflow analysis, a separate dataflow fact is maintained at each program point under the control flow graph (CFG) representation.
Based on the effect of each statement, a transfer function is applied to transform the dataflow fact accordingly along the CFG. The transformation process is performed iteratively via a worklist algorithm until a fixed point is reached~\cite{David:97-pepm}, meaning that all the dataflow facts are unchanged anymore.

\MyPara{Challenges.} Despite its importance, performing interprocedural dataflow analysis on large-scale systems code is well known to be challenging. 
First, as modern real-world programs are usually of large scale (like million lines of code), maintaining solutions at all program points with limited memory can hardly be scalable. Even worse, for certain analysis, the dataflow solution maintained at each point itself is highly space-intensive.
Although prior work attempts to adopt sparse representations \cite{Choi:1991-acs-popl,Ramalingam:2002-ser-tcs,Hardekopf-2009-popl,sui-2016-svf-cc}, the huge memory consumption still severely limits the scalability. 
As evidenced by recent studies~\cite{Shi:2018-pldi-pinpoint,zuo2021chianina}, the analysis over sparse value-flow graph can easily exceed hundreds of Gigabytes, showing the memory consumption a factual bottleneck. 
Second, the computation of flow-sensitive analysis requires updating the dataflow fact with respect to each statement along the CFG by performing the transfer function.
The process is highly computation-intensive because: 
(1) the amount of transfer function executed is at least linear in the number of program statements, which is large-scale given the modern large-size software under analysis;  
(2) the computation of each transfer function is perhaps expensive as well. For instance, in the flow-sensitive pointer/alias analysis,  the dataflow fact at each program point should capture the alias information among all the variables in the entire program. Updating variable relations by each transfer function consumes high CPU cycles. 

\MyPara{State-of-the-Art.}
To accelerate interprocedural dataflow analysis, a few attempts  to distribute/parallelize the computation have been made.
For distributed approaches, Garbervetsky et al. \cite{Garbervetsky:2017-fse} presented a distributed worklist algorithm on the basis of the actor model. However as stated explicitly in their paper, it cannot support the standard dataflow analysis due to the lack of flow ordering between actors. Albarghouthi et al. \cite{Albarghouthi:2012-pldi-bolt} parallelized the demand-driven top-down analyses based on MapReduce paradigm. They only targeted verification and software model checking without supporting dataflow analysis. BigSpa \cite{Zuo:2018-bigspa-ipdps,bigspa-tpds} supports the distributed acceleration for CFL reachability-based analysis \cite{Reps-1997-reachability}. Unfortunately, a lot of dataflow analyses, e.g., cache analysis and numerical analysis, do not belong to this category. 
Greathouse  et al. \cite{Greathouse-cgo-2011} proposed scalable dataflow analysis. However, they focused on dynamic analysis rather than static analysis. 
In brief, there exist no distributed systems supporting static dataflow analysis.

As for parallel approaches, Lee and Ryder \cite{lee1992comprehensive} exploited algorithmic parallelism to accelerate dataflow analysis. 
Rodriguez et al. \cite{Rodriguez-2011} proposed an actor model-based parallel algorithm for interprocedural finite distributive subset (IFDS) analysis~\cite{Reps-1995-dataflow}. 
Moreover, some researchers \cite{nagaraj2013parallel,su2014parallel} also studied parallel algorithms for pointer analysis. 
Note that the above approaches only support specific analysis rather than the general class of dataflow analyses.
More importantly, they rely heavily on local shared memory for computation. There is no doubt that they can rarely scale to large systems such as Linux kernel \cite{Aiken:2007-paste-saturn,zuo2021chianina}.
Recently, Zuo et al. \cite{zuo2021chianina} developed \toolchianina, a single machine-based analysis framework which can scale general dataflow analysis to millions lines of code.
Unfortunately, due to the involvement of disks, it readily takes hours or even days to finish the analysis for large-scale programs. Such inefficiency can hardly meet the requirement of quick analysis response (usually in minutes) in the modern continuous integration and deployment (CI/CD) pipelines \cite{Tricorder-icse-15,inputsplit-fse-2022}.

\MyPara{Our Work.} 
With the advent of cloud computing, the large-scale distributed cluster of commodity computers has nowadays become prevalent. It not only offers powerful computing capability, but can be easily accessible by a single developer.
Exploiting cloud resources for static analysis would be a promising breakthrough point for achieving both scalability and efficiency. 
However, as mentioned earlier, there exists no distributed system running on a cluster which can support the general dataflow analysis. Adapting the existing parallel algorithms (such as \toolchianina) to distributed environment is non-trivial. Parallel algorithms only focus on computation on shared memory, which lacks the consideration of partitioning, task dispatching, fault tolerance, and efficient communications between cluster nodes. 
None of the existing parallel approaches can directly do it without re-designing and re-implementing the system.


\change{
In this work, we propose a novel system that can leverage large-scale distributed cloud resources to scale and accelerate the general class of interprocedural dataflow analyses.
Inspired by large-scale graph processing \cite{Malewicz:2010-pregel,mccune2015thinking,Khan2016VC}, we revisit the traditional worklist algorithm from the perspective of  \emph{distributed vertex-centric computation model}, and devise a dedicated distributed worklist algorithm tailored for interprocedural dataflow analysis.
In addition, to further improve the analysis efficiency, an incremental distributed algorithm is proposed to realize incremental dataflow analysis.
We implement the distributed algorithms atop the general distributed graph processing platform (\ie, Apache Giraph \cite{giraph-book,giraph-vldb}) and develop a framework named \tool which supports both whole-program analysis (\ie, \toolwho) and incremental analysis (\ie, \toolinc), running on the cloud so as to take full advantage of the modern distributed computing resources.}

The underlying platform (\ie, Apache Giraph) provides the basic functionalities to support reliable and robust distributed processing, including input partitioning, task dispatching, cross-node communications, and fault tolerance.
\tool, as a generic framework, provides several APIs to specify the transfer functions and merge operator similar to other monotone dataflow frameworks \cite{Chambers-1996-tr-uw,Nielson:1999:PPA-sv}, thus alleviating the burden of implementing various client analyses.
By filling these APIs, users can readily implement a particular dataflow analysis on top of \tool.



\MyPara{Contributions.}
The contributions are listed as follows:
\begin{itemize}
\item We devise an optimized distributed vertex-centric computation model to accelerate static dataflow analysis by leveraging large-scale cloud resources.  
\item \change{We develop and implement a distributed dataflow analysis framework called \tool  running on a real-world cloud, which supports two modes: whole-program analysis and incremental analysis, while providing a variety of high-level APIs to easily implement client dataflow analyses.} 
\item We evaluate the performance and scalability of \tool over large-scale real-world software systems (e.g., Firefox and Linux kernel). The experimental results validate the promising performance of \tool---it can finish analyzing millions lines of code in minutes provided that a cluster of 125 commodity PCs.
\end{itemize}

\MyPara{Outlines.}
\change{
The rest of the paper is organized as follows. \S \ref{sec:background} gives the necessary background of dataflow analysis and distributed graph processing.  \S \ref{sec:design} and \S \ref{sec:design-incre} present the distributed worklist algorithms proposed for whole-program analysis and incremental analysis, respectively, followed by the implementation details in \S \ref{sec:implementation}. We discuss the programming model provided by our framework to implement various client analyses in \S \ref{sec:model}. \S \ref{sec:evaluation} describes the empirical evaluation of \tool in terms of performance and scalability. We give certain discussions in \S \ref{sec:discuss} and review the related work in \S \ref{sec:related}. Finally, \S \ref{sec:conclude} concludes.
}

\section{Background}
\label{sec:background}

\subsection{Intraprocedural Dataflow Analysis}


Dataflow analysis is a technique for gathering program information with respect to various program points along program flows. 
A client dataflow analysis can usually be formulated as an instance of the monotone dataflow analysis framework \cite{Kildall:1973-dataflowlattice,Kam:1977-monotoneframework}, which consists of 
the analysis domain including operations to copy and merge domain elements, and the transfer functions over domain elements with respect to each type of statement in the control flow graph (CFG). 
An iterative worklist algorithm then takes as input an instance of the monotone framework, performs the transfer function for each program statement iteratively along the CFG, and computes a fixed point as the analysis result \cite{Kam:1976-iterative}.
Algorithm \ref{a:worklist} shows the worklist algorithm for forward analysis in detail. 

For each statement $k$ in the CFG, two elements $\mathcal{IN}_k$ and $\mathcal{OUT}_k$  represent the incoming and outgoing dataflow facts, respectively. 
At each merging point of CFG in which case a node $k$ has multiple predecessors $p \in preds(k)$, the incoming dataflow fact $\mathcal{IN}_k$ of node $k$ is the combination of all the outgoing facts $\mathcal{OUT}_p$ (shown as Line \ref{a:worklist-merge}) where $\otimes$ indicates the merge operator specified by users which can be meet (for must-analysis) or join (for may-analysis). 
A transfer function for statement $k$  then takes as input $\mathcal{IN}_k$ and returns the new outgoing fact, as shown by Line \ref{a:worklist-transfer}. 
The worklist algorithm is conducted along the CFG to update the dataflow elements $\mathcal{IN}_k$ and $\mathcal{OUT}_k$ for each statement in an iterative manner until a fixed point is reached, meaning that all the dataflow facts are unchanged anymore \cite{Kam:1976-iterative}. 

\begin{algorithm}[htb!]
	\SetKwInOut{Input}{input}\SetKwInOut{Output}{output}
	\caption{Worklist Algorithm for Forward Analysis\label{alg:worklist}}
	\label{a:worklist}
	\DontPrintSemicolon
	\small
	$\mathcal{W} \leftarrow$ \{all the entry statements of the CFG\} \;
	\Repeat{$\mathcal{W} \equiv \emptyset$}
	{
		remove $k$ from $\mathcal{W}$ \;
		$\mathcal{IN}_k \leftarrow \otimes_{p \in \textit{preds}~(k)}~\mathcal{OUT}_p$  \textcolor{olive}{/*merge function*/} \label{a:worklist-merge}\;
		$\textit{Temp} \leftarrow (\mathcal{IN}_k \setminus KILL_k) \cup GEN_k$  \textcolor{olive}{/*transfer function*/} \label{a:worklist-transfer}\;
		\If{$\textit{Temp} \neq \mathcal{OUT}_k$}
		{
			$\mathcal{OUT}_k \leftarrow \textit{Temp}$ \;
			$\mathcal{W} \leftarrow \mathcal{W} \cup \textit{succs}~(k)$ \;
		}
	}
\end{algorithm}

\subsection{Interprocedural Dataflow Analysis}
Interprocedural dataflow analysis takes into account the propagation of dataflow values across multiple procedures. 
Context-sensitive interprocedural analysis distinguishes the distinct calls of a procedure to eliminate the invalid paths, thus achieving high precision. 
Generally, there exist two dominant approaches to context-sensitive interprocedural analysis, namely the \emph{summary-based (or functional) approach} and the \emph{cloning-based approach} \cite{Sharir:two-inter-dataflow}. 

The summary-based approach commonly constructs a summary (transfer) function for each procedure. 
At each call site where the procedure is invoked, the analysis computes the effects of the procedure by directly applying the summary function to the specific inputs at the call site. 
As such, the re-analysis of the procedure body is avoided while enabling context sensitivity.
However, it is not possible to construct such (symbolic) summary functions in general. 
Take the pointer analysis as an example, we can hardly establish a succinct summarization for each procedure since the effects of a procedure are heavily dependent of the alias relations of the inputs at each call site. The evaluation of a summary function on a particular input may not be cheaper than reanalyzing the whole procedure \cite{Wilson:1995-ptf-pldi}.  
Another option is the explicit representation, a.k.a. tabulation method or partial transfer functions \cite{Wilson:1995-ptf-pldi,Murphy:1999-ptf-pepm}. Given a finite lattice, it enumerates the summary function as input-output dataflow value pairs for each procedure. The output value of a summary function can be directly exploited when the identical input value is encountered again for the same procedure. However, as a large number of states need to be maintained, this approach usually suffers from huge space consumption. 


The alternative of achieving context-sensitivity is a cloning-based approach, where a separate clone of the procedure body is created at each callsite \cite{Emami:1994-clone-pldi,Whaley:2004-clone-bdd}. As such, each procedure is re-analyzed under each calling context, preventing the analysis from propagating dataflow values along invalid paths. 
%
In this work, we adopt the cloning-based approach to achieve context-sensitivity. 
The basic analysis logic of interprocedural analysis is the same as that of intraprocedural analysis shown as Algorithm \ref{alg:worklist}, except that the CFG becomes the interprocedural CFG. 
More specifically, to construct the interprocedural CFG, the CFG for each function is firstly generated. Based on a pre-computed call graph, the CFG for each function is cloned and incorporated into that of each of its callers by creating assignment edges to connect vertices representing formal parameters and actual arguments. 
In order to achieve the sweet spot between scalability and precision, we can actually perform cloning only at certain levels, which is theoretically equivalent to the $k$-CFA call string approach  \cite{Shivers:1991:CAH-cmu}. 


\subsection{Vertex-Centric Graph Processing\label{subsec:back-vertex-centric}}
With the inception of Pregel system \cite{Malewicz:2010-pregel}, vertex-centric graph processing becomes a hotspot in the large-scale graph processing community \cite{Khan2016VC}. 
Following Pregel, various  algorithmic techniques and systems were proposed, such as asynchronous model (GraphLab \cite{low2012distributed}), in-memory data parallel model (GraphX \cite{Gonzalez:2014-graphx}). People are able to achieve efficient, scalable, and fault-tolerant graph computing on a large cluster of computers by leveraging these systems. 


\begin{algorithm}[htb!]
	\SetKwInOut{Input}{input}\SetKwInOut{Output}{output}
	\caption{\small Synchronous Vertex-centric Graph Processing \label{alg:vertex-centric}}
	\label{a:vertex-centric}
	\DontPrintSemicolon
	\small
	
	\KwData{$\mathcal{A}$: the set of active vertices during processing}
	
	\BlankLine
	\Repeat{$\mathcal{A} \equiv \emptyset$}
	{
            \ForPar(\textcolor{brown}{/*done by system*/}){each vertex $k \in \mathcal{A} $}{\label{algo2:inter-nodes}
                    Remove $k$ from $\mathcal A$ \label{a-pregel:active-remove} \textcolor{brown}{/*done by system*/}\;
                    \textcolor{brown}{/*perform user-specified logic for each vertex, in particular including Gather, Apply and Scatter*/} \;
                    
                    $\mathcal{M}_k \leftarrow \textcolor{blue}{\textbf{Gather}}(k)$ \label{a-pregel:gather}  \textcolor{olive}{/*gather messages or information from neighbors*/} \;
                    
                    $\mathcal{D}_k \leftarrow \textcolor{blue}{\textbf{Apply}}(\mathcal{M}_k, k)$ \label{a-pregel:apply} \textcolor{olive}{/*update value of k based on gathered information*/} \;
                    
                    $\langle\mathcal{M}, \mathcal{A'} \rangle \leftarrow \textcolor{blue}{\textbf{Scatter}}(\mathcal{D}_k, k)$  \label{a-pregel:scatter} \textcolor{olive}{/*activate new vertices and/or send out messages*/} \;
            }
            \BlankLine
            \textcolor{brown}{/*synchronize before next superstep*/}\;
            \textsc{Synchronize}() \label{a-pregel:sync-message} \textcolor{brown}{/*done by system*/}\;
            $\mathcal{A} \leftarrow \mathcal{A'}$ \label{a-pregel:sync-active} \textcolor{brown}{/*done by system*/}\;
	}
	
\end{algorithm}

Algorithm \ref{a:vertex-centric} gives the pseudo-code of a synchronous vertex-centric processing algorithm.
Given an initialized set of active vertices, it conducts an iterative computation where each iteration is termed as a superstep. At each superstep, all the active vertices in $\mathcal{A}$ are processed in a distributed and parallel way across the entire cluster. Over each active vertex $k$, Gather-Apply-Scatter (a.k.a., GAS model) is performed \cite{Gonzalez:2012-powergraph}.   
At first, the messages or information from its neighbors are gathered (Line \ref{a-pregel:gather}).
At the Apply phase, it updates its associated value $\mathcal{D}_k$ according to its current value and the information gathered (Line \ref{a-pregel:apply}). Based on the newly computed value, it updates the active vertices accordingly, and/or sends necessary messages to its neighbors (Line \ref{a-pregel:scatter}) at the Scatter phase. 
Before the next superstep,  all the messages generated at the current superstep and active vertices are synchronized (Lines \ref{a-pregel:sync-message}-\ref{a-pregel:sync-active}). 
The whole computation terminates until no active vertex is generated.

Note that vertex-centric graph processing \cite{Malewicz:2010-pregel,Gonzalez:2012-powergraph,mccune2015thinking} is a programming model for implementing graph processing applications. 
Users write graph algorithms from the perspective of vertices. They only need to specify the code executed at each vertex, particularly Gather-Apply-Scatter functions (Lines \ref{a-pregel:gather}-\ref{a-pregel:scatter}). 
\textcolor{black}{The underlying graph processing system is responsible for dividing the input large-scale graph into multiple partitions, loading partitions into different cluster nodes, launching multiple threads/processes to execute user-defined code simultaneously, performing necessary synchronizations, optimizing communication among nodes, maintaining replicas to ensure fault tolerance,} etc.

In this work, we take inspiration from vertex-centric graph processing, design and implement a distributed framework \tool tailored to interprocedural dataflow analysis of large-scale code. Similar to the existing general-purpose graph systems, \tool provides user-friendly APIs (\eg, merge and transfer functions) based on which users can readily implement their own client analyses without worrying about scalability. The intrinsic system support under \tool  ensures the distributed capability in lifting the sophisticated analysis to large-scale programs.

\section{Distributed Whole-program Dataflow Analysis}
\label{sec:design}

Inspired by large-scale graph processing, we revisit the classic worklist algorithm of dataflow analysis (Algorithm \ref{alg:worklist}) from the perspective of vertex-centric computation model (Algorithm \ref{alg:vertex-centric}), and accordingly present our first distributed worklist algorithm for whole-program dataflow analysis, \ie, Algorithm~\ref{a:naive-algo} in \cref{subsec:design-naive}.
This algorithm faithfully follows the classic worklist algorithm, and thus it is easy to understand; however, its scalability is also limited under the distributed setting. As a result, in \cref{subsec:design-opt}, we further propose an optimized algorithm that achieves better performance than Algorithm~\ref{a:naive-algo}  as demonstrated in \cref{sec:evaluation}.

\begin{algorithm}[htb!]
	\caption{\small Distributed Worklist Algorithm}
	\label{a:naive-algo}
	\DontPrintSemicolon
	\small
	\KwData{$\mathcal{W}$: the list of all active vertices during analysis;  
	$\mathcal{DS}_k: \{\mathcal{OUT}_p \mid p \in preds(k) \}$ a set containing all the dataflow facts of $k$'s predecessors}
	
	\BlankLine
	\BlankLine

	$\mathcal{W} \leftarrow $ \{all the entry vertices in CFG\} \label{a-naive:worklist-init} \;
	\Repeat{$\mathcal{W} \equiv \emptyset$}
	{
            \ForPar(\textcolor{white}{/*done by */}){each CFG vertex $k \in \mathcal W$}{\label{a-naive:forpar} 
                    Remove $k$ from $\mathcal W$ \label{a-naive:active-remove} \textcolor{white}{/*done by system*/}\;
                    $\mathcal{DS}_k \leftarrow \textsc{GatherAll}(k)$ \label{a-naive:gather} \textcolor{olive}{/*gather all the predecessors' dataflow facts*/}\;
                    $\mathcal {IN}'_k \gets \textbf{Merge}(\mathcal{DS}_k )$ \label{a-naive:combine} \textcolor{olive}{/*merge*/} \;
                    $\mathcal {OUT}_k^\prime \leftarrow \textbf{Transfer}(\mathcal{IN}'_k, k)$ \label{a-naive:transfer} \textcolor{olive}{/*transfer*/} \;
                    \If(\textcolor{olive}{/*propagate*/}){\textbf{Propagate}($\mathcal{OUT}_k$,$\mathcal{OUT}_k^\prime $)}{\label{a-naive:propagate}
                        $\mathcal{OUT}_k \leftarrow \mathcal{OUT}_k^\prime$ \label{a-naive:fact-update}\;
                        $\mathcal{W'} \leftarrow \mathcal{W'} \cup \textit{succs}~(k)$ \label{a-naive:active-update}\;
                    }
            }
            
            \BlankLine
            \textsc{Synchronize}() \label{a-naive:sync-message} \textcolor{white}{/*done by system*/}\;
            $\mathcal{W} \leftarrow \mathcal{W'}$ \label{a-naive:sync-active} \textcolor{white}{/*done by system*/}\;
	}
\end{algorithm}


\subsection{Distributed Vertex-centric Worklist Algorithm \label{subsec:design-naive}}
By directly instantiating \text{Gather-Apply-Scatter} interface and other respective data structures in Algorithm \ref{alg:vertex-centric}, we devise the first distributed worklist algorithm for dataflow analysis, which is listed as Algorithm \ref{a:naive-algo}. 
It takes as input a large interprocedural control flow graph (CFG) or an arbitrary sparse representation \cite{Hardekopf:2011-fpa-cgo,Choi:1991-acs-popl,Ramalingam:2002-ser-tcs}.
At the beginning, all the entry vertices in the input CFG are added to $\mathcal W$ as the initial active vertices (Line~\ref{a-naive:worklist-init}).
During each superstep, the underlying system launches a large number of threads/processes to handle the computation on each vertex in parallel (Line~\ref{a-naive:forpar}).
On each vertex $k$, all the dataflow facts from $k$'s predecessors are firstly gathered (Line \ref{a-naive:gather}). 
This can be implemented by directly invoking the existing APIs provided by pull-based graph systems (\eg, PowerGraph \cite{Gonzalez:2012-powergraph}) or designing a pulling mechanism on top of push-based systems (such as Giraph \cite{giraph-vldb}). 
Next a \emph{Merge} function takes all the dataflow facts gathered from predecessors (\ie, $\mathcal{DS}_k$) as input, and produces the incoming dataflow fact $\mathcal{IN}'_k$ (Line \ref{a-naive:combine}). A transfer function is then performed to generate the new outgoing dataflow fact $\mathcal{OUT}'_k$ (Line \ref{a-naive:transfer}).
After that, we check if the updated dataflow fact $\mathcal{OUT}'_k$ is different from that (\ie, $\mathcal{OUT}_k$) at previous superstep. If so, the propagation is employed to update the dataflow fact as the newly computed value (Line \ref{a-naive:fact-update}). Simultaneously, all of $k$'s successors are activated and put into active list $\mathcal{W}'$ for the next superstep (Line \ref{a-naive:active-update}).

\begin{figure}[htb!]
\centering
\hspace{-1em}
\subfloat[initial state]{         
\label{fig:naive-init}
\begin{minipage}[c]{.3\linewidth}
\centering
\includegraphics[width=0.7\textwidth]{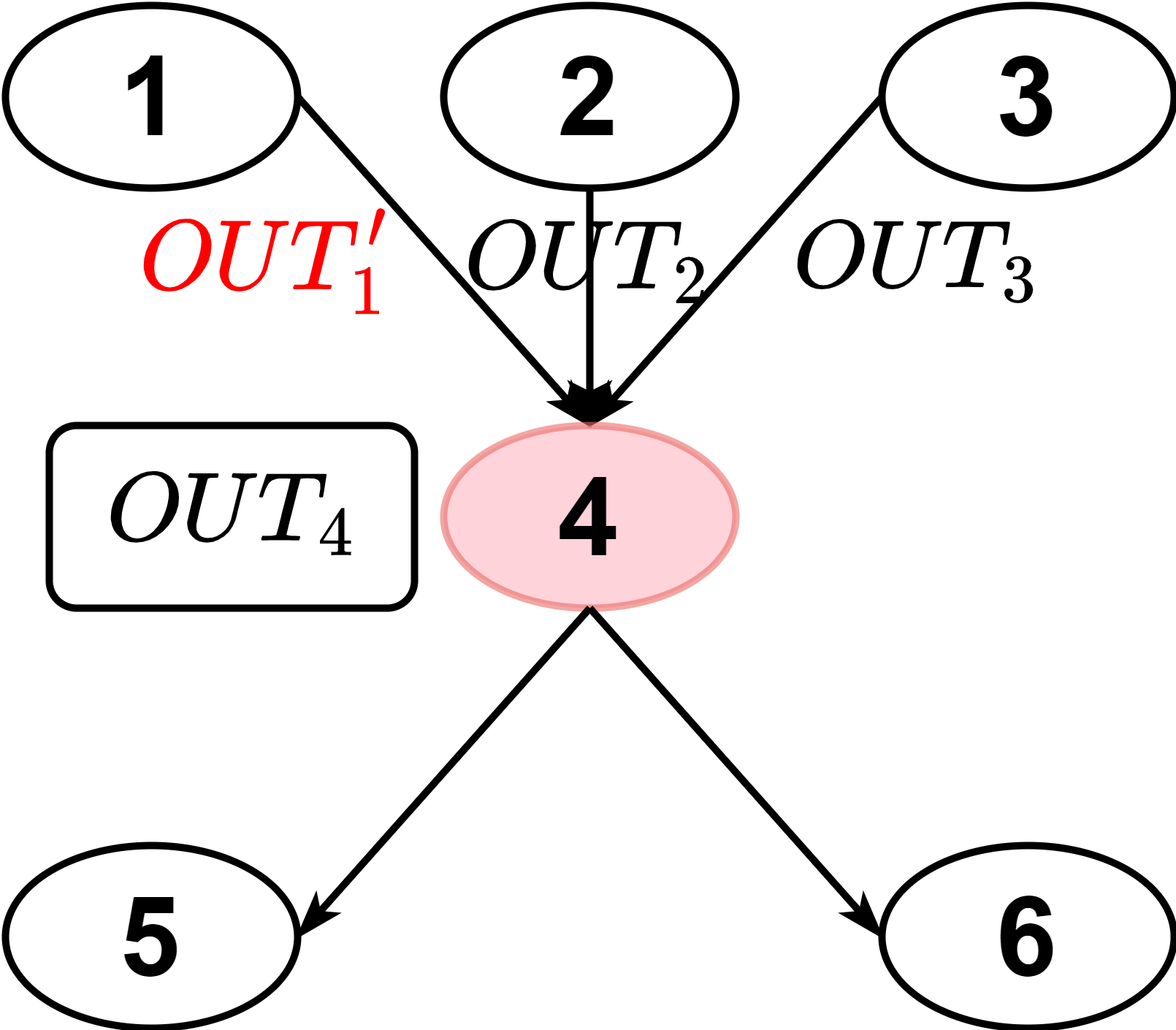}
\end{minipage}
}
\subfloat[gather and merge]{         
\label{fig:naive-gather}
\begin{minipage}[c]{.3\linewidth}
\centering
\includegraphics[width=0.7\textwidth]{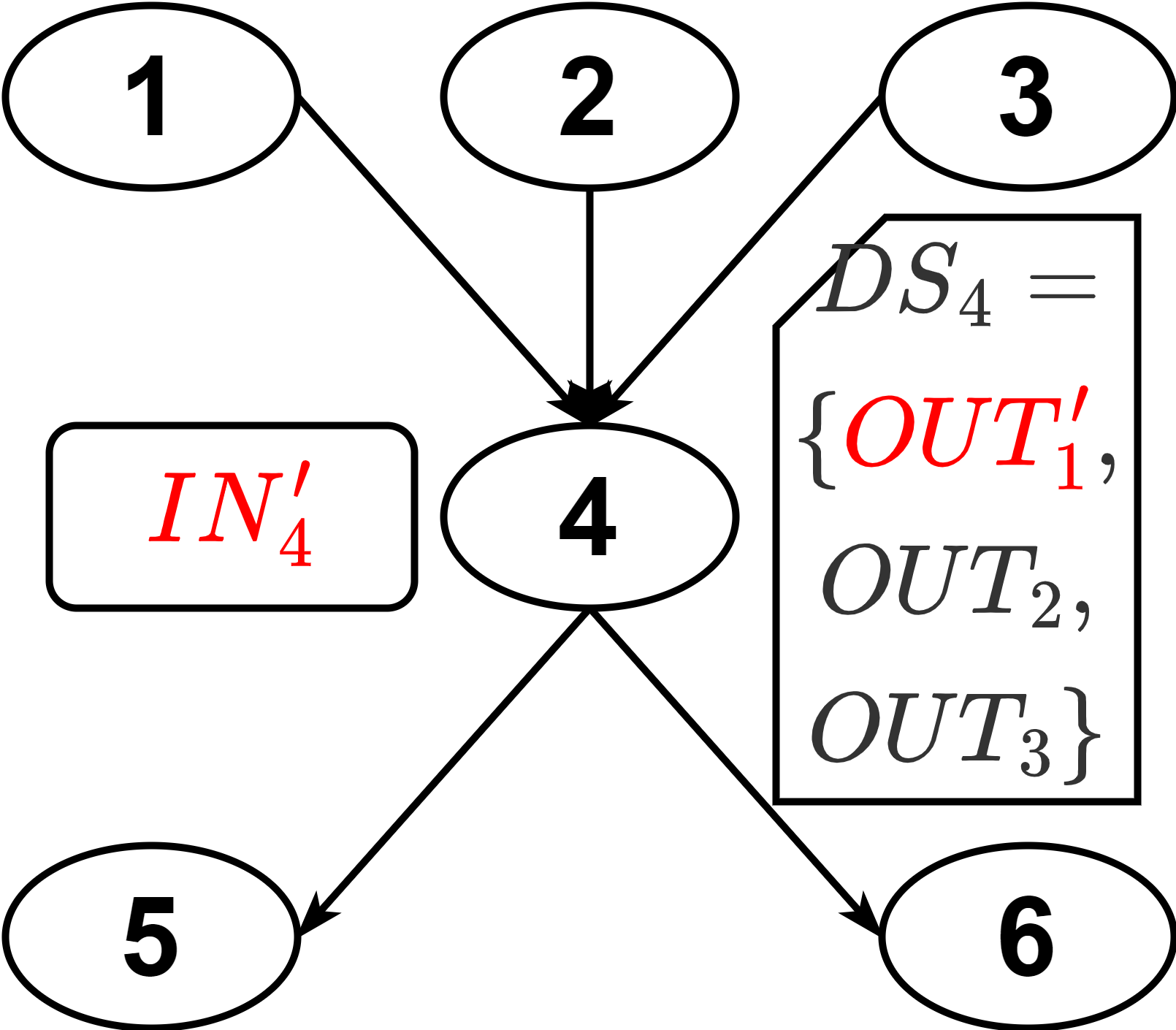}
\end{minipage}
}
\subfloat[transfer and propagate]{         
\label{fig:naive-compute}
\begin{minipage}[c]{.3\linewidth}
\centering
\includegraphics[width=0.7\textwidth]{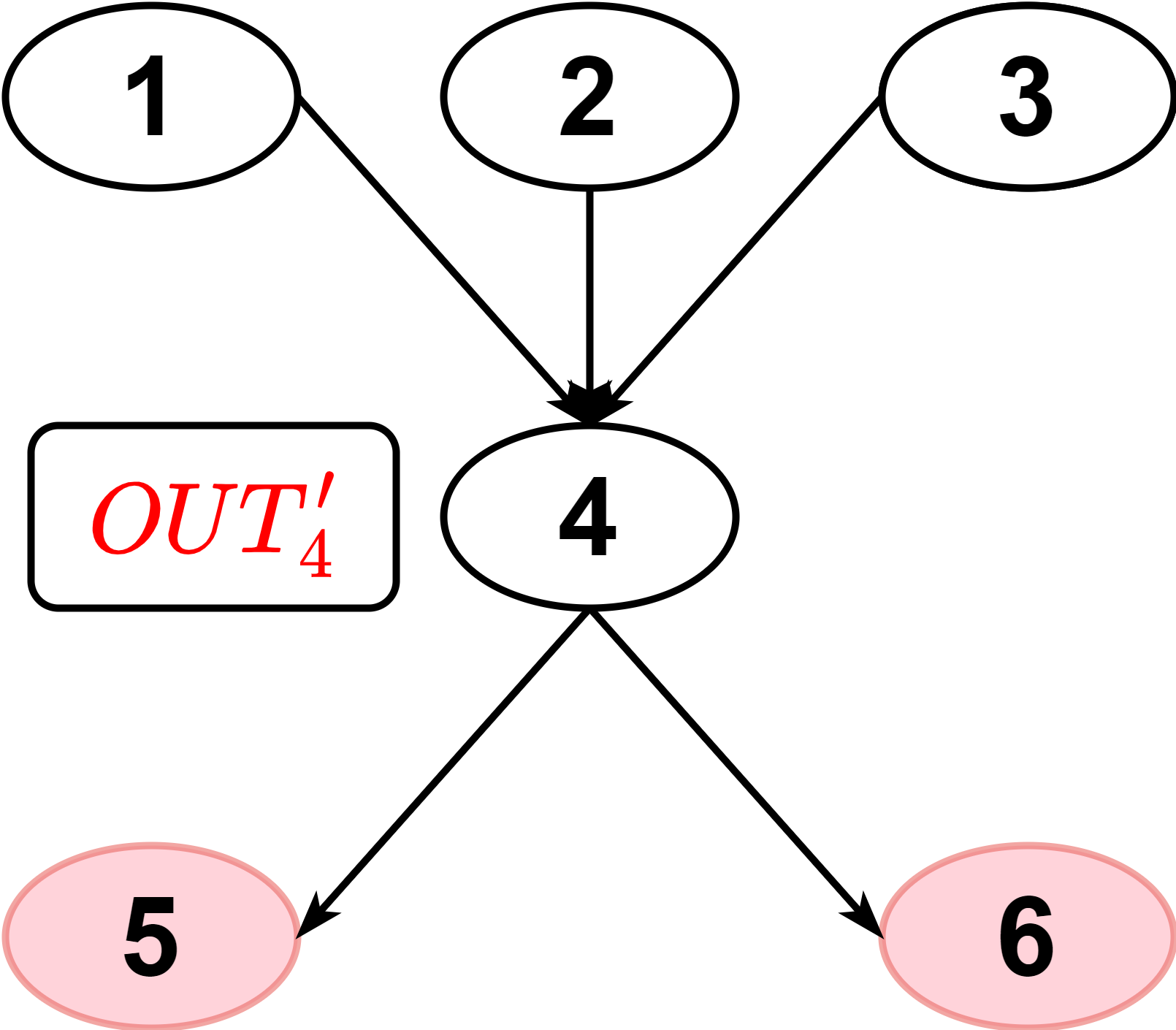}
\end{minipage}
}
\caption{One superstep computation at vertex 4 in Algorithm \ref{a:naive-algo}.}
\label{fig:naive-proc}
\end{figure}

\MyPara{Example.} 
Figure \ref{fig:naive-proc} illustrates the computation procedure at vertex $4$ in the above algorithm, where the vertices with yellow background are active. Vertices $1,2,3$ are the predecessors of $4$, and $5,6$ are its successors. 
Suppose that at the beginning of a certain superstep, the active vertex $4$ has its outgoing dataflow fact $\mathcal{OUT}_4$.
Predecessor $1$ has the newly updated fact $\mathcal{OUT}'_1$, while predecessors $2$ and $3$ hold the old dataflow facts $\mathcal{OUT}_2$ and $ \mathcal{OUT}_3$, respectively (shown as Figure \ref{fig:naive-init}). Firstly, all the predecessors' dataflow facts are gathered as $\mathcal{DS}_4 = \{\mathcal{OUT}'_1, \mathcal{OUT}_2, \mathcal{OUT}_3\}$ and $\mathcal{IN}'_4$ is generated by merging $\mathcal{DS}_4$ shown in Figure \ref{fig:naive-gather}. $\mathcal{OUT}_4^\prime$ is computed by performing transfer function on $\mathcal{IN}'_4$. Assuming that $\mathcal{OUT}_4^\prime$ is different from $\mathcal{OUT}_4$, propagation is employed so that all the successors $5$ and $6$ are marked as active shown as Figure \ref{fig:naive-compute}.

Despite that the above algorithm succeeds in leveraging large-scale distributed computing resources to accelerate dataflow analysis, it may still suffer from poor scalability especially when analyzing large-scale programs such as the Linux kernel or Firefox (elaborated shortly in \cref{sec:evaluation}). 
As shown in Algorithm \ref{a:naive-algo}, each vertex has to collect a full set of dataflow facts associated with all its predecessors (\ie, $\mathcal{DS}_k$) for computation.  
In the worst case, the dataflow fact of each vertex would be made multiple copies each of which is sent to one of its successors. 
As a result, the total number of dataflow facts held in memory and passed across networks grows exponentially with the size of interprocedural control flow graph under analysis. 
This number could be super large in practice especially when performing context-sensitive analysis over large-scale programs. 
Passing/gathering a huge number of expensive dataflow facts not only exhausts the precious memory of a cluster quickly but also increases the burden of network communications, leading to poor scalability. 
We implemented Algorithm \ref{a:naive-algo} as a prototype named \toolnaive and conducted the empirical evaluation of it. The experimental results discussed shortly in \cref{sec:evaluation} show that \toolnaive works well for medium-size programs, but quickly runs out of memory on a 500-worker cluster when analyzing large-scale programs, such as the Linux kernel or Firefox. 
In the following (\cref{subsec:design-opt}), we propose an optimized algorithm which addresses the aforementioned limitations by significantly pruning away the data gathered, thus achieving better scalability and performance.


\subsection{Optimized Distributed Worklist Algorithm \label{subsec:design-opt}}

\begin{algorithm}[htb!]
	\caption{\small Optimized Distributed Worklist Algorithm}
	\label{a:opt-algo}
	\DontPrintSemicolon
	\small
	\KwData{$\mathcal{W}$: the list of all active vertices during analysis;  
	$\mathcal{M}_k:\{\mathcal{OUT'}_p \mid \text{p is a predecessor of k} \}$ a set containing the dataflow facts of $k$'s predecessors which are updated at previous superstep 
	}
	
	\BlankLine
	\BlankLine

	$\mathcal{W} \leftarrow $ \{all the entry vertices in CFG\} \label{a-opt:worklist-init} \;
	\Repeat{$\mathcal{W} \equiv \emptyset$}
	{
            \ForPar{each CFG vertex $k \in \mathcal W$}{\label{a-opt:forpar} 
                    Remove $k$ from $\mathcal W$ \label{a-opt:active-remove} \;
                    $\mathcal{M}_k \leftarrow \textsc{GatherMessages}(k)$ \label{a-opt:gather} \textcolor{olive}{/*gather dataflow facts of the updated predecessors*/}\;
                    $\mathcal {IN}'_k \gets \textbf{Merge}(\mathcal{M}_k,\mathcal {IN}_k)$ \label{a-opt:combine} \textcolor{olive}{/*merge*/} \;
                    $\mathcal {OUT}_k^\prime \leftarrow \textbf{Transfer}(\mathcal{IN}'_k, k)$ \label{a-opt:transfer} \textcolor{olive}{/*transfer*/} \;
                    \If(\textcolor{olive}{/*propagate*/}){\textbf{Propagate}($\mathcal{OUT}_k$,$\mathcal{OUT}_k^\prime $)}{\label{a-opt:propagate}
                        $\mathcal{OUT}_k \leftarrow \mathcal{OUT}_k^\prime$ \label{a-opt:fact-update}\;
                        \ForEach{successor $d$ of $k$ \label{a-opt:propagate-for}}{
                	        $\textsc{SendMessages}(d, \mathcal {OUT}'_k) $ \label{a-opt:sendmsg} \textcolor{olive}{/*send the updated dataflow facts to successors*/}\;
                	        $\mathcal{W'} \leftarrow \mathcal{W'} \cup \{d\}$ \label{a-opt:active-update}\;
	                    }
                    }
                    $\mathcal {IN}_k \gets \mathcal {IN}'_k$ \;
            }
            
            \BlankLine
            \textsc{Synchronize}() \label{a-opt:sync-message} \textcolor{white}{/*done by system*/}\;
            $\mathcal{W} \leftarrow \mathcal{W'}$ \label{a-opt:sync-active} \textcolor{white}{/*done by system*/}\;
	}
\end{algorithm}

As discussed earlier, each active vertex requires the dataflow facts associated with all its predecessors to complete the computation in the original worklist algorithm (\ie, Line \ref{a:worklist-merge} of Algorithm \ref{alg:worklist} and Line \ref{a-naive:combine} of Algorithm \ref{a:naive-algo}). 
That is why extensive dataflow facts have to be transferred across the cluster network and then merged locally on each vertex, resulting in poor scalability. 
To tackle the problem, we devise an optimized algorithm which prunes the dataflow facts to be gathered.  
In particular, instead of gathering the full set of dataflow facts from all the predecessors, only the predecessors' dataflow facts that are newly updated at the previous superstep are passed and merged. 
Since a significant portion of dataflow facts are not changed at one superstep, the optimized algorithm can thus prune away many unnecessary and memory-consuming dataflow facts to be gathered, greatly reducing the overall message traffic and computation cycles for merging. 
We will discuss the correctness of such optimization -- it produces the same analysis results as the original algorithm, and give the formal proof shortly in \S \ref{subsec:proof}.


We propose an optimized distributed worklist algorithm shown as Algorithm \ref{a:opt-algo}. 
For each active vertex $k$, only the set of predecessors' dataflow facts which are updated at previous superstep are gathered. This can be achieved via the push-based message passing mechanism.  
Specifically, each vertex $k$ passively receives the messages passed to it (\ie, $\mathcal{M}_k$) from its predecessors at previous superstep (Line \ref{a-opt:gather}).
Each message in fact corresponds to a dataflow fact sent from one of the predecessors which is updated at the previous superstep.
Subsequently, the dataflow facts $\mathcal{OUT}'_p \in \mathcal{M}_k$ are merged with the incoming dataflow fact of $k$ at last superstep (\ie, $\mathcal{IN}_k$) to generate the new incoming dataflow fact (\ie, $\mathcal{IN}'_k$) (Line \ref{a-opt:combine}). 
We then update the dataflow fact accordingly via a transfer function (Line \ref{a-opt:transfer}).
Next we check if the updated dataflow fact $\mathcal{OUT}'_k$ is different from that (\ie, $\mathcal{OUT}_k$) at previous superstep. If so, the propagation is employed to update the dataflow fact as the newly computed value (Line \ref{a-opt:fact-update}). At the same time, $\mathcal{OUT}'_k$ is sent as a message to each of its successors $d$ (Line \ref{a-opt:sendmsg}), while activating $d$ for the next superstep (Line \ref{a-opt:active-update}).

\begin{figure}[htb!]
\centering
\hspace{-1em}
\subfloat[initial state]{         
\label{fig:opt-init}
\begin{minipage}[c]{.3\linewidth}
\centering
\includegraphics[width=0.75\textwidth]{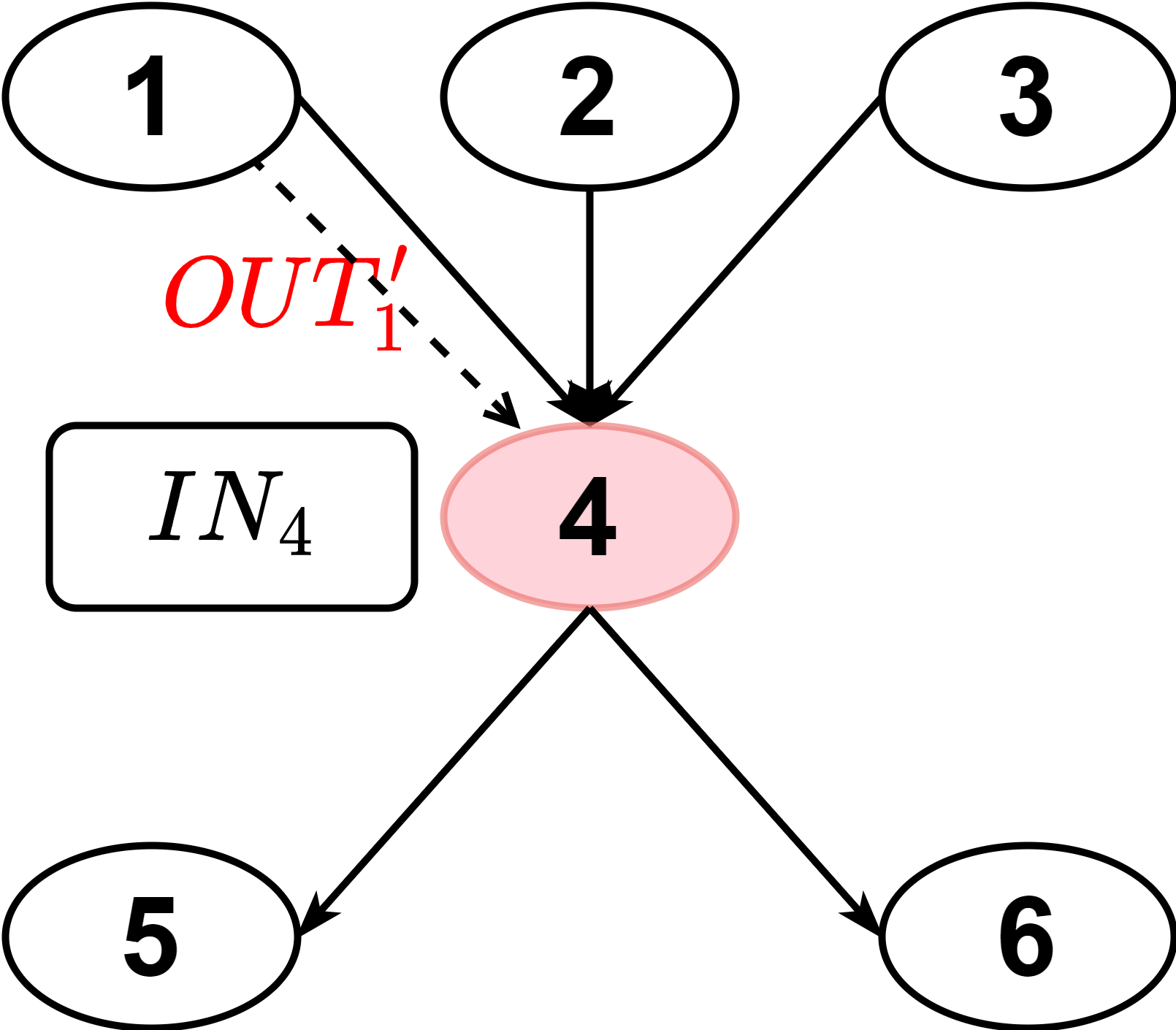}
\end{minipage}
}
\subfloat[gather and merge]{         
\label{fig:opt-gather}
\begin{minipage}[c]{.3\linewidth}
\centering
\includegraphics[width=0.75\textwidth]{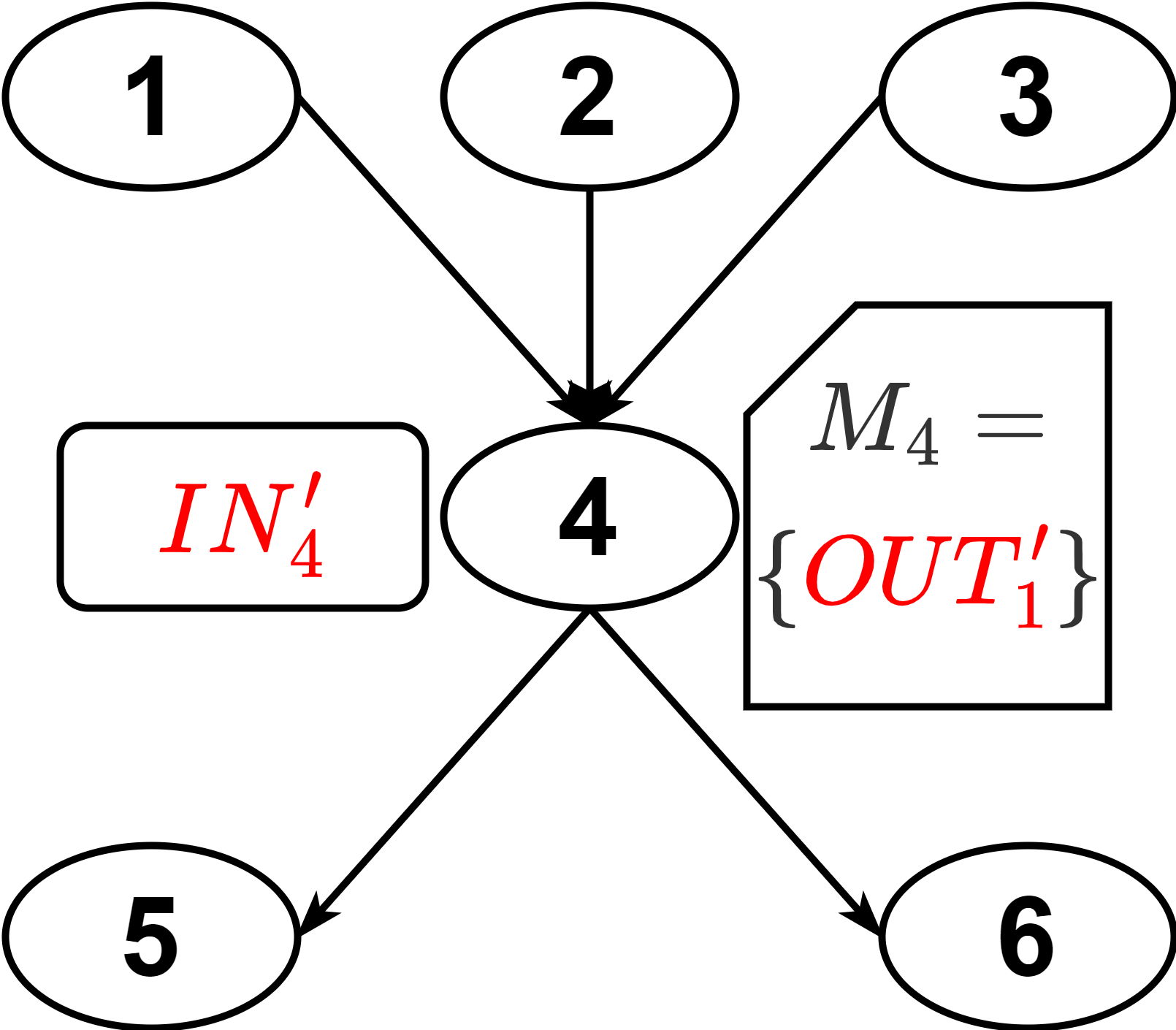}
\end{minipage}
}
\subfloat[transfer and propagate]{         
\label{fig:opt-compute}
\begin{minipage}[c]{.3\linewidth}
\centering
\includegraphics[width=0.75\textwidth]{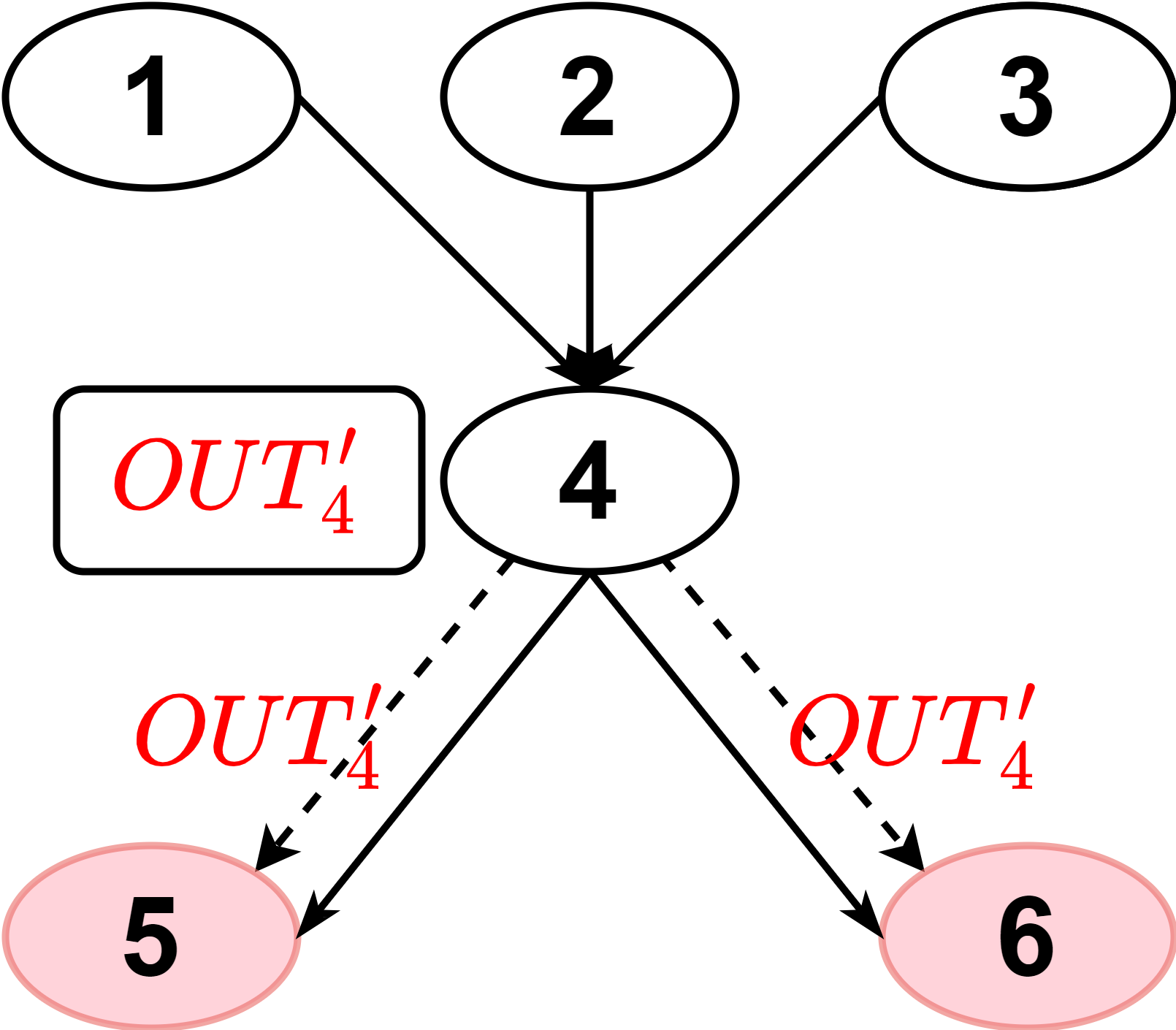}
\end{minipage}
}
\caption{One superstep computation at vertex 4 in Algorithm \ref{a:opt-algo}.}
\label{fig:opt-proc}
\end{figure}

\MyPara{Example.} 
Figure \ref{fig:opt-proc} illustrates the computation procedure at vertex $4$ in Algorithm \ref{a:opt-algo} for the same example as Figure \ref{fig:naive-proc}. 
Suppose that at the beginning of a certain superstep, the active vertex $4$'s predecessor $1$ has the newly updated fact $\mathcal{OUT}'_1$ and sends it as a message to $4$ at the last superstep denoted by the dashed arrows in Figure \ref{fig:opt-init}.
The message is gathered as $\mathcal{M}_4 = \{\mathcal{OUT}'_1\}$ and then $\mathcal{IN}'_4$ is generated by incrementally merging the dataflow facts in $\mathcal{M}_4$ with $\mathcal{IN}_4$ shown as Figure \ref{fig:opt-gather}.
Finally as shown by Figure \ref{fig:opt-compute}, $\mathcal{OUT}_4^\prime$ is computed by performing transfer function on $\mathcal{IN}'_4$. Assuming that $\mathcal{OUT}_4^\prime$ is different from $\mathcal{OUT}_4$, propagation is employed to send the newly updated $\mathcal{OUT}_4^\prime$ to all the successors $5$ and $6$ while marking them as active.

\subsection{Correctness Proof of Optimized Algorithm \label{subsec:proof}}
The underlying rationale of such optimization is that the \emph{merge} operation for the general monotone dataflow analysis satisfies the accumulative property. 
In other words, on each active vertex $k$, merging only the updated dataflow facts of $k$'s predecessors with the old $\mathcal{IN}_k$ of last superstep should produce identical results to that merging the full set of dataflow facts of all its predecessors.    
The following Theorem \ref{theorem:accumulative-property} gives its formal definition.

\begin{theorem}[\textbf{Accumulative Property}]
\label{theorem:accumulative-property}
Given an active vertex $k$, let $preds(k)$ be the set of $k$'s predecessors. Without loss of generality, suppose at the previous superstep, a partial set of $k$'s predecessors \ie, $P'(k) \subseteq preds(k)$ update their outgoing dataflow facts, while the outgoing facts of the remaining \ie, $P(k) = preds(k) - P'(k)$ stay unchanged. $\mathcal{IN}_k$ and $\mathcal{IN}'_k$ indicate the incoming dataflow fact of $k$ at previous and current superstep, respectively. The accumulative property is satisfied if and only if the following equation holds. 
{ \[\mathcal{IN}'_k  \equiv \mathcal{IN}_k {\color{black}{\otimes (}}\otimes_{p \in \textit{P}'(k)}\mathcal{OUT}'_p {\color{black}{)}}\]}
\ie, 
\begin{align}
& {\color{black}{(}} \otimes_{p \in \textit{P}(k)}\mathcal{OUT}_p {\color{black}{) \otimes (}} \otimes_{p \in \textit{P}'(k)}\mathcal{OUT}'_p {\color{black}{)}} \equiv \nonumber\\
              &~~~~~~~~~~~~~~~~~~~~{\color{black}{(}} \otimes_{p \in \textit{preds}(k)}\mathcal{OUT}_p {\color{black}{) \otimes (}} \otimes_{p \in \textit{P}'(k)}\mathcal{OUT}'_p {\color{black}{)}} \nonumber
\end{align}
\end{theorem}


\begin{proof}
Generally, there are two cases of monotone dataflow analysis, namely (1) increasing analysis with the join operator $\sqcup$ and (2) decreasing analysis with the meet operator $\sqcap$.

\textbf{For case (1):} $\otimes = \sqcup$ and for each predecessor $p \in P'(k)$,  $\mathcal{OUT}_p  \leq \mathcal{OUT}'_p$ holds where $\leq$ denotes the partial order relation and $\leq$ is reflexive, anti-symmetric and transitive according to its definition.  

As defined by the $\sqcup$ operator which computes the least upper bound of two elements in the lattice, the following inequality \ref{eq:left} holds.  
\begin{equation}\label{eq:left}
\mathcal{OUT}'_p  \leq \mathcal{OUT}_p \sqcup \mathcal{OUT}'_p
\end{equation}

Given that $\mathcal{OUT}_p  \leq \mathcal{OUT}'_p$ (for increasing analysis) and $\mathcal{OUT}'_p  \leq \mathcal{OUT}'_p$ ($\leq$ is reflexive), the following can be deduced:
\begin{equation}\label{eq:right}
\mathcal{OUT}_p \sqcup \mathcal{OUT}'_p \leq \mathcal{OUT}'_p
\end{equation}

As $\leq$ is anti-symmetric, given inequalities \ref{eq:left} and \ref{eq:right} hold, we can imply the following equation \ref{eq:equal}.
\begin{equation}\label{eq:equal}
\mathcal{OUT}'_p  \equiv \mathcal{OUT}_p \sqcup \mathcal{OUT}'_p
\end{equation}

Therefore, for all $p \in \textit{P}'(k)$, the following equation \ref{eq:pprime} holds. 
\begin{equation}\label{eq:pprime}
\sqcup_{p \in \textit{P}'(k)}\mathcal{OUT}'_p \equiv \sqcup_{p \in \textit{P}'(k)}(\mathcal{OUT}_p \sqcup \mathcal{OUT}'_p)
\end{equation}

Because the $\sqcup$ operator in monotone dataflow analysis is both associative and commutative, we can imply that:
\begin{equation}\label{eq:ie5}
\sqcup_{p \in \textit{P}'(k)}\mathcal{OUT}'_p \equiv {\color{black}{(}} \sqcup_{p \in \textit{P}'(k)}\mathcal{OUT}_p {\color{black}{) ~\sqcup~(}} \sqcup_{p \in \textit{P}'(k)}\mathcal{OUT}'_p {\color{black}{)}}
\end{equation}

By joining $\sqcup_{p \in \textit{P}(k)}\mathcal{OUT}_p$ with both sides of the equation \ref{eq:ie5}, we can get the following: 
\begin{align}\label{eq:ie6}
&\!\!\!\!\!\! {\quad\quad \color{black}{(}} \sqcup_{p \in \textit{P}(k)}\mathcal{OUT}_p {\color{black}{)~\sqcup~(}} \sqcup_{p \in \textit{P}'(k)}\mathcal{OUT}'_p {\color{black}{)}} \equiv \nonumber\\
              &~~~~~{\color{black}{(}}\sqcup_{p \in \textit{P}(k)}\mathcal{OUT}_p {\color{black}{) \sqcup}} (  {\color{black}{(}} \sqcup_{p \in \textit{P}'(k)}\mathcal{OUT}_p  {\color{black}{) \sqcup(}} \sqcup_{p \in \textit{P}'(k)}\mathcal{OUT}'_p {\color{black}{)}} ) 
\end{align}

And further equation \ref{eq:ie7} is deduced since $\sqcup$ is associative.
\begin{align}\label{eq:ie7}
&\!\!\!\!\!\!\!\! {\color{black}{(}} \sqcup_{p \in \textit{P}(k)}\mathcal{OUT}_p  {\color{black}{) ~\sqcup~(}} \sqcup_{p \in \textit{P}'(k)}\mathcal{OUT}'_p  {\color{black}{)}} \equiv \nonumber\\
              &~~~~~~~~~~( {\color{black}{(}} \sqcup_{p \in \textit{P}(k)}\mathcal{OUT}_p  {\color{black}{) ~\sqcup}~(} \sqcup_{p \in \textit{P}'(k)}\mathcal{OUT}_p {\color{black}{)}} ) {\color{black}{~\sqcup~(}} \sqcup_{p \in \textit{P}'(k)}\mathcal{OUT}'_p {\color{black}{)}}
\end{align}

Since equation \ref{eq:ie8} holds,
\begin{equation}\label{eq:ie8}
{\color{black}{(}} \sqcup_{p \in \textit{P}(k)}\mathcal{OUT}_p 
 {\color{black}{) ~\sqcup~(}}  \sqcup_{p \in \textit{P}'(k)}\mathcal{OUT}_p  {\color{black}{)}} \equiv \sqcup_{p \in \textit{preds}(k)}\mathcal{OUT}_p
\end{equation}

The final equation \ref{eq:final} for case (1) is thus proved. 
\begin{align}\label{eq:final}
&\!\!\!\!\!\!\!\!  {\color{black}{(}} \sqcup_{p \in \textit{P}(k)}\mathcal{OUT}_p {\color{black}{) ~\sqcup~(}} \sqcup_{p \in \textit{P}'(k)}\mathcal{OUT}'_p  {\color{black}{)}} \equiv \nonumber\\
              &~~~~~~~~~~~~~~~~~~~~ {\color{black}{(}} \sqcup_{p \in \textit{preds}(k)}\mathcal{OUT}_p 
 {\color{black}{) ~\sqcup~(}} \sqcup_{p \in \textit{P}'(k)}\mathcal{OUT}'_p {\color{black}{)}}
\end{align}

\MyPara{\color{black}{Example.}} 
{\color{black}{We use the vertex 4 in Figure \ref{fig:opt-proc} as an example to demonstrate the proof procedure. 
Assuming that $k = 4$, $preds(4) = \{1, 2, 3\}$.  Given that at previous superstep, predecessor 1 updates its outgoing dataflow fact, thus $P'(4) = \{1\}$ and $P(4) = preds(4) - P'(4) = \{2, 3\}$. The incoming dataflow fact of $4$ at previous and current supersteps are $\mathcal{IN}_4$ and $\mathcal{IN}'_4$, respectively. For case (1), suppose each dataflow fact corresponds to a set. The join operator $\sqcup$ indicates the set union $\cup$. The partial order relation $\leq$ is set inclusion $\subseteq$. Validating the accumulative property specific to this example is to prove the following equation holds: 
\[\mathcal{IN}'_4 \equiv  \mathcal{IN}_4 \cup \mathcal{OUT}'_1\]

Given the join operator $\cup$ and partial order relation $\subseteq$, it is apparent that the equation $\mathcal{OUT}'_1 \equiv \mathcal{OUT}_1 \cup \mathcal{OUT}'_1$ holds according to \ref{eq:left} and \ref{eq:right}. 
\begin{align}
                                    & \mathcal{OUT}'_1 \equiv \mathcal{OUT}_1 \cup \mathcal{OUT}'_1 \nonumber\\   
    \overset{\ref{eq:ie6}} \Longrightarrow~ & (\mathcal{OUT}_2 \cup \mathcal{OUT}_3) \cup \mathcal{OUT}'_1 \equiv \nonumber\\ 
    &~~~~~~~~~ (\mathcal{OUT}_2 \cup \mathcal{OUT}_3)  \cup (\mathcal{OUT}_1 \cup \mathcal{OUT}'_1) \nonumber\\  
    \overset{\ref{eq:ie7}} \Longrightarrow~ & (\mathcal{OUT}_2 \cup \mathcal{OUT}_3) \cup \mathcal{OUT}'_1 \equiv \nonumber\\ 
    &~~~~~~~~~ (\mathcal{OUT}_2 \cup \mathcal{OUT}_3  \cup \mathcal{OUT}_1) \cup \mathcal{OUT}'_1 \nonumber\\    
    \overset{\ref{eq:ie8}} \Longrightarrow~ & \mathcal{IN}'_4 \equiv  \mathcal{IN}_4 \cup \mathcal{OUT}'_1 \nonumber
\end{align}
}}

\textbf{For case (2):} $\otimes = \sqcap$ and for each predecessor $p \in P'(k)$,  $\mathcal{OUT}'_p  \leq \mathcal{OUT}_p$ holds. We can follow the similar proof logic.  

As the meet $\sqcap$ operator calculates the greatest lower bound of elements, the following inequality \ref{eq:left2} holds.  
\begin{equation}\label{eq:left2}
\mathcal{OUT}_p \sqcap \mathcal{OUT}'_p \leq \mathcal{OUT}'_p 
\end{equation}

Given that $\mathcal{OUT}'_p  \leq \mathcal{OUT}_p$ (for decreasing analysis) and $\mathcal{OUT}'_p  \leq \mathcal{OUT}'_p$ ($\leq$ is reflexive), the following can be deduced:
\begin{equation}\label{eq:right2}
\mathcal{OUT}'_p \leq \mathcal{OUT}_p \sqcap \mathcal{OUT}'_p
\end{equation}

As $\leq$ is anti-symmetric, given inequalities \ref{eq:left2} and \ref{eq:right2}, the following equation \ref{eq:equal2} can be concluded.
\begin{equation}\label{eq:equal2}
\mathcal{OUT}'_p  \equiv \mathcal{OUT}_p \sqcap \mathcal{OUT}'_p
\end{equation}

Therefore, we can imply the following equations. 
\begin{align}
            & \sqcap_{p \in \textit{P}'(k)}\mathcal{OUT}'_p \equiv \sqcap_{p \in \textit{P}'(k)}(\mathcal{OUT}_p \sqcap \mathcal{OUT}'_p) \nonumber\\
\Rightarrow~    & \sqcap_{p \in \textit{P}'(k)}\mathcal{OUT}'_p \equiv {\color{black}{(}} 
 \sqcap_{p \in \textit{P}'(k)}\mathcal{OUT}_p  {\color{black}{) ~\sqcap~(}} \sqcap_{p \in \textit{P}'(k)}\mathcal{OUT}'_p  {\color{black}{)}} \nonumber\\
\Rightarrow~  &  {\color{black}{(}} \sqcap_{p \in \textit{P}(k)}\mathcal{OUT}_p  {\color{black}{) ~\sqcap~(}} \sqcap_{p \in \textit{P}'(k)}\mathcal{OUT}'_p  {\color{black}{)}} \equiv \nonumber\\
            &   ~~~~~~~~~ {\color{black}{(}} \sqcap_{p \in \textit{P}(k)}\mathcal{OUT}_p  {\color{black}{) ~\sqcap~}} ( {\color{black}{(}} \sqcap_{p \in \textit{P}'(k)}\mathcal{OUT}_p  {\color{black}{) ~\sqcap~(}} \sqcap_{p \in \textit{P}'(k)}\mathcal{OUT}'_p {\color{black}{)}} ) \nonumber\\
\Rightarrow~    &  {\color{black}{(}} \sqcap_{p \in \textit{P}(k)}\mathcal{OUT}_p  {\color{black}{) ~\sqcap~(}} \sqcap_{p \in \textit{P}'(k)}\mathcal{OUT}'_p  {\color{black}{)}} \equiv \nonumber\\
            &   ~~~~~~~~~ ( {\color{black}{(}} \sqcap_{p \in \textit{P}(k)}\mathcal{OUT}_p  {\color{black}{) ~\sqcap~(}} \sqcap_{p \in \textit{P}'(k)}\mathcal{OUT}_p {\color{black}{)}} )  {\color{black}{~\sqcap~(}} \sqcap_{p \in \textit{P}'(k)}\mathcal{OUT}'_p  {\color{black}{)}} \nonumber\\
\Rightarrow~   &  {\color{black}{(}} \sqcap_{p \in \textit{P}(k)}\mathcal{OUT}_p  {\color{black}{) ~\sqcap~(}} \sqcap_{p \in \textit{P}'(k)}\mathcal{OUT}'_p  {\color{black}{)}} \equiv \nonumber\\
            &   ~~~~~~~~~ {\color{black}{(}} \sqcap_{p \in \textit{preds}(k)}\mathcal{OUT}_p  {\color{black}{) ~\sqcap~(}} \sqcap_{p \in \textit{P}'(k)}\mathcal{OUT}'_p  {\color{black}{)}} \label{eq:final2}
\end{align}

As equations \ref{eq:final} and \ref{eq:final2} hold for each case, we ultimately complete the proof of Theorem \ref{theorem:accumulative-property}.

\end{proof}

\change{
\section{Distributed Incremental Dataflow Analysis}
\label{sec:design-incre}

In the previous section, we present a dedicated distributed worklist algorithm to accelerate whole-program dataflow analysis by leveraging a large amount of computing resources in the cloud.
However, modern real-world programs are updated frequently, especially the popular open-source projects with extensive developer communities, or large-scale industrial production code.
Under this scenario, when a large-scale program is slightly updated, re-performing whole-program dataflow analysis from scratch could be time-consuming and expensive.

To efficiently handle the small batch of updates on real-world programs, we extend the distributed dataflow analysis framework to support incremental analysis. To this end, we develop \toolincre which realizes distributed incremental dataflow analysis on the cloud. 
Basically, \toolincre reuses the existing analysis result and only performs local analysis on the updated parts to achieve efficiency in terms of both memory and time costs.

\subsection{Workflow of Distributed Incremental Dataflow Analysis \label{subsec:overview-incre}}

\begin{figure}[htb!]
\centering
{         
\begin{minipage}[c]{0.7\linewidth}
\centering
\includegraphics[width=1\textwidth]{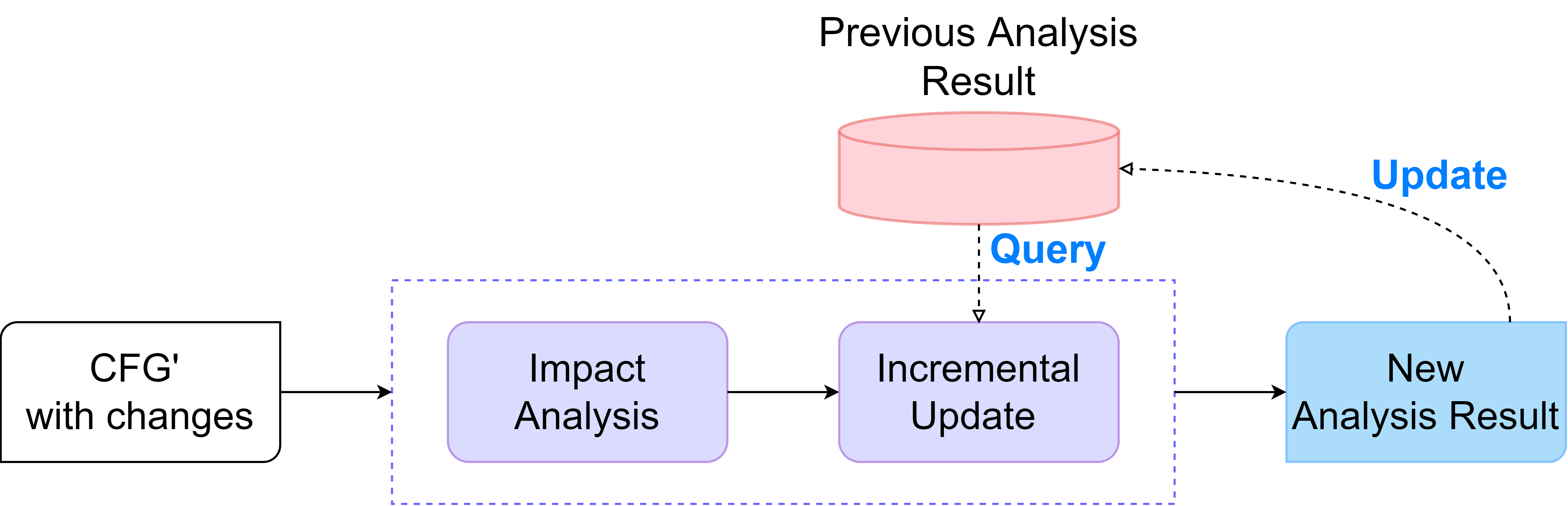}
\end{minipage}
}
\caption{Workflow of Distributed Incremental Dataflow Analysis.}
\label{fig:incre-framework}
\end{figure}


Figure \ref{fig:incre-framework} illustrates the workflow of the distributed incremental dataflow analysis. Given a CFG with a batch of updates, the framework first performs impact analysis to analyze the impact of the updates.
It determines a set of the affected CFG nodes whose associated dataflow facts have to be updated.
The affected nodes and their related edges form a new sub-CFG, where incremental update should be performed. 
Next, the incremental update takes the affected sub-CFG and the previous dataflow analysis result as input. It incrementally update the dataflow facts of the sub-CFG on the basis of previous results in an efficient manner.

Both above analysis processes are implemented as large-scale graph processing in distributed manner. 
Moreover, we exploit a distributed database (such as Redis) to maintain the whole-program analysis result for efficient queries and updates.
Redis\footnote{\url{https://redis.io/}} is a scalable, efficient and reliable distributed database system, which is suitable for our scenario where a number of dataflow facts need to be stored, queried and updated. 
As follows, we discuss our distributed incremental analysis and an optimized version in \cref{subsec:incremental-naive} and \cref{subsec:incremental-opt}, respectively.

\subsection{Distributed Incremental Analysis Algorithm \label{subsec:incremental-naive}} 



\begin{figure}[htb!]
\centering
\subfloat[added edge between existing nodes]{
\label{fig:CFG-add1}
\begin{minipage}[c]{.2\linewidth}
\centering
\includegraphics[width=1\textwidth]{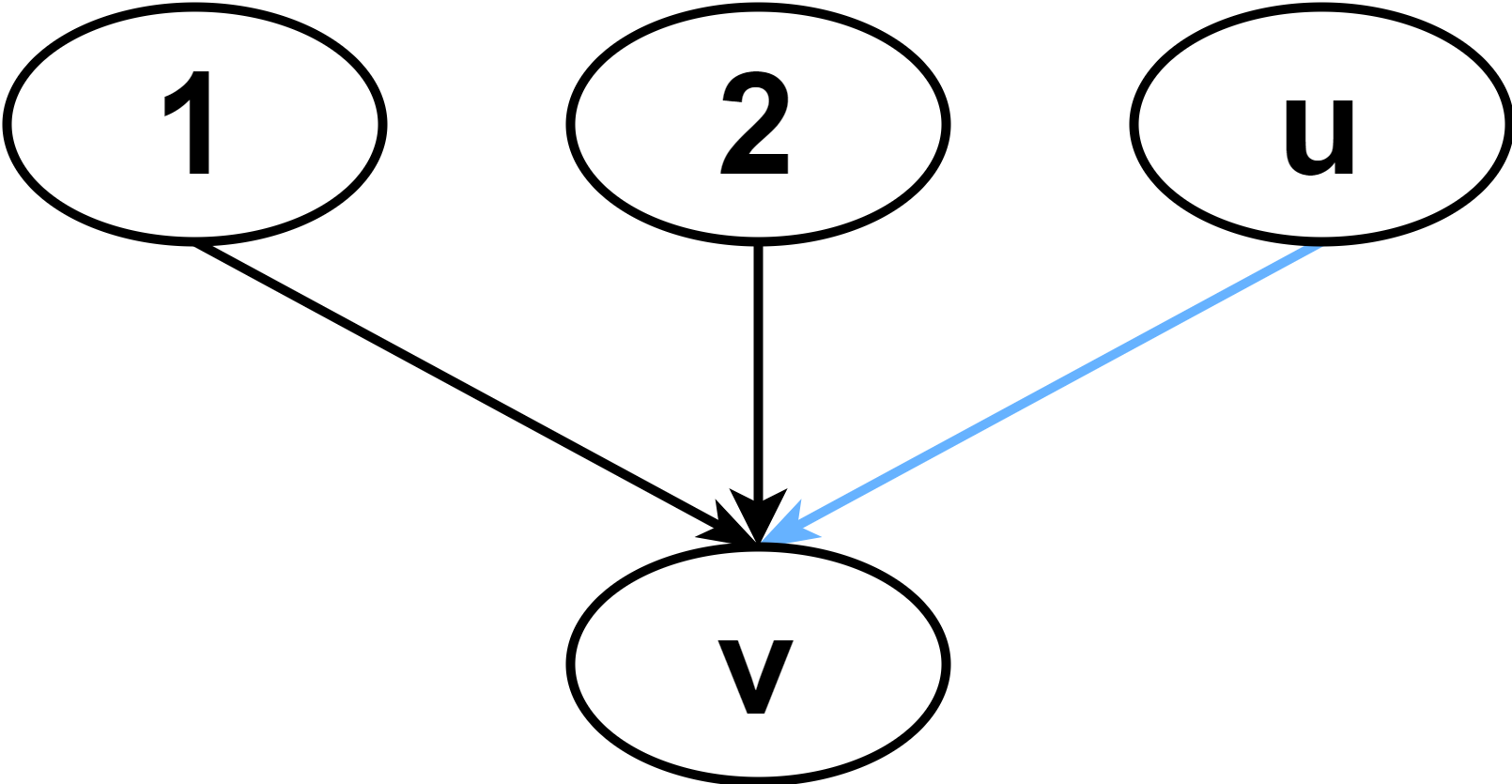}
\end{minipage}
}
\hspace{0.6em}
\subfloat[added source node]{    
\label{fig:CFG-add2}
\begin{minipage}[c]{.2\linewidth}
\centering
\includegraphics[width=1\textwidth]{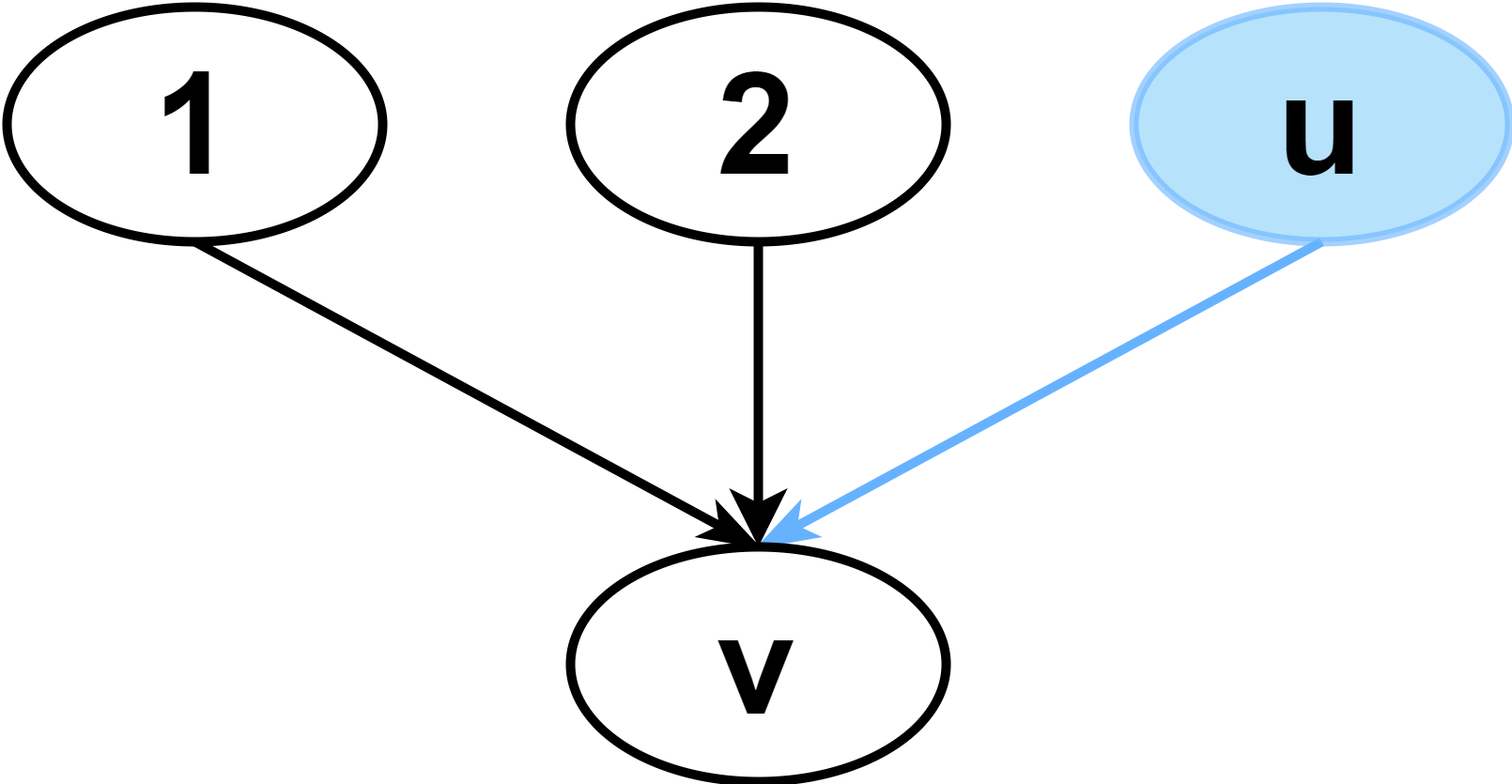}
\end{minipage}
}
\hspace{0.6em}
\subfloat[added destination node]{         
\label{fig:CFG-add3}
\begin{minipage}[c]{.2\linewidth}
\centering
\includegraphics[width=1\textwidth]{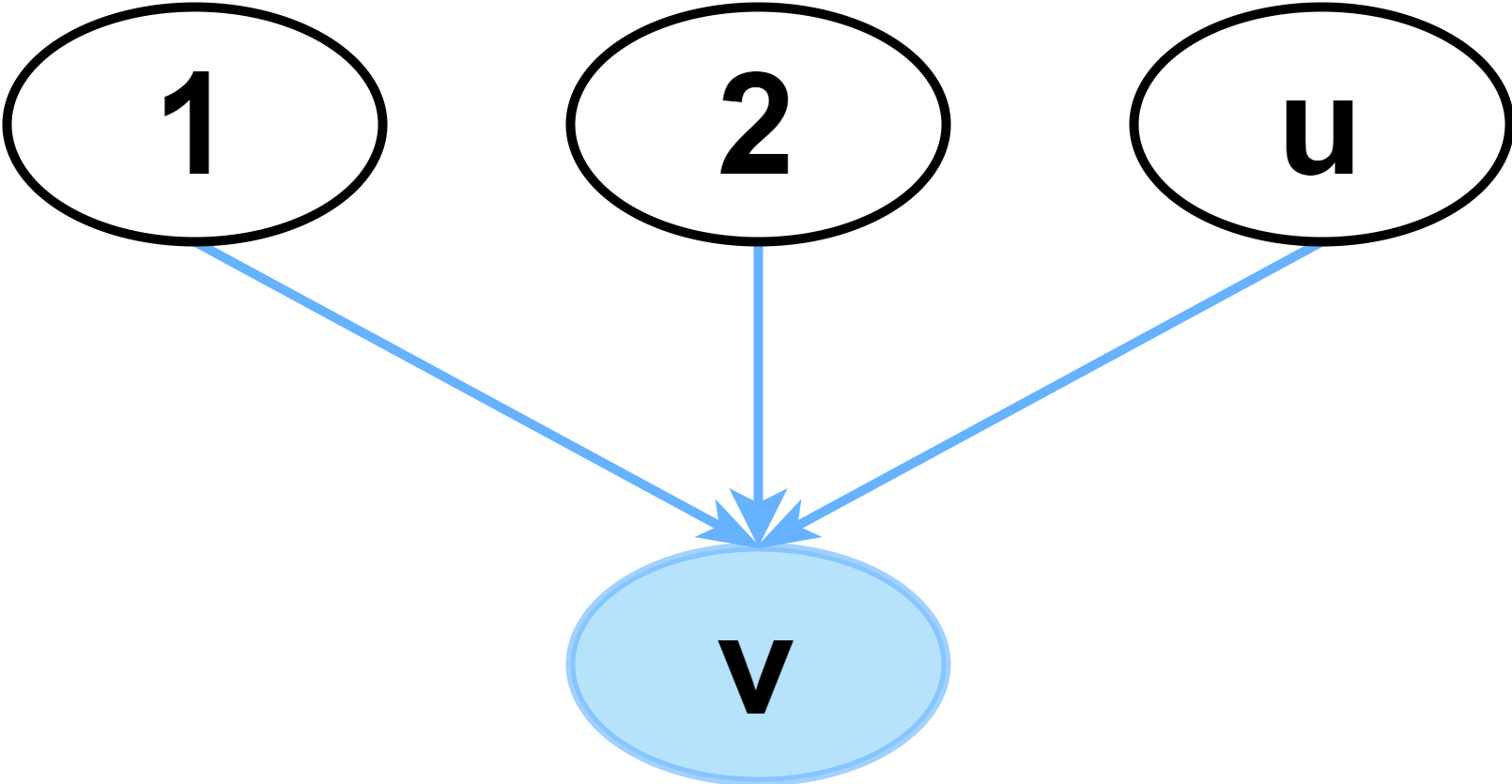}
\end{minipage}
}
\hspace{0.6em}
\subfloat[deleted edge between existing nodes]{    
\label{fig:CFG-del1}
\begin{minipage}[c]{.2\linewidth}
\centering
\includegraphics[width=1\textwidth]{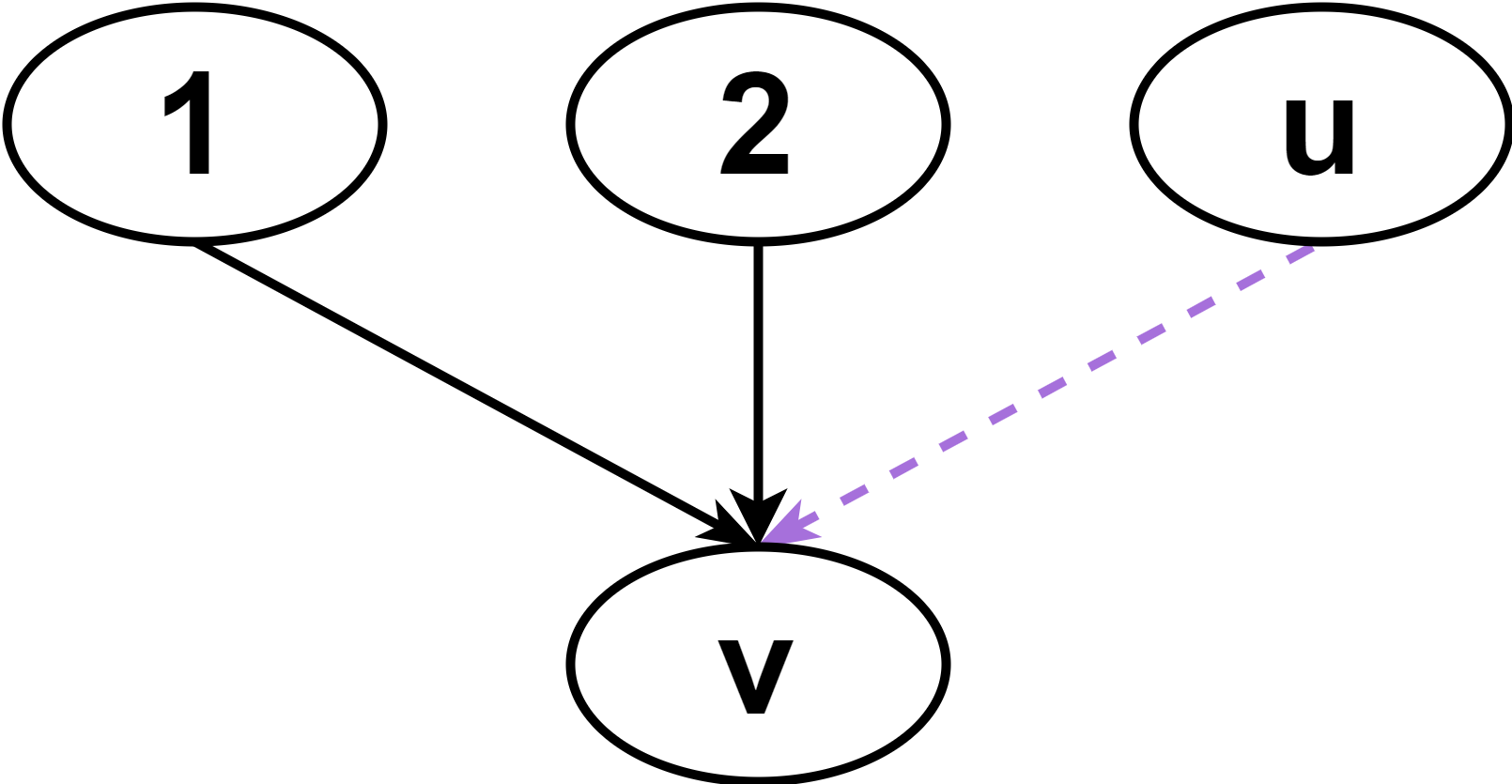}
\end{minipage}
}

\subfloat[deleted source node]{         
\label{fig:CFG-del2}
\begin{minipage}[c]{.2\linewidth}
\centering
\includegraphics[width=1\textwidth]{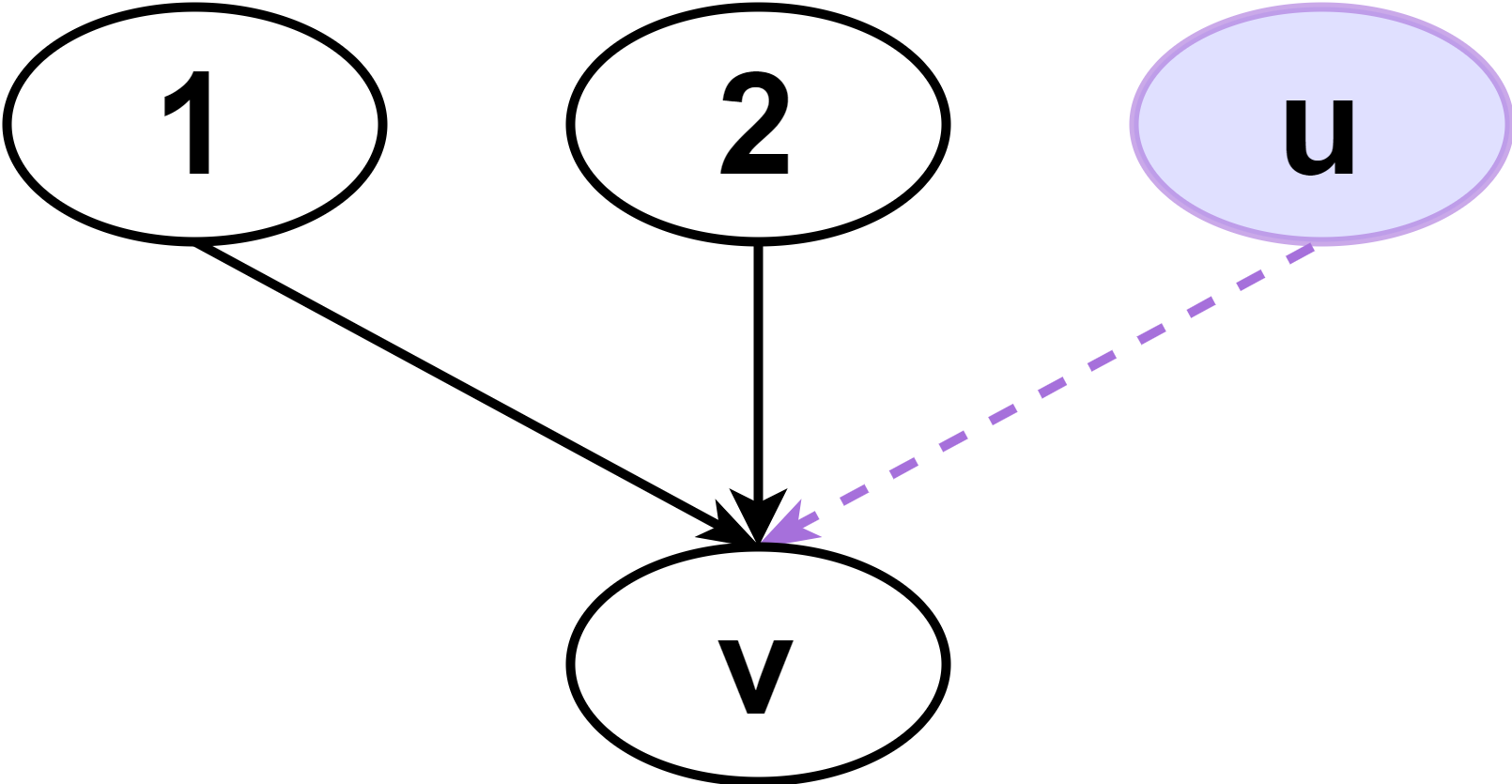}
\end{minipage}
}
\hspace{0.6em}
\subfloat[deleted destination node]{  
\label{fig:CFG-del3}
\begin{minipage}[c]{.2\linewidth}
\centering
\includegraphics[width=1\textwidth]{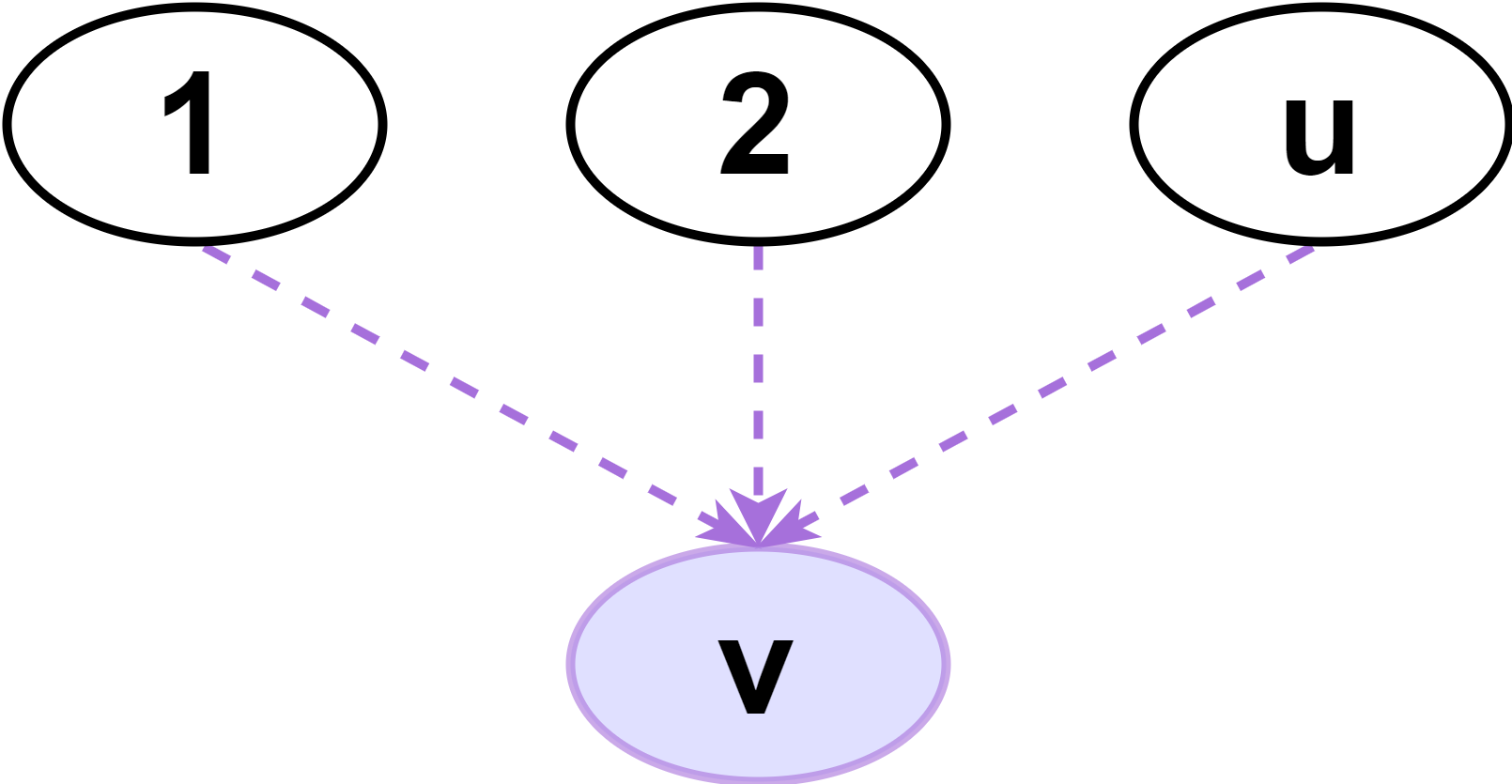}
\end{minipage}
}
\hspace{0.6em}
\subfloat[changed source node]{ 
\label{fig:CFG-change2}
\begin{minipage}[c]{.2\linewidth}
\centering
\includegraphics[width=1\textwidth]{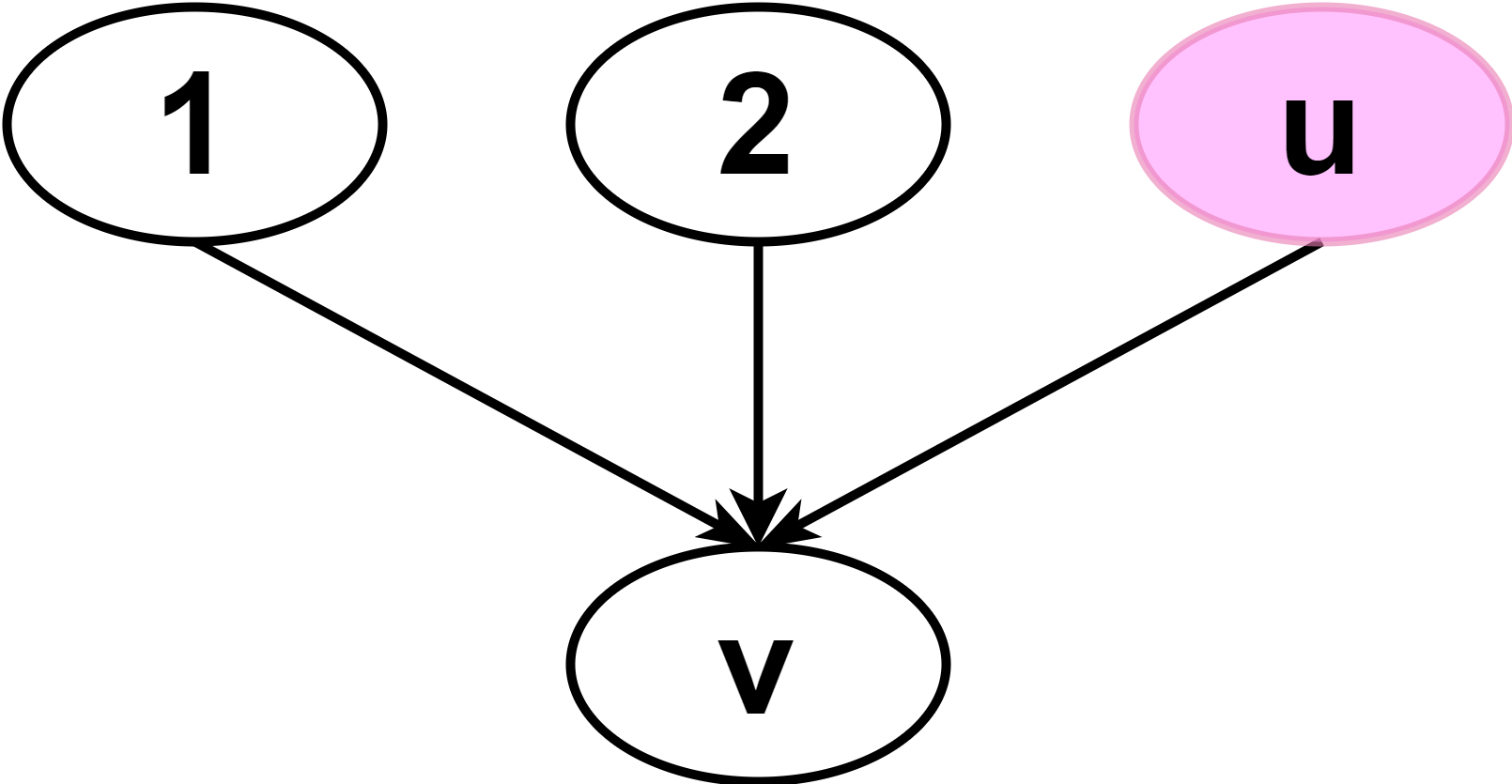}
\end{minipage}
}
\hspace{0.6em}
\subfloat[changed destination node]{     
\label{fig:CFG-change1}
\begin{minipage}[c]{.2\linewidth}
\centering
\includegraphics[width=1\textwidth]{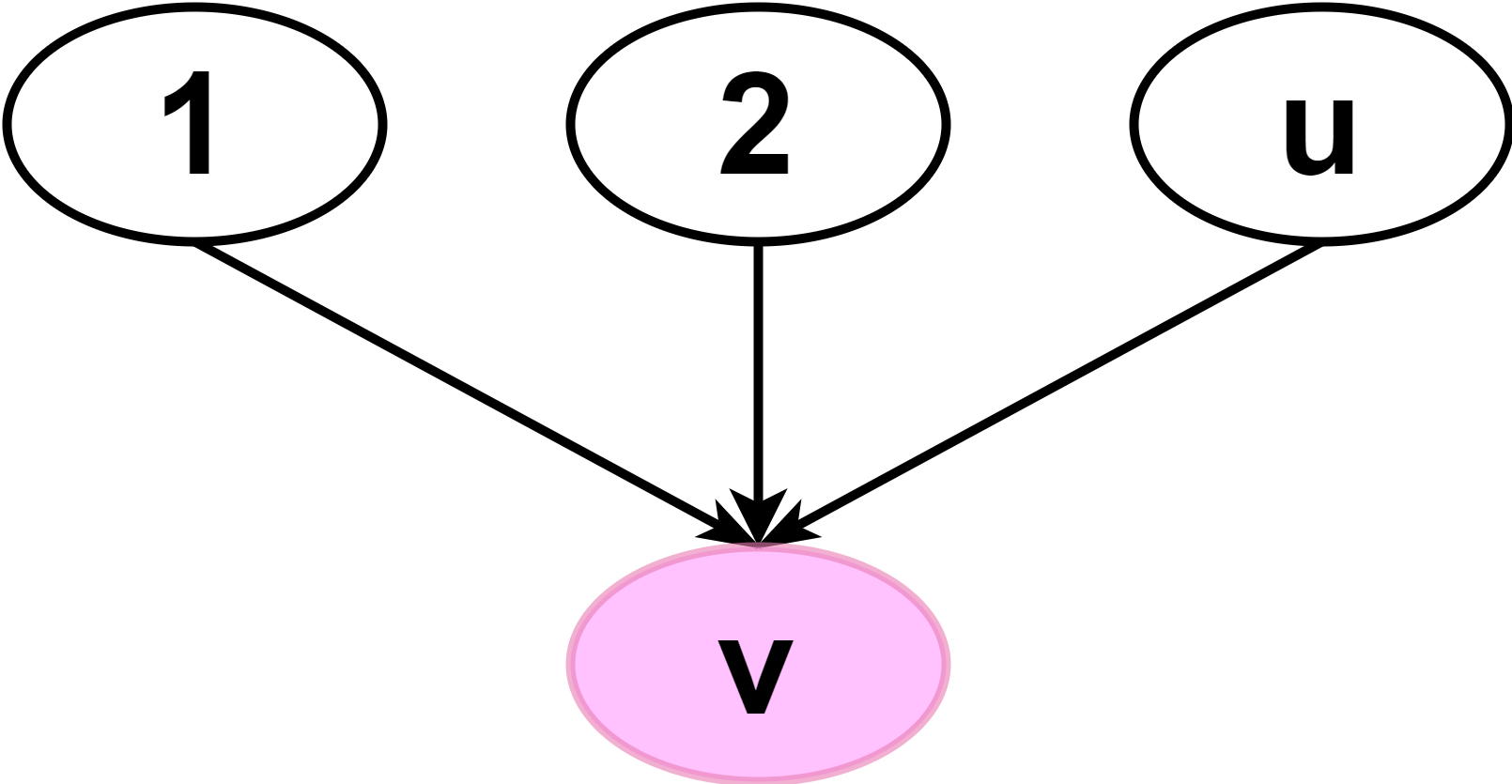}
\end{minipage}
}
\caption{Atomic changes on CFG.}
\label{fig:CFG-atomic-change}
\end{figure}

Given the code changes of CFG, we first conduct an impact analysis to identify a set of affected nodes on which incremental updates have to be performed.
Here we first enumerate all the potential cases of atomic changes on CFG shown as Figure \ref{fig:CFG-atomic-change}.
The atomic changes determine the initial set of affected nodes at the beginning. 

Figures {\ref{fig:CFG-add1}-\ref{fig:CFG-add3}} illustrate the cases of added nodes and edges on CFG at the smallest granularity. These kinds of changes affect the dataflow fact by influencing the combine procedure.
For the destination node $v$ of a newly added edge $e = \langle u,v \rangle$, the edge brings a new dataflow fact $\mathcal{OUT}_u$ from its new predecessor $u$. As the set of dataflow facts taken by combine operator is changed, the dataflow fact of $v$ has to be updated. As a result, in both three cases of addition, the node $v$ is apparently affected. 
Moreover, as a newly added source node shown as Figure \ref{fig:CFG-add2}, $u$ automatically belongs to the affected nodes.
Figures {\ref{fig:CFG-del1}-\ref{fig:CFG-del3}} illustrate the cases involving edge and node deletion. Similar to the cases of addition, the cases of deletion affect the combine procedure. Thus for all the three cases, the destination node $v$ is identified as affected.
Because $u$ in Figure \ref{fig:CFG-del2} and $v$ in Figure \ref{fig:CFG-del3} are deleted and absent in the updated CFG, we do not include them into the sub-CFG, but marking them as deleted.
Figures \ref{fig:CFG-change2} and \ref{fig:CFG-change1} show the cases with node change. The node change affects the dataflow facts by changing the transfer function. 
When node changes (\eg, the node $u$ in Figure \ref{fig:CFG-change2}), its transfer function changes as well.
The new transfer function can convert the old incoming fact into a new outgoing fact that is different from the previous one.
The associated dataflow fact needs to be updated. Hence, the changed node is determined as affected.

From the above,  it can be seen that the nodes where the atomic change occurs (\eg, nodes $u$ and $v$) are obviously affected, and their dataflow facts need to be updated accordingly. 
When the new dataflow fact of a node $v$ is transferred and propagated along the CFG iteratively, it also has a impact on all of its successor nodes.
As such, we propose a transitive closure analysis (Algorithm \ref{a:transitive-closure}) to identify all the nodes affected by changes. 
Next, all the affected nodes together with edges among them constitute a sub-CFG. 
We then directly invoke the whole-program dataflow analysis algorithm (Algorithm \ref{a:opt-algo}) over the sub-CFG to realize efficient incremental computation. 
%

\begin{algorithm}[htb!]
	\caption{\small Distributed Incremental Analysis Algorithm}
	\label{a:reachability-naive-algo}
	\DontPrintSemicolon
	\small
	\KwData{
    $\mathcal{C}$: all the atomic changes in the CFG; 
    $\mathcal{R}$: the previous dataflow analysis result;
    $\mathcal{A}$: a set of affected nodes whose dataflow fact needs to be updated; 
    $\mathcal{G}$: a sub-CFG consisting of all the affected nodes and edges among them.}
	
\BlankLine
\tcp{\textcolor{blue}{impact analysis}}
     $\mathcal{A} \leftarrow \emptyset$ \label{a-naive-reach:set-init} \;
    \ForPar {each atomic change $c \in \mathcal{C}$}
    { \label{a-naive-reach:set-base} 
        \Switch{the type of $c$}
        { 
            \uCase{added edge $\langle u \rightarrow v \rangle$ between existing nodes}
            {   \label{a-naive-reach:add-begin}
                $\mathcal{A} \leftarrow \mathcal{A} \cup \{v\} $  \label{a-naive-reach:add-edge} \;  
            }
            \uCase{added source node $u$ and edge $\langle u \rightarrow v \rangle$}
            {   
                $\mathcal{A} \leftarrow \mathcal{A} \cup \{u, v\} $ \label{a-naive-reach:add-src}\;
            }
            \uCase{added destination node $v$}
            {   
                $\mathcal{A} \leftarrow \mathcal{A} \cup \{v\} $ \label{a-naive-reach:add-dst}\;
            }
            \uCase{deleted edge $\langle u \rightarrow v \rangle$ between existing nodes}
            {   \label{a-naive-reach:delete-begin}
                $\mathcal{A} \leftarrow \mathcal{A} \cup \{v\} $ \label{a-naive-reach:delete-edge}\;
            }
            \uCase{deleted source node $u$ and edge $\langle u \rightarrow v \rangle$}
            {   
                $\mathcal{A} \leftarrow \mathcal{A} \cup \{v\} $ \label{a-naive-reach:delete-src}\;
            }
            \uCase{deleted destination node $v$}
            {   
            }
            \uCase{changed source node $u$}
            {   
                $\mathcal{A} \leftarrow \mathcal{A} \cup \{u\} $  \label{a-naive-reach:change-src} \;
            }
            \uCase{changed destination node $v$}
            {   
                $\mathcal{A} \leftarrow \mathcal{A} \cup \{v\} $ \label{a-naive-reach:change-dst} \;
            }
        }
    }
    
    $\mathcal{A} \gets $ \textsc{TransitiveClosure}($\mathcal{A}, CFG$) \label{a-naive-reach:set-trans}\;
    $V\mathcal{(G)} \leftarrow \mathcal{A}$ \textcolor{olive}{/*all the affected nodes constitute the nodes of sub-CFG*/} \label{a-naive-reach:sub-nodes} \;
    $E\mathcal{(G)} \!\!\leftarrow\!\! \{ \forall \langle u \!\!\rightarrow\!\! v \rangle \!\!\in\!\! CFG ~|~ u \!\in\! A ~\&~ v \!\in\! A  \}$ \textcolor{olive}{/*edges between affected nodes constitute edges of sub-CFG*/} \label{a-naive-reach:sub-edges} \;

\BlankLine
\tcp{\textcolor{blue}{incremental update}}   

    \ForPar {each node $k \in V\mathcal{(G)}$}
    {
        $\mathcal{IN}_k \leftarrow \bot/\top $; $\mathcal{OUT}_k \leftarrow \bot/\top$; \textcolor{olive}{/*initialize the dataflow facts of each node in the sub-CFG as the bottom (or top) value for increasing (or decreasing) analysis*/} \label{a-naive-reach:fact-init}\;
        \For {each predecessor $p$ of $k$ in the CFG \label{a-naive-reach:pred}}
        {        
            \If {$p \notin V\mathcal{(G)}$\label{a-naive-reach:pred-notin}}   
            { 
               $\mathcal{OUT}_p \!\leftarrow\!  \textsc{QueryOUT}(\mathcal{R}, p)$ \textcolor{olive}{/*query existing outgoing fact for unaffected predecessor*/} \label{a-naive-reach:pred-query} \; 
              $\mathcal{M}_k \leftarrow \mathcal{M}_k \cup \mathcal{OUT}_p$ \textcolor{olive}{/*initialize the value of $\mathcal{M}_k$ using $\mathcal{OUT}_p$*/} \label{a-naive-reach:message-init}\;
            }
        }
    }
    \textsc{Algorithm 4}($\mathcal{G}$) \textcolor{olive}{/*invoke Algorithm 4 on sub-CFG for incremental update*/}\label{algo:call4-naive}\;
\end{algorithm}

\begin{algorithm}
\caption{\small Transitive Closure Computation}
	\label{a:transitive-closure}
	\DontPrintSemicolon
	\small
 	\KwData{
    $\mathcal{W}$: the list of all active nodes during analysis;
    $\mathcal{A}$: a set of affected nodes.}
    
    \BlankLine
    
    $\mathcal{W} \leftarrow \mathcal{A}$ \;
    \Repeat{$\mathcal{W} == \emptyset$}
    {\ForPar {each node $k \in \mathcal{W}$}
        {
            remove $k$ from $\mathcal{W}$ \;
            \For {each successor $d$ of $k$ in the CFG}
                {
                \If{$d \notin \mathcal{A}$}
                {
                    $\mathcal{W}^{'} \leftarrow \mathcal{W'} \cup \{ d \}$ \;
                    $\mathcal{A} \leftarrow \mathcal{A} \cup \{ d \}$ \;
                }
                }
        }
        \textsc{Synchronize}()  \;
        $\mathcal{W} \leftarrow \mathcal{W'}$
    } 
\end{algorithm}

Algorithm \ref{a:reachability-naive-algo} gives the detailed algorithm for our distributed incremental analysis. 
At the beginning, the set of affected nodes $\mathcal{A}$ is initialized as empty (Line~\ref{a-naive-reach:set-init}).
Next, for the affected nodes at each atomic CFG change $c$, they are added to $\mathcal{A}$ accordingly (Line \ref{a-naive-reach:set-base}-\ref{a-naive-reach:change-dst}).
Having the set of directly affected nodes at the atomic changes, we next perform transitive closure computation to include all the indirectly affected nodes. 
In particular, the \textsc{TransitiveClosure} procedure (shown in Algorithm \ref{a:transitive-closure}) takes the atomic set and the CFG as input, identify and include all the successor nodes affected to $\mathcal{A}$ (Line \ref{a-naive-reach:set-trans}).
Next, all the affected nodes in $\mathcal{A}$ constitute the nodes in sub-CFG $\mathcal{G}$ (Line \ref{a-naive-reach:sub-nodes}), and the edges between them constitute the edges of sub-CFG (Line \ref{a-naive-reach:sub-edges}). 
The incremental update then only needs to be performed over such sub-CFG $\mathcal{G}$ to update the dataflow facts accordingly. 

Importantly, before performing the analysis over the sub-CFG, we first need to do necessary initialization. 
Basically, for each node $k$ of the sub-CFG, their new incoming and outgoing facts are initialized as $\bot$ (or $\top$) for increasing (or decreasing) analysis, respectively (Line \ref{a-naive-reach:fact-init}).
Moreover, in order to perform dataflow analysis over node $k$ successfully, we need to obtain all its predecessors' outgoing facts. 
For a predecessor node $p$ of $k$ in the whole CFG, if $p$ is not affected, then $p$ is not included in the sub-CFG (Lines \ref{a-naive-reach:pred} and \ref{a-naive-reach:pred-notin}).
In this case, we need to firstly query the outgoing fact of $p$, \ie, $\mathcal{OUT}_p$ from prior results, and then initialize the value of $\mathcal{M}_k$ as $\mathcal{OUT}_p$  (Lines \ref{a-naive-reach:pred-query}-\ref{a-naive-reach:message-init}).  
After the initialization, we can directly invoke the optimized whole-program dataflow analysis (\ie, Algorithm \ref{a:opt-algo}) over the sub-CFG to realize incremental computation (Line \ref{algo:call4-naive}).

\begin{figure}[htb!]
\centering
\hspace{-1em}
\subfloat[CFG with changes]{         
\label{fig:reach-naive-cfg}
\begin{minipage}[c]{.23\linewidth}
\centering
\includegraphics[width=0.9\textwidth]{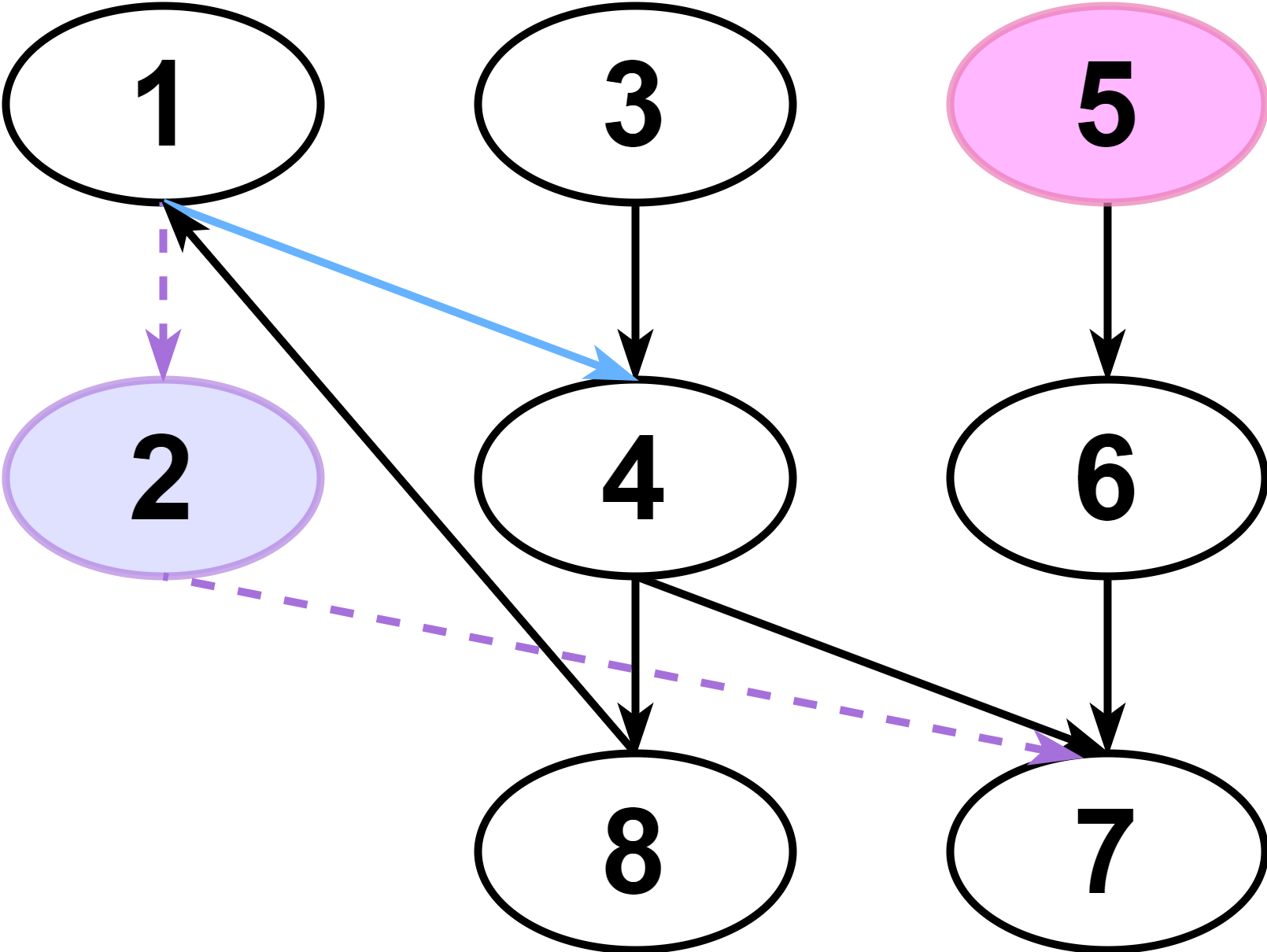}
\end{minipage}
}
\subfloat[initial affected nodes]{         
\label{fig:reach-naive-init}
\begin{minipage}[c]{.23\linewidth}
\centering
\includegraphics[width=0.9\textwidth]{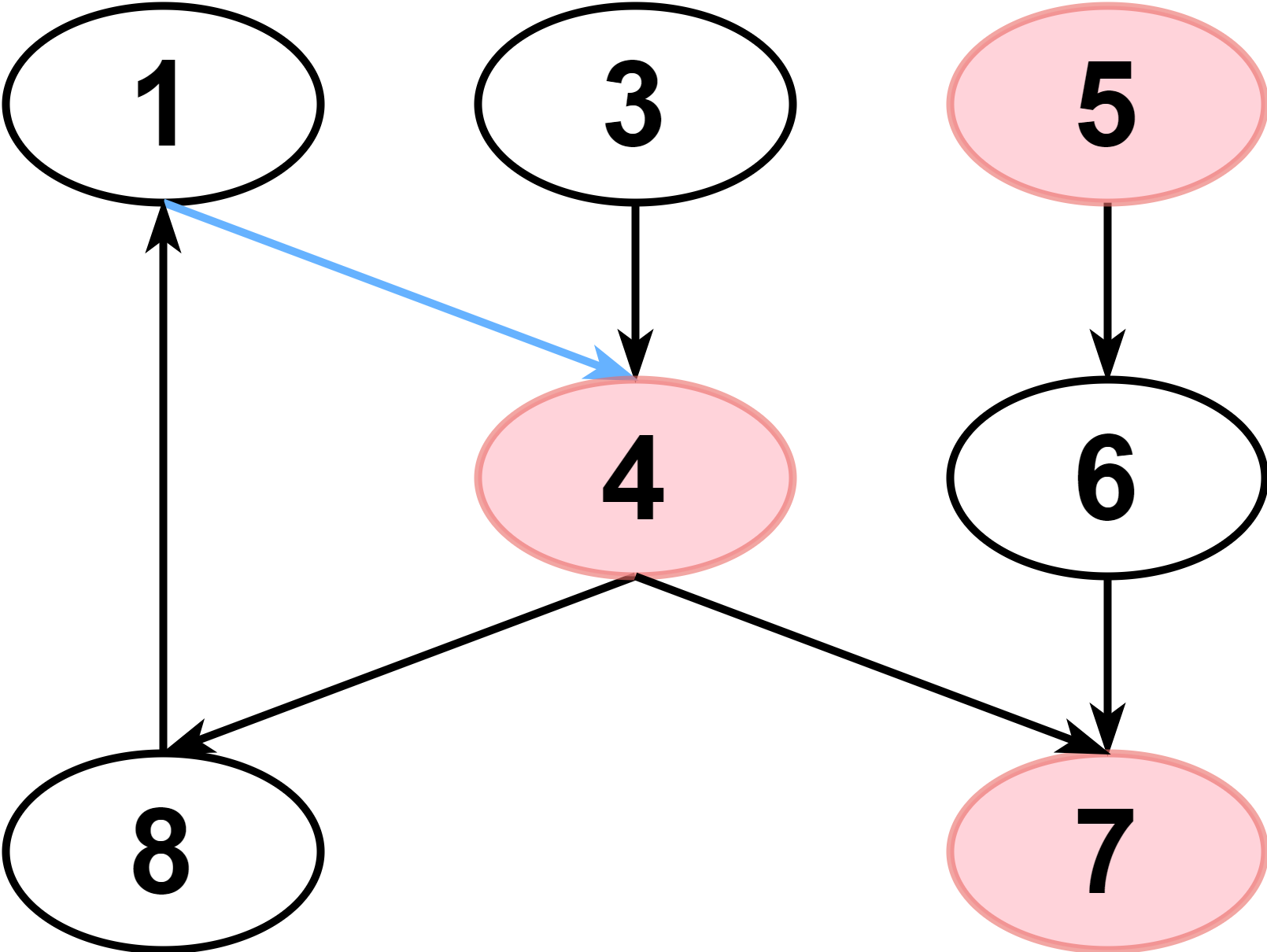}
\end{minipage}
}
\subfloat[final affected nodes]{         
\label{fig:reach-naive-compute}
\begin{minipage}[c]{.23\linewidth}
\centering
\includegraphics[width=0.9\textwidth]{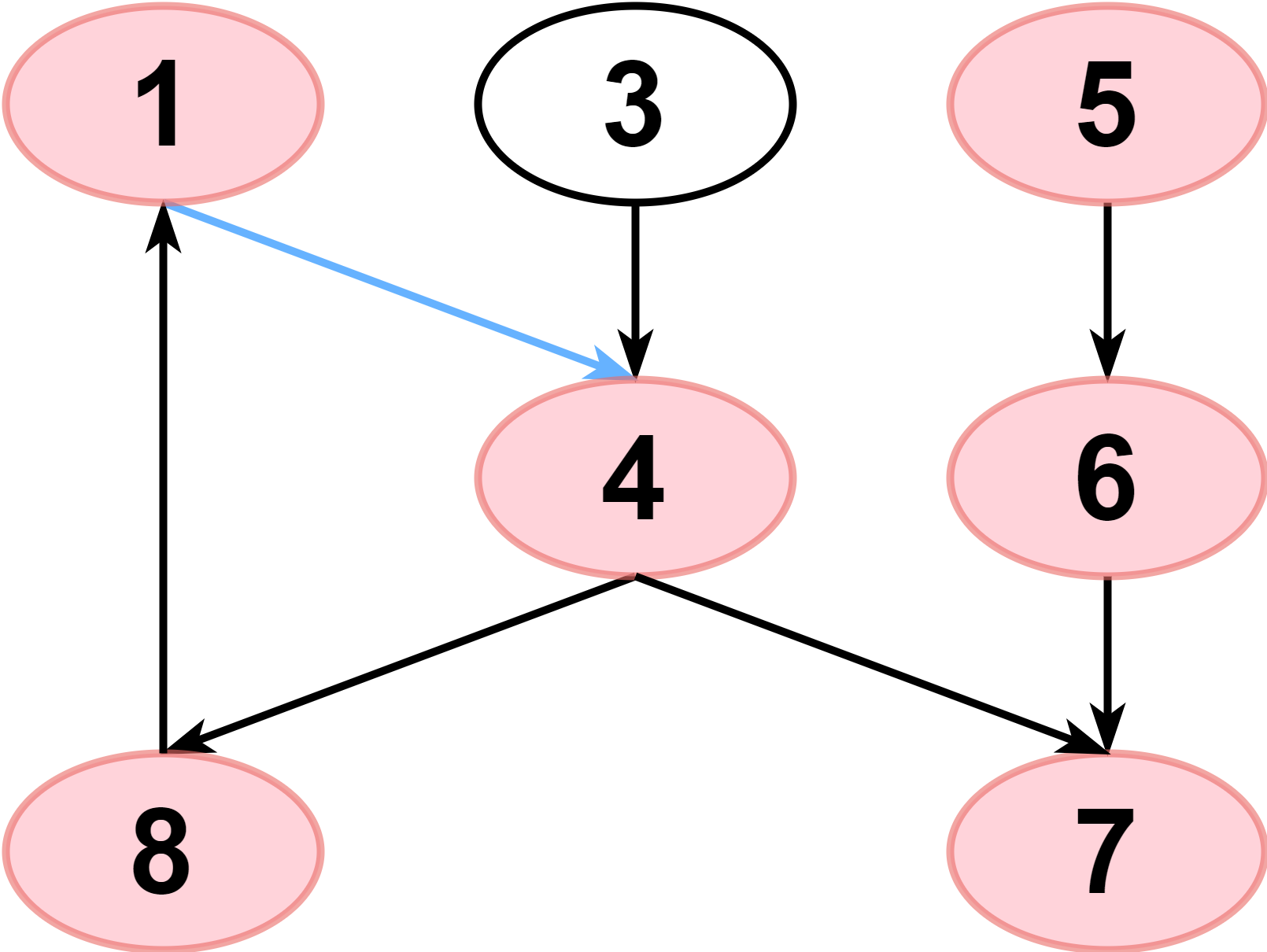}
\end{minipage}
}
\subfloat[sub-CFG]{         
\label{fig:incre-naive-sub-cfg}
\begin{minipage}[c]{.23\linewidth}
\centering
\includegraphics[width=0.9\textwidth]{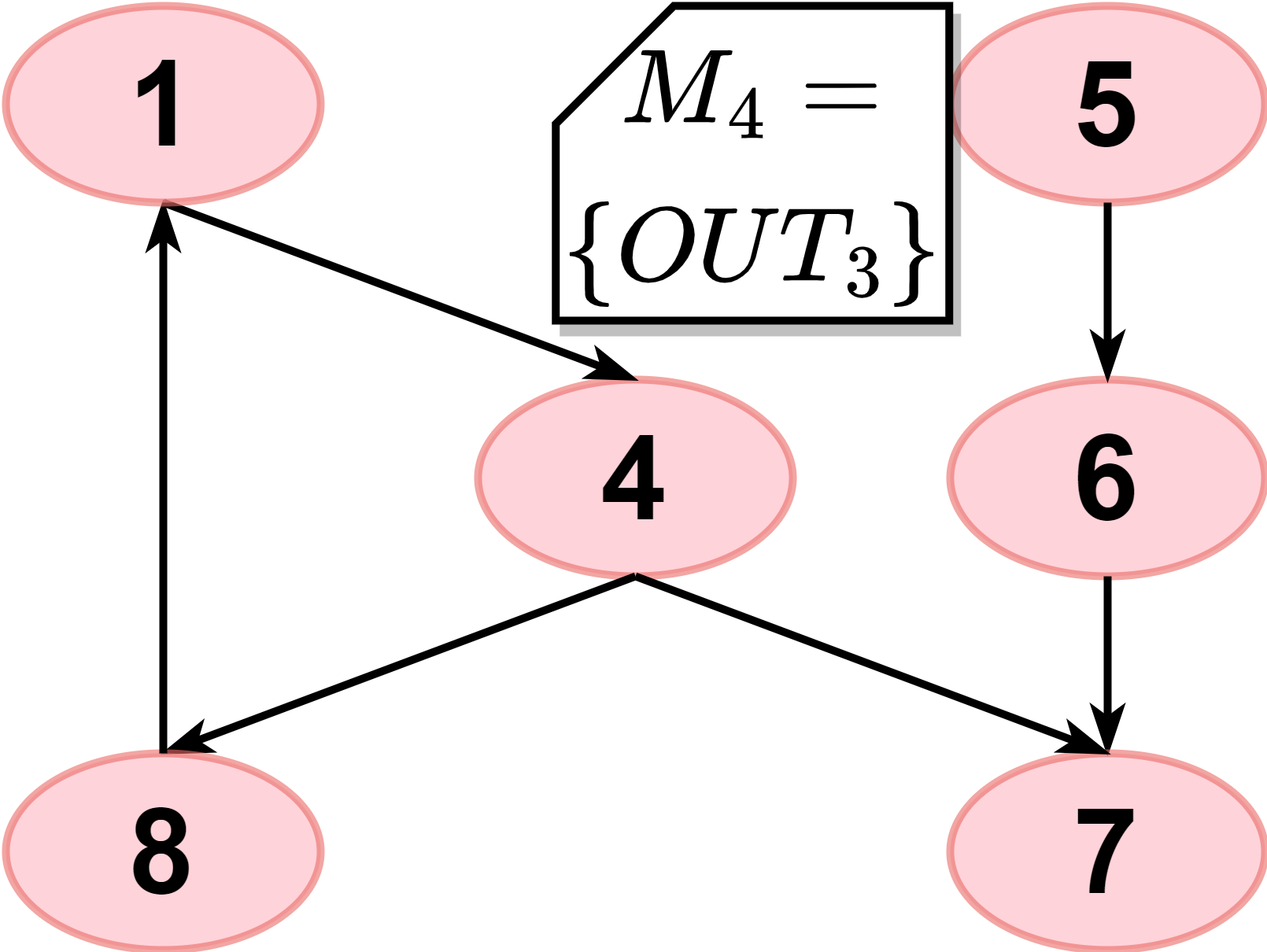}
\end{minipage}
}
\caption{Sub-CFG for incremental update in Algorithm \ref{a:reachability-naive-algo}.}
\label{fig:CFG-change-naive-exp}
\end{figure}

\MyPara{Example.} 
Figure \ref{fig:CFG-change-naive-exp} shows an example demonstrating how to compute the set of affected nodes $\mathcal{A}$ in Algorithm \ref{a:reachability-naive-algo}.
Similar to Figure \ref{fig:CFG-atomic-change}, the nodes and edges in blue, purple and pink indicate addition, deletion, and change cases, respectively.
Suppose that the batch of updates on the CFG includes adding a new edge $\langle 1,4 \rangle$, deleting node 2, and changing node 5 (shown as Figure \ref{fig:reach-naive-cfg}).
Apparently, node $2$ is marked as deleted. 
According to Lines \ref{a-naive-reach:set-base}-\ref{a-naive-reach:change-dst} of Algorithm \ref{a:reachability-naive-algo}, the nodes $\{4, 5, 7\}$ are initially identified as affected and put into $\mathcal{A}$ (marked as red in Figure \ref{fig:reach-naive-init}). 
Next, the transitive closure computation determines the all the successors affected, \ie, $\{1, 6, 8\}$. Thus, the final result is $\mathcal{A}= \{1, 4, 5, 6, 7, 8\}$ (shown in Figure \ref{fig:reach-naive-compute}).
Among these nodes, only node 4 has an unaffected predecessor 3 (marked in white in Figure \ref{fig:incre-naive-sub-cfg}). It queries the previous outgoing fact $\mathcal{OUT}_3$ from predecessor 3 as messages $\mathcal{M}_4$. Finally, Algorithm \ref{a:opt-algo} is invoked to do incremental update over the sub-CFG (shown as Figure \ref{fig:incre-naive-sub-cfg}).

\subsection{Optimized Distributed Incremental Analysis Algorithm \label{subsec:incremental-opt}}

The above algorithm succeeds in realizing incremental analysis by only performing updates on the affected nodes. 
However, as shown in Algorithm \ref{a:reachability-naive-algo}, the dataflow analysis performed over the sub-CFG has to start from scratch, \ie,  with the incoming and outgoing facts associated with each node as the initial value (\ie, $\bot$ or $\top$). 
In other words, the iterative computation has to start from the very beginning, which usually takes a number of iterations to reach fix-point convergence. 
To further improve the efficiency of incremental analysis, here we propose an optimized analysis algorithm which leverages the previous analysis result to accelerate the convergence.

In monotone dataflow analysis, the dataflow facts associated with each node monotonically increase/decrease  until they no longer change (\ie, they reach the convergence state). 
The \textbf{\textit{closer}} the dataflow fact is to the final state, the \textbf{\textit{faster}} the computation converges.
In other words, an affected node reaches its final dataflow fact with fewer iterations if the computation starts with a closer dataflow fact initialized.
In incremental analysis scenario, the potential closer dataflow fact which can be exploited would be the analysis result obtained from the previous analysis. 
Instead of initializing the incoming and outgoing facts as $\bot$ or $\top$, we directly assign $\mathcal{IN}_k$ and $\mathcal{OUT}_k$ from previous analysis as the initial value.
In this way, it is possible to get a closer incoming dataflow fact of the affected node. The number of iterations (aka., supersteps) needed can thus be greatly reduced, thus achieving convergence speedup.
However, naively initializing the dataflow facts of all the affected nodes as previous analysis results may lead to analysis incorrectness. 
That is the analysis results of such optimization may differ from the results of the naive algorithm (\ie, Algorithm \ref{a:reachability-naive-algo}). 
The reason is that in incremental analysis, the actual dataflow facts of $k$, \ie, $\mathcal{IN}_k$ and $\mathcal{OUT}_k$ may converge before reaching the state of the previous analysis result. In this case, the previous result is not a safe starting state.  
Here, the key is for each affected node in the sub-CFG, how to determine if the previous analysis result is a safe initial state for incremental analysis.


To this end, we define the following incremental property. For an affected node, if it satisfies the incremental property, it is safe to reuse the previous analysis results as their initial dataflow fact values. 
We will discuss the correctness of such optimization and give the proof shortly in \cref{subsec:proof-incre}.

\begin{theorem}[\textbf{Incremental Property}]
\label{theorem:incremental-property}
Given an affected node $k$ in the sub-CFG, let $\mathcal{IN}^*_k$ (and $\mathcal{OUT}^*_k$) and $\mathcal{IN}_k$ (and $\mathcal{OUT}_k$) be the incoming (and outgoing) dataflow facts before and after the change, respectively. The incremental property is satisfied if and only if the following equation holds:
\[\mathcal{IN}^*_k \leq \mathcal{IN}_k ~~~~\&~~~~ \mathcal{OUT}^*_k \leq \mathcal{OUT}_k\]
\end{theorem}

Considering the atomic changes in Figure \ref{fig:CFG-atomic-change}, these changes can illustrate the impact on dataflow facts from the perspective of incremental property.
For the destination node of a newly added edge, the edge brings a new dataflow fact from its predecessor that satisfies the incremental property. For example, in Figure \ref{fig:CFG-add1}, for increasing analysis, the propagated dataflow fact from $u$ to $v$ before edge $\langle u,v \rangle$ addition is $\mathcal{IN}^*_v = \mathcal{OUT}^*_1 \cup \mathcal{OUT}^*_2$. 
While after the edge addition, the dataflow fact propagated to $v$ becomes $\mathcal{IN}_v = \mathcal{OUT}^*_1 \cup \mathcal{OUT}^*_2 \cup \mathcal{OUT}^*_u$. Apparently, $\mathcal{IN}^*_v \equiv (\mathcal{OUT}^*_1 \cup \mathcal{OUT}^*_2) \leq (\mathcal{OUT}^*_1 \cup \mathcal{OUT}^*_2 \cup \mathcal{OUT}^*_u) \equiv \mathcal{IN}_v$. Since node $v$ is not changed, and its transfer function is monotonic, we can conclude that $\mathcal{OUT}^*_v \leq \mathcal{OUT}_v$. 
As a result, the computation at $v$ satisfies incremental property, vice versa for decreasing analysis. 
On the contrary, the cases that exist edge deletion and node change may violate the incremental property. For the destination node of a newly deleted edge, it no longer receives dataflow fact from the previous predecessor (shown as Figure \ref{fig:CFG-del1}). Hence the partial relation between the old and new outgoing fact is  $\mathcal{OUT}_u \leq \mathcal{OUT}^*_u$. Similarly, for node change (shown as Figure \ref{fig:CFG-change1}), the old and new outgoing dataflow facts from predecessors may violate partial order relation since the impact of statements changes on the monotonicity of the transfer function is indeterminate.
To sum up, different atomic changes have different impacts on the incremental property. 
The nodes involved in addition satisfy incremental property, whereas the nodes involved in deletion and change may not satisfy incremental property.
Moreover, as the impact of each atomic change propagates along the sub-CFG, a node in the sub-CFG may be influenced by multiple atomic changes.
For a node in the sub-CFG, if it is only influenced by edge addition, then we can conclude that it satisfies incremental property; otherwise, it does not.

\begin{algorithm}[htb!]
	\caption{\small Optimized Distributed Incremental Analysis Algorithm}
	\label{a:reachability-optimized-algo}
	\DontPrintSemicolon
	\small
	\KwData{
    $\mathcal{C}$: all the atomic changes in the CFG; 
    $\mathcal{R}$: the previous dataflow analysis result;
    $\mathcal{A}_{add}$: a set of affected nodes related to addition; 
    $\mathcal{A}_{delete}$: a set of affected nodes related to deletion; 
    $\mathcal{A}_{change}$: a set of affected nodes related to change; 
    $\mathcal{G}$: a sub-CFG consisting of all the affected nodes and edges among them.}
	
	\BlankLine
 \tcp{\textcolor{blue}{impact analysis}}
     $\mathcal{A}_{add} \leftarrow \emptyset; ~~\mathcal{A}_{delete} \leftarrow \emptyset; ~~\mathcal{A}_{change} \leftarrow \emptyset$; \label{a-opt-reach:3set-init} \;
    \ForPar {each atomic change $c \in \mathcal{C}$}
    { \label{a-opt-reach:3set-base} 
        \Switch{the type of $c$}
        {
            \uCase{added edge $\langle u \rightarrow v \rangle$ between existing nodes}
            {   \label{a-opt-reach:add-begin}
                $\mathcal{A}_{add} \leftarrow \mathcal{A}_{add} \cup \{v\} $ \label{a-opt-reach:add-edge} \;  
            }
            \uCase{added source node $u$ and edge $\langle u \rightarrow v \rangle$}
            {   
                $\mathcal{A}_{add} \leftarrow \mathcal{A}_{add} \cup \{u, v\} $ \label{a-opt-reach:add-src}\;
            }
            \uCase{added destination node $v$}
            {   
                $\mathcal{A}_{add} \leftarrow \mathcal{A}_{add} \cup \{v\} $ \label{a-opt-reach:add-dst}\;
            }
            \uCase{deleted edge $\langle u \rightarrow v \rangle$ between existing nodes}
            {   \label{a-opt-reach:delete-begin}
                $\mathcal{A}_{delete} \leftarrow \mathcal{A}_{delete} \cup \{v\} $ \label{a-opt-reach:delete-edge}\;
            }
            \uCase{deleted source node $u$ and edge $\langle u \rightarrow v \rangle$}
            {   
                $\mathcal{A}_{delete} \leftarrow \mathcal{A}_{delete} \cup \{v\} $ \label{a-opt-reach:delete-src} \;
            }
            \uCase{deleted destination node $v$}
            {   
            }
            \uCase{changed source node $u$}
            {   
                $\mathcal{A}_{change} \leftarrow \mathcal{A}_{change} \cup \{u\} $ \label{a-opt-reach:change-src} \;
            }
            \uCase{changed destination node $v$}
            {   
                $\mathcal{A}_{change} \leftarrow \mathcal{A}_{change} \cup \{v\} $ \label{a-opt-reach:change-dst} \;
            }
        }
    }
    
    $\mathcal{A}_{add} \gets $ \textsc{TransitiveClosure}($\mathcal{A}_{add}, CFG$) \label{a-opt-reach:add-trans}\;
    $\mathcal{A}_{delete} \gets $ \textsc{TransitiveClosure}($\mathcal{A}_{delete}, CFG$) \label{a-opt-reach:delete-trans}\;
    $\mathcal{A}_{change} \gets $ \textsc{TransitiveClosure}($\mathcal{A}_{change}, CFG$) \label{a-opt-reach:change-trans}\;

    $V\mathcal{(G)} \leftarrow \mathcal{A}_{add} \cup \mathcal{A}_{delete} \cup \mathcal{A}_{change} $ \textcolor{olive}{/*all the affected nodes constitute the nodes of sub-CFG*/}  \label{a-opt-reach:sub-nodes} \;
    $E\mathcal{(G)} \!\!\leftarrow\!\! \{ \forall \langle u \!\!\rightarrow\!\! v \rangle \!\in\! CFG ~|~ u \!\in\! V\mathcal{(G)} ~\&~ v \!\in\! V\mathcal{(G)}  \}$ \textcolor{olive}{/*edges between affected nodes constitute edges of sub-CFG*/} \label{a-opt-reach:sub-edges} \;

\BlankLine
\tcp{\textcolor{blue}{incremental update}}
    \ForPar {each node $k \in V\mathcal{(G)}$}
    {
        \uIf{$k \in \mathcal{A}_{add} - \mathcal{A}_{delete} - \mathcal{A}_{change}$ \label{a-opt-init:ifcondition}}
        { 
            $\mathcal{IN}_k \leftarrow \textsc{QueryIN}(\mathcal{R}, k); ~~\mathcal{OUT}_k \leftarrow \textsc{QueryOUT}(\mathcal{R}, k);$ \textcolor{olive}{/*initialize the dataflow facts as the previous values*/} \label{a-opt-reach:accu-fact}\;
        }
        \Else
        { 
            $\mathcal{IN}_k \leftarrow \bot/\top $; ~~$\mathcal{OUT}_k \leftarrow \bot/\top$; \textcolor{olive}{/*initialize dataflow facts as the  bottom (or top) value for increasing (or decreasing) analysis*/} \label{a-opt-reach:nonaccu-init}\;

            \For {each predecessor $p$ of $k$ in the CFG}
            {        
                \If {$p \notin V\mathcal{(G)} ~~||~~ p \in \mathcal{A}_{add} - \mathcal{A}_{delete} - \mathcal{A}_{change}$}  
                { \label{a-opt-reach:nonaccu-close}
                   $\mathcal{OUT}_p \!\leftarrow\!  \textsc{QueryOUT}(\mathcal{R}, p)$ \textcolor{olive}{/*query existing outgoing fact for unaffected predecessor and new predecessors influenced only by addition cases*/} \label{a-opt-reach:nonaccu-outset} \; 
                  $\mathcal{M}_k \leftarrow \mathcal{M}_k \cup \mathcal{OUT}_p$ \textcolor{olive}{/*initialize the value of $\mathcal{M}_k$ using $\mathcal{OUT}_p$*/} \label{a-opt-reach:nonaccu-message}\;
                }
            }
        }
        
    }
    \textsc{Algorithm 4}($\mathcal{G}$) \textcolor{olive}{/*invoke Algorithm 4 on sub-CFG for incremental update*/}\label{algo:call4-opt}\;
\end{algorithm}

Base on the above theorem, we classify the affected nodes $\mathcal{A}$ into the following three categories.
Specifically, $\mathcal{A} = \mathcal{A}_{add} \cup \mathcal{A}_{delete} \cup \mathcal{A}_{change}$ where: 
\begin{itemize}
    \item  $\mathcal{A}_{add}$ is the set of affected nodes influenced by addition cases. 
    \item  $\mathcal{A}_{delete}$ is a set of affected nodes influenced by deletion cases.
    \item  $\mathcal{A}_{change}$ is a set of affected nodes influenced by node change cases.
\end{itemize} 

For the nodes $k \in \mathcal{A}_{add} - \mathcal{A}_{delete} -\mathcal{A}_{change}$, it satisfies incremental property, and thus it can directly reuse the results from previous analysis as its initial values. 
For other affected nodes $k \in \mathcal{A}$, there is no guarantee that its previous and new dataflow facts satisfy the incremental property. Therefore, we can only initialize their dataflow facts as $\bot/\top$ for the sake of correctness.

Algorithm \ref{a:reachability-optimized-algo} presents the procedure of the optimized incremental analysis.
At the beginning, the three sets of affected nodes are all initialized as empty (Line~\ref{a-opt-reach:3set-init}).
Next, for affected nodes in different atomic CFG changes, they are added to $\mathcal{A}_{add}$, $\mathcal{A}_{delete}$, and $\mathcal{A}_{change}$ accordingly (Line \ref{a-opt-reach:3set-base}-\ref{a-opt-reach:change-dst}).
The \textit{TransitiveClosure} procedure takes each of the atomic sets and the CFG with updates as input, and includes all the transitively reachable nodes into the respective sets (Line \ref{a-opt-reach:add-trans}-\ref{a-opt-reach:change-trans}).
All the affected nodes in the union of three sets constitute the nodes of sub-CFG (Line \ref{a-opt-reach:sub-nodes}).   
The edges between these affected nodes constitute the edges of sub-CFG (Line \ref{a-opt-reach:sub-edges}).
For the affected node only influenced by addition cases (\ie, $k \in \mathcal{A}_{add} - \mathcal{A}_{delete} -\mathcal{A}_{change}$), it directly queries its previous incoming and outgoing facts for initialization (Line~\ref{a-opt-reach:accu-fact}). 
For other affected nodes influenced by deletion or change cases, their new incoming and outgoing fact are initialized as $\bot/\top$ to ensure correctness (Line~\ref{a-opt-reach:nonaccu-init}).
Moreover, for the predecessor $p$ of $k$, if $p$ is unaffected or satisfies incremental property, we query its outgoing fact from prior results, and then initialize the value of $\mathcal{M}_k$ as $\mathcal{OUT}_p$ (Line \ref{a-opt-reach:nonaccu-close}-\ref{a-opt-reach:nonaccu-message}). 
After the initialization, the optimized whole-program dataflow analysis (\ie, Algorithm \ref{a:opt-algo}) is performed over the sub-CFG to realize incremental computation (Line \ref{algo:call4-opt}). 

\begin{figure}[htb!]
\centering
\hspace{-1em}
\subfloat[CFG with changes]{         
\label{fig:reach-cfg}
\begin{minipage}[c]{.23\linewidth}
\centering
\includegraphics[width=0.9\textwidth]{samples/impact_exp1.png}
\end{minipage}
}
\subfloat[initial affected nodes]{         
\label{fig:reach-init}
\begin{minipage}[c]{.23\linewidth}
\centering
\includegraphics[width=0.9\textwidth]{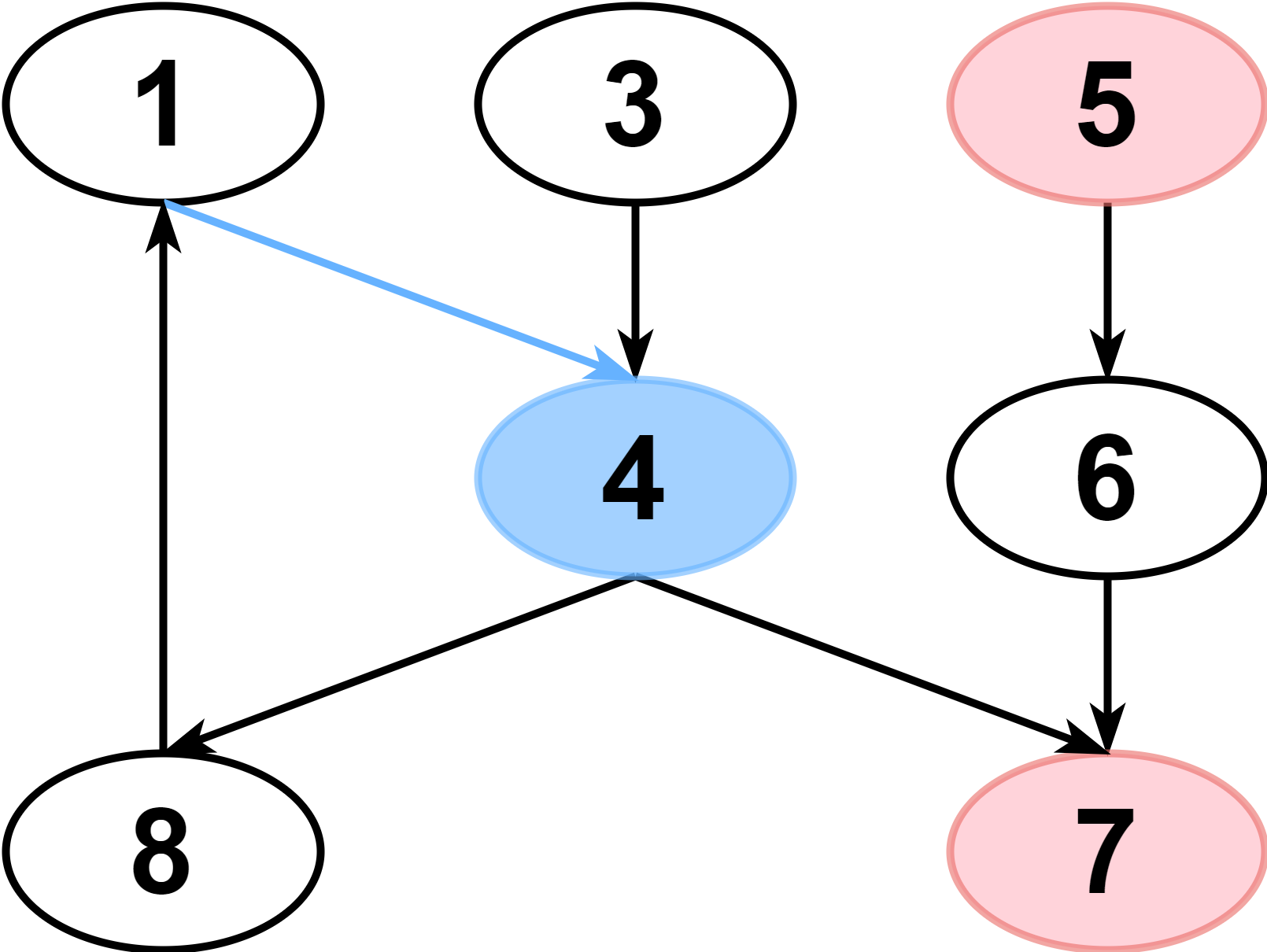}
\end{minipage}
}
\subfloat[final affected nodes]{         
\label{fig:reach-compute}
\begin{minipage}[c]{.23\linewidth}
\centering
\includegraphics[width=0.9\textwidth]{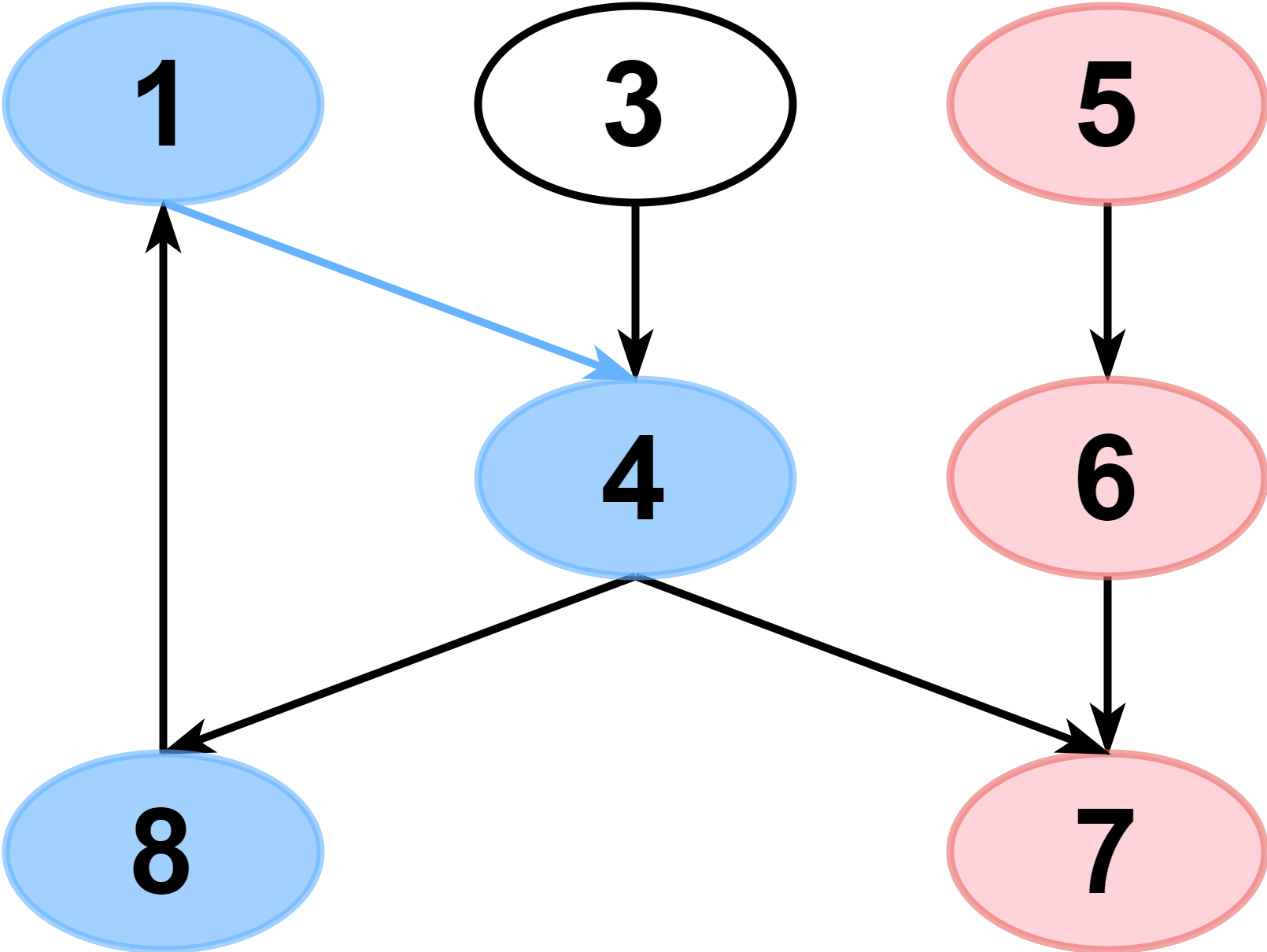}
\end{minipage}
}
\subfloat[sub-CFG]{         
\label{fig:incre-sub-cfg}
\begin{minipage}[c]{.23\linewidth}
\centering
\includegraphics[width=0.9\textwidth]{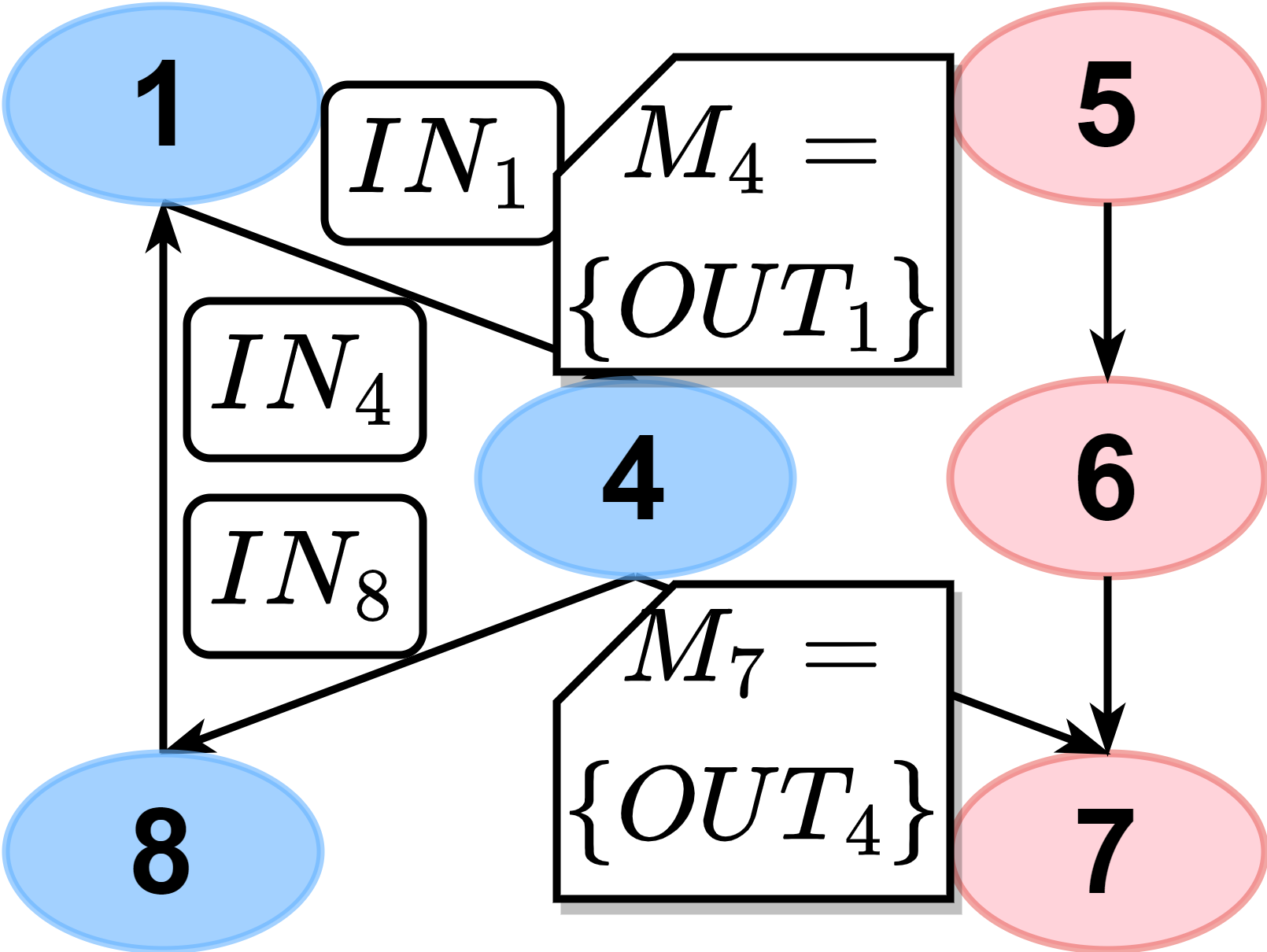}
\end{minipage}
}
\caption{Sub-CFG for incremental update in Algorithm \ref{a:reachability-optimized-algo}.}
\label{fig:CFG-change-exp}
\end{figure}

\MyPara{Example.} 
Here we use the same example as Figure \ref{fig:CFG-change-naive-exp} to demonstrate the procedure of the optimized incremental analysis in Algorithm \ref{a:reachability-optimized-algo}.
At the beginning of the optimized impact analysis, $\mathcal{A}_{add}, \mathcal{A}_{delete}$, and $\mathcal{A}_{change}$ are all initialized as empty. Based on the impact of atomic changes, the affected nodes are classified into the three set accordingly.  We get the initial results as $\mathcal{A}_{add} = \{4\}$, $\mathcal{A}_{delete} = \{7\}$, and $\mathcal{A}_{change} = \{5\}$. The affected nodes influenced by addition case are marked blue while others marked red in Figure \ref{fig:reach-init}. 
Next, all the successor nodes affected are identified. The final results are $\mathcal{A}_{add} = \{1, 4, 7, 8\}$, $\mathcal{A}_{delete} = \{7\}$, and $\mathcal{A}_{change} = \{5, 6, 7\}$ (shown in Figure \ref{fig:reach-compute}).
The nodes that are only affected by the addition cases are $\{1, 4, 8\}$. For these three nodes, we can directly reuse their old incoming facts for initialization. 
Besides, nodes 4 and 7 have a predecessor 1 and 4 which is only influenced by addition, respectively. We query the previous outgoing facts $\mathcal{OUT}_1$ and $\mathcal{OUT}_4$ as messages to node $4$ and $7$, \ie, $\mathcal{M}_4$ and $\mathcal{M}_7$. Here $\mathcal{OUT}_3$ is not included into $\mathcal{M}_4$, since it is already subsumed by $\mathcal{IN}_4$. 
Finally, Algorithm \ref{a:opt-algo} is invoked to do incremental update over the sub-CFG (shown as Figure \ref{fig:incre-sub-cfg}). 

\subsection{Correctness Proof of Incremental Algorithms\label{subsec:proof-incre}}
\begin{figure}[htb!]
\centering
{         
\begin{minipage}[c]{0.8\linewidth}
\centering
\includegraphics[width=1\textwidth]{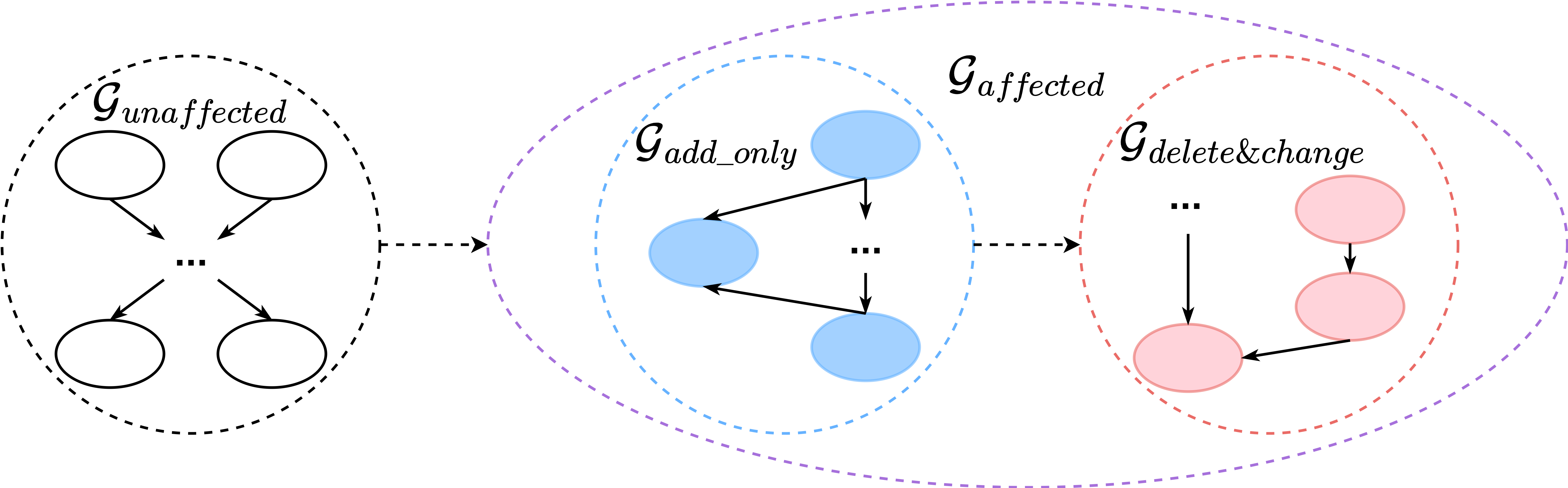}
\end{minipage}
}
\caption{The division of a CFG into two sub-CFGs namely $\mathcal{G}_{\textit{unaffected}}$ and $\mathcal{G}_{\textit{affected}}$, where $\mathcal{G}_{\textit{affected}}$ can be further divided into $\mathcal{G}_{\textit{add\_only}}$ and $\mathcal{G}_{\textit{delete\&change}}$.}
\label{fig:division-CFG-Updates}
\end{figure}

In this section, we give the correctness proofs of our incremental analysis algorithms. 
For the sake of illustration, we divide a CFG into three sub-CFGs shown as Figure \ref{fig:division-CFG-Updates}.
As discussed in Algorithm \ref{a:reachability-naive-algo}, the nodes affected by changes as well as the edges among them constitute a sub-CFG denoted as $\mathcal{G}_{\textit{affected}}$. 
The nodes which are not affected by changes and the edges among them constitute a sub-CFG denoted as  $\mathcal{G}_{\textit{unaffected}}$. 
The nodes in the two sub-CFGs are disjoint. 
There may exist edges from $\mathcal{G}_{\textit{unaffected}}$ to $\mathcal{G}_{\textit{affected}}$.
But  no edge exists from $\mathcal{G}_{\textit{affected}}$ to $\mathcal{G}_{\textit{unaffected}}$ due to the transitive closure analysis over affected nodes in Algorithm \ref{a:reachability-naive-algo}. 
Furthermore, according to the different updates (\ie, addition, deletion or change),  $\mathcal{G}_{\textit{affected}}$ can be further partitioned into $\mathcal{G}_{\textit{add\_only}}$ and $\mathcal{G}_{\textit{delete\&change}}$.
The affected nodes influenced only by addition together with the edges among them constitute $\mathcal{G}_{\textit{add\_only}}$. The other affected nodes and edges among them constitute $\mathcal{G}_{\textit{delete\&change}}$.
Similarly, only edges from $\mathcal{G}_{\textit{add\_only}}$ to $\mathcal{G}_{\textit{delete\&change}}$ are possible due to the existence of transitive closure analysis.

\MyPara{Correctness Proof of Incremental Analysis in Algorithm \ref{a:reachability-naive-algo}.}
Proving the correctness of our incremental analysis is essentially to prove that the results produced by performing the whole-program analysis over the entire CFG from scratch should be consistent with that obtained by our incremental analysis.
To be specific, the baseline is to perform the whole-program dataflow analysis over the entire CFG (including both $\mathcal{G}_{\textit{unaffected}}$ and $\mathcal{G}_{\textit{affected}}$), while with the dataflow facts associated with all the nodes initialized as $\bot$.
In contrast, our incremental analysis is to perform the whole-program analysis only over $\mathcal{G}_{\textit{affected}}$ with $\bot$ as the initial value. The results associated with $\mathcal{G}_{\textit{unaffected}}$ are still the convergent state of the previous analysis. 

Here we prove the consistency of two analyses by controlling the processing order of CFG nodes in the worklist algorithm. 
Based on the formal proof in \cite{Kam:1977-monotoneframework}, it is well-known that the maximal fixed point is unique, independent of the processing order of nodes in an iterative algorithm. 
The results of the baseline are identical to that produced by a two-phase procedure. 
Specifically, at the first phase, only the nodes in $\mathcal{G}_{\textit{unaffected}}$ are processed iteratively until the dataflow facts associated with these nodes reach fixed point.  
Since no edge exists from $\mathcal{G}_{\textit{affected}}$ to $\mathcal{G}_{\textit{unaffected}}$,  the convergent state of $\mathcal{G}_{\textit{unaffected}}$ is independent of the state of $\mathcal{G}_{\textit{affected}}$. Therefore, the current  state of $\mathcal{G}_{\textit{unaffected}}$ is exactly its final convergent state. 
When the computation over $\mathcal{G}_{\textit{unaffected}}$ is finished,  the nodes in  $\mathcal{G}_{\textit{affected}}$ are then processed iteratively to reach the fixed point. 
Apparently, our incremental analysis in Algorithm \ref{a:reachability-naive-algo} is totally same as the above two-phase procedure.
Therefore, we can conclude that our incremental analysis produces the same results as the baseline. 


\MyPara{Correctness Proof of Optimized Incremental Analysis in Algorithm \ref{a:reachability-optimized-algo}.}
Different from Algorithm \ref{a:reachability-naive-algo}, the optimized analysis further distinguishes the affected sub-CFGs as $\mathcal{G}_{\textit{add\_only}}$ and $\mathcal{G}_{\textit{delete\&change}}$. 
The nodes in $\mathcal{G}_{\textit{add\_only}}$ take the previous convergent results as their initial dataflow values. 
Whereas, the nodes in $\mathcal{G}_{\textit{delete\&change}}$ take $\bot$ as the initial value.

Here we prove the correctness of our optimized incremental analysis by following the same idea.
The results of the baseline are identical to that produced by a three-phase procedure. 
To be specific, at the first phase, only the nodes in $\mathcal{G}_{\textit{unaffected}}$ are processed iteratively. When all the nodes in $\mathcal{G}_{\textit{unaffected}}$ reach the fixed point, the nodes in $\mathcal{G}_{\textit{add\_only}}$ are then processed iteratively to reach the fixed point. Finally, the nodes in $\mathcal{G}_{\textit{delete\&change}}$ are processed. 
What we need to prove is that the convergent state of each sub-CFG is consistent with that achieved by our optimized analysis. 
As for $\mathcal{G}_{\textit{unaffected}}$, since its convergent state is independent of any other parts, it is apparent that the convergent state produced by the above three-phase procedure is identical to that by our optimized analysis.   
For $\mathcal{G}_{\textit{add\_only}}$, the consistency between the results of two analyses is not straightforward. In the optimized incremental analysis, each node of $\mathcal{G}_{\textit{add\_only}}$ takes its previous convergent result as the initial value, whereas the baseline takes $\bot$ as the initial value. 
We discuss the formal proof of consistency on $\mathcal{G}_{\textit{add\_only}}$ shortly in Theorem \ref{theorem:consistency-gaddonly}.
Last, for $\mathcal{G}_{\textit{delete\&change}}$, since the results of two analyses on both $\mathcal{G}_{\textit{unaffected}}$ and $\mathcal{G}_{\textit{add\_only}}$ are identical, the dataflow facts propagated to $\mathcal{G}_{\textit{delete\&change}}$ are totally same. 
We can conclude that the convergent results on $\mathcal{G}_{\textit{delete\&change}}$ computed by the baseline and our optimized incremental analysis are consistent.
Putting it all together, we prove that our optimized incremental analysis produces the same results as the baseline. 




\begin{theorem}[\textbf{Consistency on $\mathcal{G}_{\textit{add\_only}}$}]
\label{theorem:consistency-gaddonly}
Given the sub-CFG $\mathcal{G}_{\textit{add\_only}}$, assume that the respective sub-CFG before the change which consists of the same set of nodes is $\mathcal{G}_{\textit{add\_only}}^*$.
We define the analysis functionality of one iteration on $\mathcal{G}^*_{\textit{add\_only}}$ and $\mathcal{G}_{\textit{add\_only}}$ as $\mathcal{F}^*$ and $\mathcal{F}$, respectively.
The dataflow analysis with its previous convergent results as the initial value produces the same results as that with $\bot$ as the initial value, which shows as the following equation:
\[\mathcal{F}^{\infty}(\mathcal{F}^{*\infty}(\widehat{\bot})) \equiv \mathcal{F}^{\infty}(\widehat{\bot})\]
\end{theorem}

\begin{proof}
There are two cases of monotone dataflow analysis, \ie, increasing analysis with the join operator $\sqcup$ and decreasing analysis with the meet operator $\sqcap$.
Without loss of generality, we take increasing analysis as an example. We will prove the consistency of the two analyses as their convergent state partial to each other.
First, we prove that the convergent state with $\bot$ as the initial value, is partial to that with the previous convergent state as the initial value.

In increasing analysis, the incoming fact of each node is initialized as $\bot$. We use $\widehat{\bot}$ to represent the overall initialized incoming facts of nodes across the sub-CFG. Then, the dataflow analysis on the sub-CFG $\mathcal{G}^*_{\textit{add\_only}}$ can be formulated as starts with $\widehat{\bot}$. As $\bot$ is the bottom of the lattice, the beginning state of the analysis is partial to the subsequent state of the analysis on $\mathcal{G}^*_{\textit{add\_only}}$. This relation can be described as following :

\begin{align}\label{eq:ie11}
&\!\!\!\!\mathcal{F}^*(\widehat{\bot}) \geq \widehat{\bot}
\end{align}

We apply infinite times of $\mathcal{F}^*$ to get the convergent state calculated at the end of the analysis. The following inequality \ref{eq:ie12} can be obtained:

\begin{align}\label{eq:ie12}
&\!\!\!\!\mathcal{F}^{*\infty}(\widehat{\bot}) \geq \widehat{\bot}
\end{align}

The left side of the above inequality corresponds to the previous convergent results, \ie, the initial state of our optimized incremental analysis. And the right side is $\bot$, \ie, the initial state of the baseline. 
Based on the two sides, we perform the analysis iteratively to get the convergent state on $\mathcal{G}_{\textit{add\_only}}$, and imply that:

\begin{align}\label{eq:ie13}
&\!\!\!\!\mathcal{F}^{\infty}(\mathcal{F}^{*\infty}(\widehat{\bot})) \geq \mathcal{F}^{\infty}(\widehat{\bot})
\end{align}

Thus, the convergent state achieved starting from $\bot$ is partial to that achieved by starting from the previous convergent results. 

Next, we prove that the convergent state produced by starting from the previous convergent results is partial to that with $\bot$ as the initial value. 
The analyses on $\mathcal{G}^*_{\textit{add\_only}}$ and $\mathcal{G}_{\textit{add\_only}}$ both starts with $\widehat{\bot}$. As discussed in \S\ref{subsec:incremental-opt}, the incoming and outgoing facts associated with each node in $\mathcal{G}^*_{\textit{add\_only}}$ and $\mathcal{G}_{\textit{add\_only}}$ satisfy the incremental property (\ie, Theorem \ref{theorem:incremental-property}). Therefore, the following inequality holds: 

\begin{align}\label{eq:ie144}
&\!\!\!\!\mathcal{F}^*(\widehat{\bot}) \leq \mathcal{F}(\widehat{\bot})
\end{align}

Both computations are conducted infinite times, we thus get the following inequality.

\begin{align}\label{eq:ie14}
&\!\!\!\!\mathcal{F}^{*\infty}(\widehat{\bot}) \leq \mathcal{F}^{\infty}(\widehat{\bot})
\end{align}

Next, we take both sides of it as input to perform the computation $\mathcal{F}$ iteratively to reach the convergent state. And the following inequality  can be deduced: 

\begin{align}\label{eq:ie15}
&\!\!\!\!\mathcal{F}^{\infty}(\mathcal{F}^{*\infty}(\widehat{\bot})) \leq \mathcal{F}^{\infty}(\mathcal{F}^{\infty}(\widehat{\bot})) 
\end{align}

By reducing the right side of the inequality, we can get:

\begin{align}\label{eq:eq16}
&\!\!\!\!\mathcal{F}^{\infty}(\mathcal{F}^{*\infty}(\widehat{\bot})) \leq \mathcal{F}^{\infty}(\widehat{\bot})
\end{align}

Thus, the convergent state achieved starting from the previous convergent results is partial to that achieved by starting from $\bot$.  
All in all, we prove that $\mathcal{F}^{\infty}(\mathcal{F}^{*\infty}(\widehat{\bot})) \equiv \mathcal{F}^{\infty}(\widehat{\bot})$.
\end{proof}

}

\section{Implementation}
\label{sec:implementation}

We implemented \tool by following the distributed worklist algorithm  on top of Apache Giraph 1.4.0\footnote{\url{https://giraph.apache.org/}}, a well-maintained open source Java implementation of Pregel \cite{giraph-book,giraph-vldb}. 
Giraph replicates Pregel's concepts and adds several new features to this model, including master computation, out-of-core computation, and sharded aggregators, etc. 
In particular, Giraph first divides the input graph into a number of partitions based on Hadoop distributed file system. 
Within each superstep of the BSP model, Giraph launches multiple workers and enables each worker to process a partition separately in a distributed way. 
Giraph offers multiple effective partitioning schemes, which \tool directly adopts to achieve good workload balance and scalability. 
Besides, \tool leverages two extra options offered by Giraph to realize the pulled-based worklist algorithm. 
(1) BasicComputation Class. BasicComputation is a general option for performing computations in Giraph. 
It can be used to access the graph’s information, such as the superstep ID and information of vertices and edges.
We extend it to distinguish analysis phase and acquire edge information in the implementation of \tool. 
(2) Broadcast Class. Broadcast is the simplest way for master node to communicate with worker nodes in the scope of the entire cluster, ensuring that all vertices access the same information. 
\tool exploits this feature to broadcast workers of entry nodes in CFG.


\change{
In the incremental analysis discussed in \S\ref{sec:design-incre}, we additionally employ Redis\footnote{\url{https://redis.io/}}, an open source, in-memory, NoSQL key-value database to achieve efficient, scalable, and robust query of analysis results. 
As one of the most popular in-memory databases,  Redis offers low-latency read and write operations for large-scale data. 
To survive system failure and network bottlenecks, Redis provides persistence of datasets via snapshotting, journaling and replication.  
All in all, Redis meets the requirement of incremental dataflow analysis where a vast amount of dataflow facts need to be maintained, queried and updated in an efficient, scalable, and reliable way.
To be specific, we make use of the following Redis features for better resource utilization and query efficiency.
(1) Connection Pooling. Connection Pooling is a convenient Redis client library that offers connection management optimization. By utilizing it, \toolincre  avoids the overhead of establishing and releasing new connections for each request when reusing or updating a large number of dataflow facts.
(2) Pipelining. Pipelining allows users to batch multiple commands into a single request without waiting for the response to each individual command. We adopt it to merge multiple read or write requests of dataflow fact into a single request, thereby greatly reducing the overall response and transfer time. (3) Proxy. Proxy is responsible for routing commands, balancing loads, and dealing with failovers. It simplifies the client's query logic without in-depth understanding of the underlying communication and fault tolerance mechanisms. We use it to efficiently and reliably access a large number of dataflow facts in the Redis cluster.
}


\section{Programming Model}
\label{sec:model}




\tool as a framework supporting the general interprocedural dataflow analysis, provides a set of necessary APIs to users. 
Users readily implement a particular client analysis based on these APIs by specifying the information of input CFG, the dataflow equations (\ie, merge, transfer), and the propagation logic. 
In the following, we first discuss the crucial APIs provided by \tool, then demonstrate how to implement a client analysis based on the APIs.

\subsection{APIs\label{subsec:bigdataflow-api}}

Given a control flow graph or other sparse representation\cite{Hardekopf:2011-fpa-cgo,Ramalingam:2002-ser-tcs}, \tool takes it as input and constructs the graph in memory. 
During a dataflow analysis, each vertex in the CFG maintains a dataflow fact, as well as the program statements associated.
Lines \ref{vertex-start}-\ref{vertex-end} in Listing \ref{lst:apis} show the abstract class of  \texttt{VertexAttribute}, which defines two members: dataflow fact of abstract class \texttt{Fact} and statements of class \texttt{Stmts}. 
Dataflow fact describes the dataflow information computed at each program point during analysis.
The abstract class \texttt{Fact} (Line \ref{fact}) leaves users the interface for specifying a particular type of dataflow fact in a client analysis. 
\texttt{Stmts} (Lines \ref{stmts-start}-\ref{stmts-end}) describes the set of statements associated with the vertex, which determines the logic of transfer functions.
In a statement-level dataflow analysis, dataflow fact is associated with each statement, where an instance of \texttt{Stmts} contains one single statement.
While in a basic block-level analysis, each instance of \texttt{Stmts} indicates a set of statements in a basic block. 

\lstset{escapeinside={(*@}{@*)}}
\begin{lstlisting} [language=Java,label={lst:apis},caption=The APIs.]
abstract class VertexAttribute{(*@\label{vertex-start}@*) 
   Fact fact;
   Stmts stmts;
}(*@\label{vertex-end}@*) 

abstract class Fact{} (*@\label{fact}@*) 

abstract class Stmt {} (*@\label{stmt}@*) 
class Stmts { (*@\label{stmts-start}@*) 
   Stmt[] stmts;
} (*@\label{stmts-end}@*) 
 
interface Analysis { (*@\label{analysis-start}@*) 
   Fact merge(Set<Fact> predFacts, Fact oldIN);
   Fact transfer(Stmts stmts, Fact inFact);
   boolean propagate(Fact oldFact, Fact newFact);
} (*@\label{analysis-end}@*) 

interface IncrementalAnalysis { (*@\label{inc-analysis-start}@*) 
   Fact merge(Set<Fact> predFacts, Fact oldIN);
   Fact transfer(Stmts stmts, Fact inFact);
   boolean propagate(Fact oldFact, Fact newFact);
} (*@\label{inc-analysis-end}@*) 
\end{lstlisting}

\change{
Besides the above crucial data structures, three necessary components of dataflow analysis are defined in the \texttt{Analysis} and  \texttt{IncrementalAnalysis} interfaces shown as Lines \ref{analysis-start}-\ref{analysis-end} and \ref{inc-analysis-start}-\ref{inc-analysis-end} in Listing \ref{lst:apis}.
Whenever the computation on a vertex $k$ is launched, \texttt{merge()} is first invoked to take the newly updated dataflow facts of predecessors together with the old incoming fact, and produce a new incoming dataflow fact for $k$. 
In general, the merge operation can be union or intersection depending on the specific client analysis. 
Users override \texttt{merge()} to specify the exact logic. 
Taking the incoming dataflow fact produced by \texttt{merge()} and the statements as input, \texttt{transfer()} computes the outgoing dataflow fact accordingly. 
Users are required to specify the particular transformation logic by overriding \texttt{transfer()} for a particular client analysis. 
\texttt{propagate()} describes the conditions for propagating dataflow facts to successors.
Usually, propagation is decided by the comparison between old fact and new fact.
User overrides \texttt{propagate()} to define concrete termination condition. 
}

\subsection{An Example of Alias Analysis \label{subsec:model-example}}
We use a context- and flow-sensitive alias analysis as an example to illustrate how to use the APIs to implement a client analysis.
Flow-sensitive alias analysis computes the alias relations between pointer variables at each program point. 
As a fundamental analysis, it has been widely used in various applications including bug detection, security enforcement, optimizations, etc.



We adopt function cloning to achieve context-sensitivity \cite{Emami:1994-clone-pldi,zuo2021chianina}.
The input CFG to \tool actually corresponds to a cloned interprocedural CFG. 
Taking the inlined ICFG as input, we first define a particular subclass \texttt{AliasStmt} to instantiate each statement for alias analysis. Its detailed implementation is omitted due to space limit.   
\texttt{Stmts} has only one \texttt{Stmt} instance as we would like to analyze the alias information at the granularity of statement. 
Here we adopt the program expression graph (PEG) \cite{zheng2008demand} as a dataflow fact to represent the alias information at each program point. 
As such, each object of \texttt{Fact} is instantiated as a \texttt{PEG} instance.
Next, \texttt{merge()} is achieved as union of the updated PEGs from predecessors with the old incoming fact. 
Within the overridden \texttt{transfer()}, edge addition and/or deletion are performed on PEG according to the semantics of each type of statement.
If the old PEG and newly updated PEG are isomorphic, \texttt{propagate()} returns false and the vertex becomes inactive.

\begin{lstlisting} [language=Java,label={lst:alias-instance},caption=The implementation of flow-sensitive alias analysis on top of \tool.]
public class AliasStmt extends Stmt {...}
class AliasVertexAttribute extends VertexAttribute
{
   super();
   fact  = new PEG();
}

class AliasAnalysis implements Analysis {
   Fact merge(Set<Fact> predFacts, Fact oldIN) {
      PEG peg = (PEG)oldIn;
      for (Fact item : predFacts) {
         if (item == null)  continue;
         PEG prePEG = (PEG)item;
         peg.merge(prePEG);
      }
      return peg;
   }
   Fact transfer(Stmts stmts, Fact fact) {
      PEG peg = (PEG)fact;
      switch (stmts[0].getType()) {
         case Load:
           transfer_load(peg,(AliasStmt)stmts[0]);
           break;
         //...
      }
      return peg;
   }
   boolean propagate(Fact oldFact, Fact newFact) {
      if(oldFact == null)  return true; 
      PEG newPEG = (PEG)newFact;
      PEG oldPEG = (PEG)oldFact;
      return !newPEG.consistent(oldPEG);
   }
}
\end{lstlisting}

As can be seen, to implement a client analysis on top of \tool, users only need to specify the necessary functionalities specific to client analysis, without worrying about any implementation details of the underlying worklist algorithm as well as other system-side optimizations.

\section{Evaluation}
\label{sec:evaluation}

Our evaluation focuses on the following three questions:

\begin{itemize}
\item Q1: What is the overall performance of \tool given a rich set of distributed computing resources? (\cref{subsec:performance})
\item Q2: How does \tool perform compared with other competitive analysis systems/tools? (\cref{subsec:evaluate-compare})
\item Q3: What about the performance of \tool given the varying numbers of cores and resources? (\cref{subsec:scalability})
\item \change{Q4: How about the performance of \tool in the mode of incremental analysis (\cref{subsec:eval-increment})?} 
\end{itemize}

\begin{table}[h]
	\caption{Characteristics of subject programs.}
	\label{tab:subjects}
	\centering
	\vspace{-0.5em}
	\scalebox{0.84}{
		\begin{tabular}{l|r|r|r|r}
			\toprule
			Subject  & Version & \#LoC  & \#Functions & Description\\
			\midrule
			\textsf{Linux}  & 5.2 & 17.5M & 565K & operating system \\
			\textsf{Firefox}  & 67.0  & 7.9M &  770K & web browser \\
			\textsf{PostgreSQL}  & 12.2  & 1.0M & 30K & database system \\
			\textsf{OpenSSL}   & 1.1.1 & 519K & 12K & TLS protocol \\
			\textsf{Httpd}  &  2.4.39 & 196K & 6K & web server \\
		\end{tabular}
	}
\end{table}


\MyPara{Subjects.}
\textcolor{black}{To measure the performance of \tool on scaling large programs, we selected five real-world software as the experimental subjects, including Linux kernel, Firefox, PostgreSQL, OpenSSL, and Apache Httpd. 
Table \ref{tab:subjects} lists detailed information about the subjects, such as the version (Version), the number of lines of code (\#LoC), the number of functions (\#Functions), and its description.}

\MyPara{Reference Tools.}\change{
To validate the advantage of \tool in terms of performance and scalability on large-scale programs, we selected the existing parallel/distributed analysis systems/tools as the competitors. 
For parallel algorithms, we chose \toolchianina \cite{zuo2021chianina}, the most recent and state-of-the-art parallel system scaling context - and flow-sensitive analysis to large-scale C programs. \toolchianina is implemented in C/C++, and leverages two-level parallel computation model and out-of-core disk support to achieve both analysis efficiency and scalability.
We ignore other sequential analysis algorithms \cite{Hardekopf:2011-fpa-cgo,sui-2016-svf-cc} since it has been validated that \toolchianina outperforms them \cite{zuo2021chianina}. 
For distributed work, since there exist no distributed systems supporting dataflow analysis, we used \toolnaive, the version implemented based on the distributed classic worklist algorithm shown as Algorithm \ref{a:naive-algo} as the reference tool. 
By default, \tool is implemented using the optimized version (\ie, Algorithm \ref{a:opt-algo}).
For incremental dataflow analysis, since no existing distributed incremental analysis tool is available for comparison, we compare it with our whole-program analysis mode, \ie, \toolwho vs. \toolinc.
}

\MyPara{Hardware and Software Settings.}\change{
All experiments were conducted in the Alibaba Cloud environment. \toolwho, \toolinc and \toolnaive are all deployed on a cluster consisting of 125 Elastic Compute Service (ECS)\footnote{\url{https://www.alibabacloud.com/product/ecs}} nodes with Alibaba Elastic MapReduce (EMR) installed. Each node (in particular \emph{ecs.r7.2xlarge}) is equipped with 8 virtual CPU cores based on Intel Xeon Scalable processors and 64GB memory, running CentOS 7.4. The adopted EMR version is 3.14.0 corresponding to Hadoop 2.7.2, Giraph 1.4.0, and Redis 5.0.0. 
To compare with \toolchianina which can only run on a single-machine with shared memory, we used the most powerful server node available in the US (Virginia) region, \ie, \emph{ecs.r6.26xlarge} with 104 virtual cores,  768G memory, and 1T SSD-backed cloud disk. 
}

\MyPara{Client Analyses.}\change{
In the experiments, we implemented two client analyses, namely context-sensitive flow-sensitive alias analysis and instruction cache analysis, on top of \toolwho, \toolinc, \toolnaive, and \toolchianina. 
The alias analysis is same as the example discussed in \cref{subsec:model-example}.
For cache analysis, we followed the abstract model of LRU caches in \cite{DBLP:journals/lites/LvGRW016} that adopts the set-associative organization. The configuration is set as 512 cache lines with LRU replacement strategy enabled. The analysis computes a cache model at each program point and decides a cache hit or miss.
We chose the above two analyses for several reasons: 1) both analyses are fundamental and widely-used; 2) they are expensive and hardly scalable given their memory-intensive dataflow fact and compute-intensive transfer function; 3) they fall into the two cases of the accumulative property in \cref{subsec:proof} respectively, thereby validating the proof more comprehensively.
}

The context-sensitivity is achieved via fully function cloning (\ie, $\infty$-CFA). 
We start the cloning based upon a call graph constructed by using a lightweight inclusion-based context-insensitive pointer analysis with support for function pointers. 
To handle recursion, we first identify the strongly connected components (SCCs) over the pre-computed call graph. Functions not in any SCC enjoy full context sensitivity. Whereas, level-2 call-string sensitivity (\ie, using 2 top-most callsites as the distinguishing context) is used for those within SCCs.
Note that function cloning is NOT the core contribution of this work. Users can adopt the classical k-limited context-sensitivity or other selective context-sensitivity techniques \cite{jeon-ata-oopsla18,li-pa-oopsla18}. This can be done by launching a cheap pre-analysis to understand the contexts desired, and then performing selective function cloning. 

For each client analysis, the version implemented on top of \toolwho, \toolinc, \toolnaive and \toolchianina are identical and possess the same analysis precision. We checked the analysis results of four tools and validated they are consistent. Specifically, we compared the total number of alias pairs (including both memory alias and value alias) generated for alias analysis, and the total number of potentially cached memory blocks for cache analysis. The columns \#PAliases and \#BCached in Table \ref{tab:performance} list the exact numbers.

\subsection{Performance of Whole-Program Analysis}

\subsubsection{Overall Performance}
\label{subsec:performance}
Tables \ref{tab:performance-alias} and \ref{tab:performance-cache} demonstrate the performance of \tool in the mode of whole-program analysis when analyzing the five real-world subjects.
Columns \#Workers, \#PMem, and Time indicate the number of workers used (one worker corresponding to one physical core), the amount of peak memory consumed, and the total analysis time, respectively.

\textcolor{black}{It is well known that the complexity of a particular dataflow analysis is heavily dependent on many factors, such as the size, density, structure of the control flow graph,  and the semantics of program under analysis. Thereby, it is difficult to give a general formula that can figure out the ideal number of workers needed. What we can do is to estimate a number as small as possible so as to the analysis task can be completed successfully and efficiently. To this end, we first run a small sample of the analysis (\ie, 1/50 of the input graph) on a small test cluster with 10 nodes. Based on the resource utilization data monitored, we estimate an initial number roughly. Next, we run the analysis on the initial number of workers. If the task fails due to insufficient memory, the number of workers is doubled until the analysis can succeed.}

\begin{table}[htb!]
	\caption{{Overall performance: columns \textbf{\#PAliases} and \textbf{\#BCached} indicate the number of alias pairs and the number of potentially cached memory blocks;  columns \textbf{\#Workers}, \textbf{\#PMem}, \textbf{Time} and \textbf{Cost} represent the number of workers used, the size of peak memory consumed, the total analysis time, and the rental cost of cloud resources, respectively;  \textbf{\#Part.} indicates the number of partitions; - indicates out-of-memory error; (a) and (b) report the results for alias and cache analysis, respectively.}}
	\label{tab:performance}
	\centering
\subfloat[Alias Analysis]{\label{tab:performance-alias}
	\scalebox{0.64}{
		\begin{tabular}{lrr|r|r|rrr|r|r|rrr|r|r|r}
			\toprule
			& & \multicolumn{4}{c}{\tool} && \multicolumn{4}{c}{\toolnaive} && \multicolumn{4}{c}{\toolchianina} \\ \cmidrule{3-6} \cmidrule{8-11} \cmidrule{13-16}
			Subject  & \#PAliases & \#Workers & \#PMem & Time & Cost && \#Workers & \#PMem & Time & Cost && \#Part. & \#PMem & Time & Cost  \\
		\midrule
			\textsf{Linux} &  12.5B & 350 & 3.5T & 16.7mins & \$11.1 && 350 &  - & - & - && 4 & 453.4G & 17.4hrs & \$110.4 \\
			\textsf{Firefox}  & 11.5B & 140 & 1.2T & 16.5mins & \$4.4 && 140 &  - & - & - && 4 &  131.6G & 5.3hrs & \$33.6 \\
			\textsf{PostgreSQL} & 727.0M & 50 & 329.7G & 2.8mins & \$0.3 && 50  & 330.7G & 4.9mins & \$0.5 && 1  & 61.9G & 50.4mins & \$5.3 \\
			\textsf{OpenSSL}  &  734.8M & 30 & 285.3G & 3.5mins & \$0.2 && 30 &  329.8G & 6.8mins & \$0.4 && 1  & 43.2G & 35.4mins & \$3.7  \\
			\textsf{Httpd}   &  183.1M & 10 & 119.9G & 2.8mins & \$0.1 && 10 &  137.9G & 4.0mins & \$0.1 && 1  & 14.2G & 11.2mins & \$1.2  \\[-.5em]
		\end{tabular}
	}
}

\subfloat[Instruction Cache Analysis]{\label{tab:performance-cache}
	\scalebox{0.64}{
		\begin{tabular}{lrr|r|r|rrr|r|r|rrr|r|r|r}
			\toprule
			& & \multicolumn{4}{c}{\tool} && \multicolumn{4}{c}{\toolnaive} && \multicolumn{4}{c}{\toolchianina} \\ \cmidrule{3-6} \cmidrule{8-11} \cmidrule{13-16}
			Subject  & \#BCached & \#Workers & \#PMem & Time & Cost && \#Workers & \#PMem & Time & Cost && \#Part. & \#PMem & Time & Cost  \\
	       \midrule
			\textsf{Linux} & 21.5B & 500 & 5.6T & 44.4mins & \$42.0 && 500 &  - & - & - && 4 & 555.4G & 9.4hrs & \$59.6 \\
			\textsf{Firefox} & 15.8B & 400 & 4.4T & 39.0mins & \$29.5 && 400 &  - & - & - &&  4  & 351.5G & 7.2hrs &  \$45.7  \\
			\textsf{PostgreSQL} & 1.4B & 180  & 1.1T & 3.2mins & \$1.1 && 180  & 1.1T & 6.5mins & \$2.2  && 1  & 115.3G & 38.1mins & \$4.0 \\
			\textsf{OpenSSL} & 2.8B & 180 & 1.3T & 6.9mins & \$2.3 && 180  & 1.5T & 13.4mins & \$4.6 &&  1 &  227.6G & 1.7hrs   & \$10.8 \\
			\textsf{Httpd}  & 782.0M & 100  & 684.3G & 3.0mins & \$0.6 && 100 & 781.3G & 4.6mins & \$0.9 &&  1 &  58.3G & 18.5mins  & \$2.0 \\[-2.em] 
		\end{tabular}
	}
}
\end{table}


As can be seen, the peak memory consumed in both alias analysis and instruction cache analysis can easily reach several terabytes for large-scale programs, such as the Linux kernel and Firefox, due to the memory-intensive dataflow fact and the huge number of program points. 
Even for the smallest subject Httpd, performing the context- and flow-sensitive analysis takes more than a hundred or even several hundreds of gigabytes.
This is consistent with the claim in \cite{Aiken:2007-paste-saturn} that memory would be the major bottleneck for analysis to scale to large programs.
By leveraging the enormous amount of memory and computing resources in a cloud environment, \tool manages to analyze all the subjects successfully and efficiently. 
The alias analysis can be completed within 20 minutes for all subjects; the more expensive cache analysis takes less than 45 minutes for the Linux kernel with 500 workers.

\subsubsection{Comparison with Other Frameworks\label{subsec:evaluate-compare}}
Given the identical version of the client analysis implemented, we compared \tool against \toolnaive and \toolchianina with respect to performance and cost.
Columns under \toolnaive and \toolchianina in Table \ref{tab:performance-alias} and \ref{tab:performance-cache} show the detailed results of \toolnaive and \toolchianina, respectively.

\MyPara{\toolchianina.}
As \toolchianina can only run on a single-machine with shared memory, we rented the most powerful server node with 104 virtual cores, 768G memory, and 1T SSD available in the US (Virginia) region of Alibaba Cloud. 
In terms of analysis time, \toolchianina with 104 threads takes more than 17 hours and 9 hours to finish alias analysis and cache analysis over Linux. While \tool completes alias and cache analysis within 20 and 45 minutes under a cluster, respectively. 
It shows that distributed parallelism enabled by \tool indeed accelerates the analysis significantly (up to 62x and 12x for alias analysis and cache analysis on Linux, respectively).
Note that \tool takes more time for cache analysis than alias analysis on all the subjects, whereas \toolchianina does not.  This can be explained from two aspects. First, cache analysis is more memory-intensive than alias analysis. The cache analysis on \tool implemented in Java deservedly pays more GC time. 
Second, as observed, the alias analysis running on \toolchianina has low CPU utility due to load imbalance and excessive thread-switching costs for certain subjects (\eg, Linux) when a large number of threads are enabled on a single machine.

As the computing resources used by \tool and \toolchianina are different, we cannot simply derive that \tool outperforms \toolchianina.
For the sake of fairness, we measured the exact amount of rental costs of cloud resources in dollars paid by \tool and \toolchianina for completing the identical analysis. 
\textcolor{black}{As cloud providers generally adopt a unified pricing strategy, there is little difference in the price of nodes with similar resources across different providers.
Without loss of generality, we calculated the cost by multiplying the analysis time and the official pay-as-you-go hourly price of Alibaba Cloud in US (Virginia) region\footnote{\url{https://www.alibabacloud.com/zh/product/ecs-pricing-list/en}}. }
In particular,  at the time of submission, each node \emph{ecs.r7.2xlarge} used by \tool takes \$0.454/hour. The price of the entire cluster is 0.454*125, \ie, \$56.75/hour.  
The single \emph{ecs.r6.26xlarge} server node used by \toolchianina takes \$6.344/hour.
The cost columns in Table \ref{tab:performance} show the detailed results. 
As can be seen, \tool spends lower rental costs than \toolchianina over all the subjects except Httpd. 
Although the price of the cluster used by \tool (\$56.75/hour) is much higher than that of the single server used by \toolchianina (\$6.344/hour), \tool takes much less time to finish the analysis than \toolchianina. 
We can thus conclude that \tool is able to offer significantly higher analysis efficiency for large-scale programs, while taking fewer costs compared to \toolchianina.

Regarding memory consumption, \tool apparently consumes much more memory than \toolchianina. 
There are several reasons. (1) \toolchianina is a disk-based system where the memory consumption is strongly restricted. It will leverage disks to maintain the huge amount of data once the memory consumption exceeds a certain threshold. In contrast, \tool prefers utilizing the memory on each node to perform communications and accelerate the analysis. (2) \tool is implemented in Java, while \toolchianina is implemented in C/C++. No doubt \toolchianina would have less memory footprint than \tool. (3) \tool is running on top of Giraph. To achieve fault tolerance, Giraph needs to maintain extra (\eg, 3) replicas for all the data stored. Moreover, for certain global data used in the analysis, \tool has to broadcast it on every node, leading to extra memory consumption.

\MyPara{\toolnaive.}
As numerous redundant and expensive dataflow facts were transmitted in the network and gathered at each vertex, \toolnaive failed to analyze the large-scale subjects in our experiments (\ie, Linux and Firefox) given the same computing resources as \tool.
It validates that \tool does save memory resources, thus offering better scalability than \toolnaive. 
For the analyses which both \toolnaive and \tool successfully complete, \tool exclusively outperforms \toolnaive in terms of time efficiency. 
This is because \toolnaive requires more data transferred and merged than \tool to accomplish the same analysis. 

\begin{figure}[htb!]
\centering
\subfloat[time]{      
\label{fig:alias-sql-time}
\begin{minipage}[c]{.34\linewidth}
\centering
\includegraphics[width=1\textwidth]{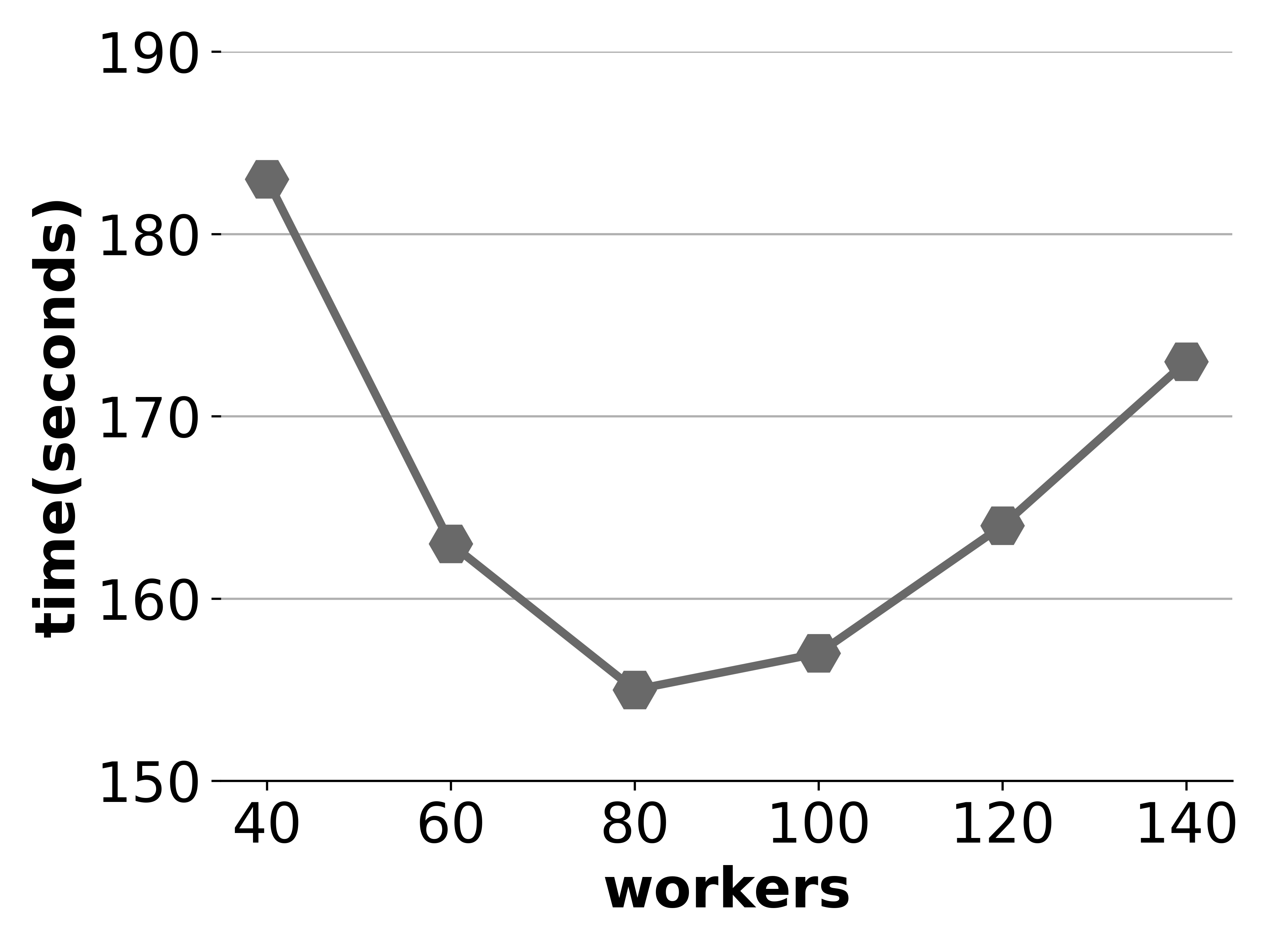}
 \end{minipage}
}
\hspace{2em}
\subfloat[peak memory]{         
\label{fig:alias-sql-memory}
\begin{minipage}[c]{.34\linewidth}
\centering
\includegraphics[width=1\textwidth]{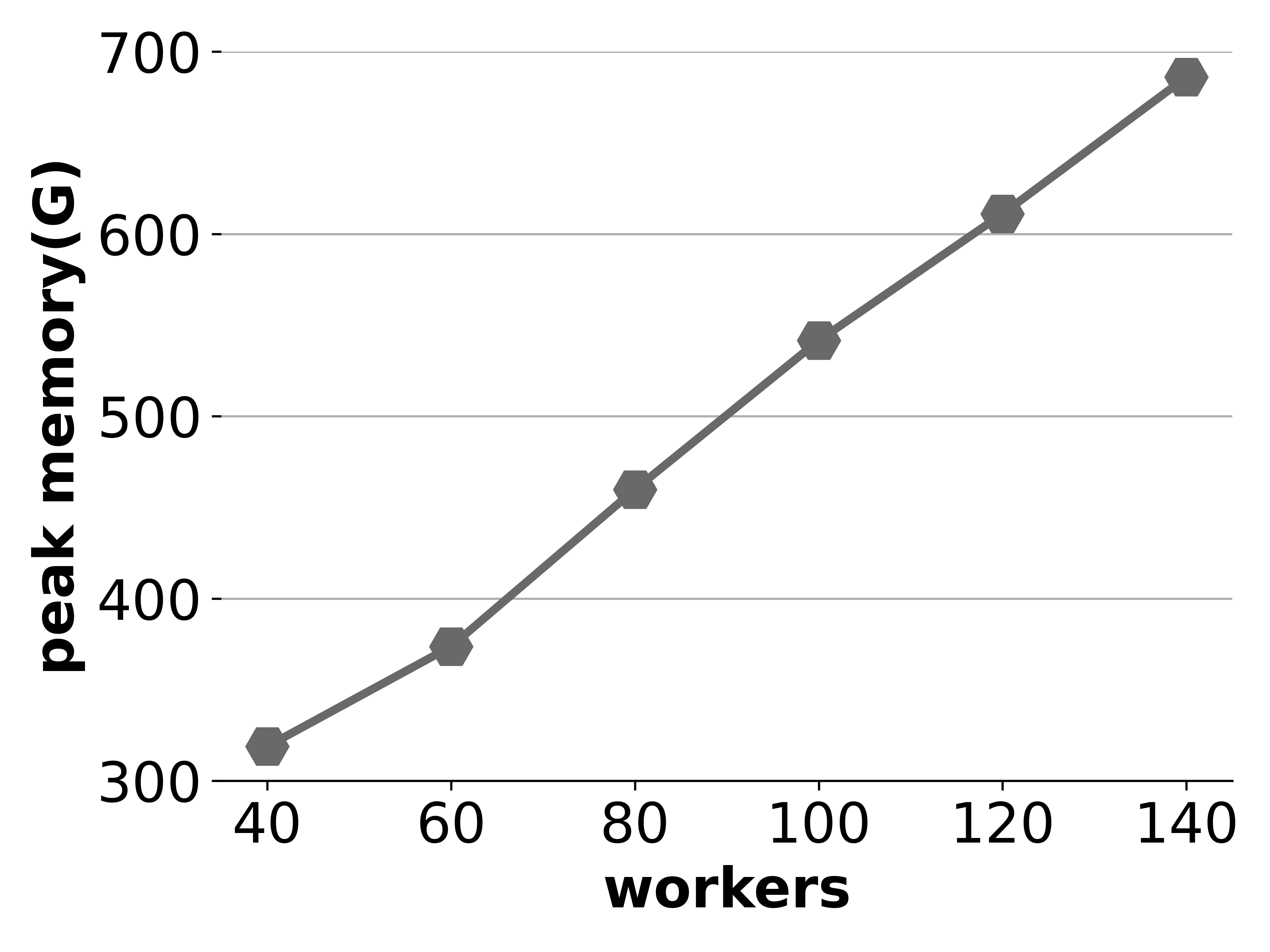}
\end{minipage}
}
\caption{The time (a) and peak memory (b) used for alias analysis on OpenSSL with varying number of workers.\label{fig:alias-sql}}
\end{figure}

\begin{figure}[htb!]
\centering
\subfloat[time]{      
\label{fig:cache-sql-time}
\begin{minipage}[c]{.34\linewidth}
\centering
\includegraphics[width=1\textwidth]{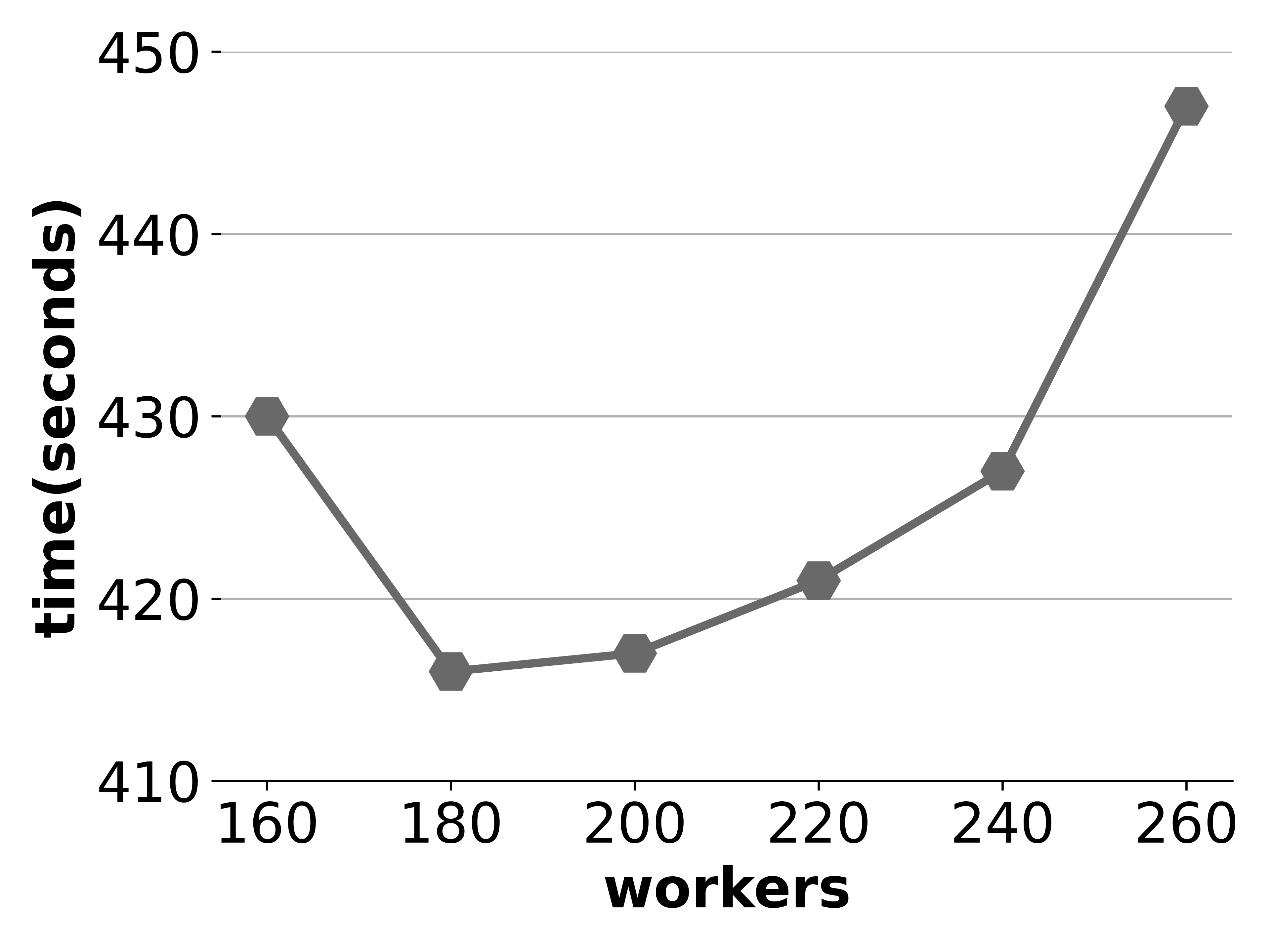}
 \end{minipage}
}
\hspace{2em}
\subfloat[peak memory]{         
\label{fig:cache-sql-memory}
\begin{minipage}[c]{.34\linewidth}
\centering
\includegraphics[width=1\textwidth]{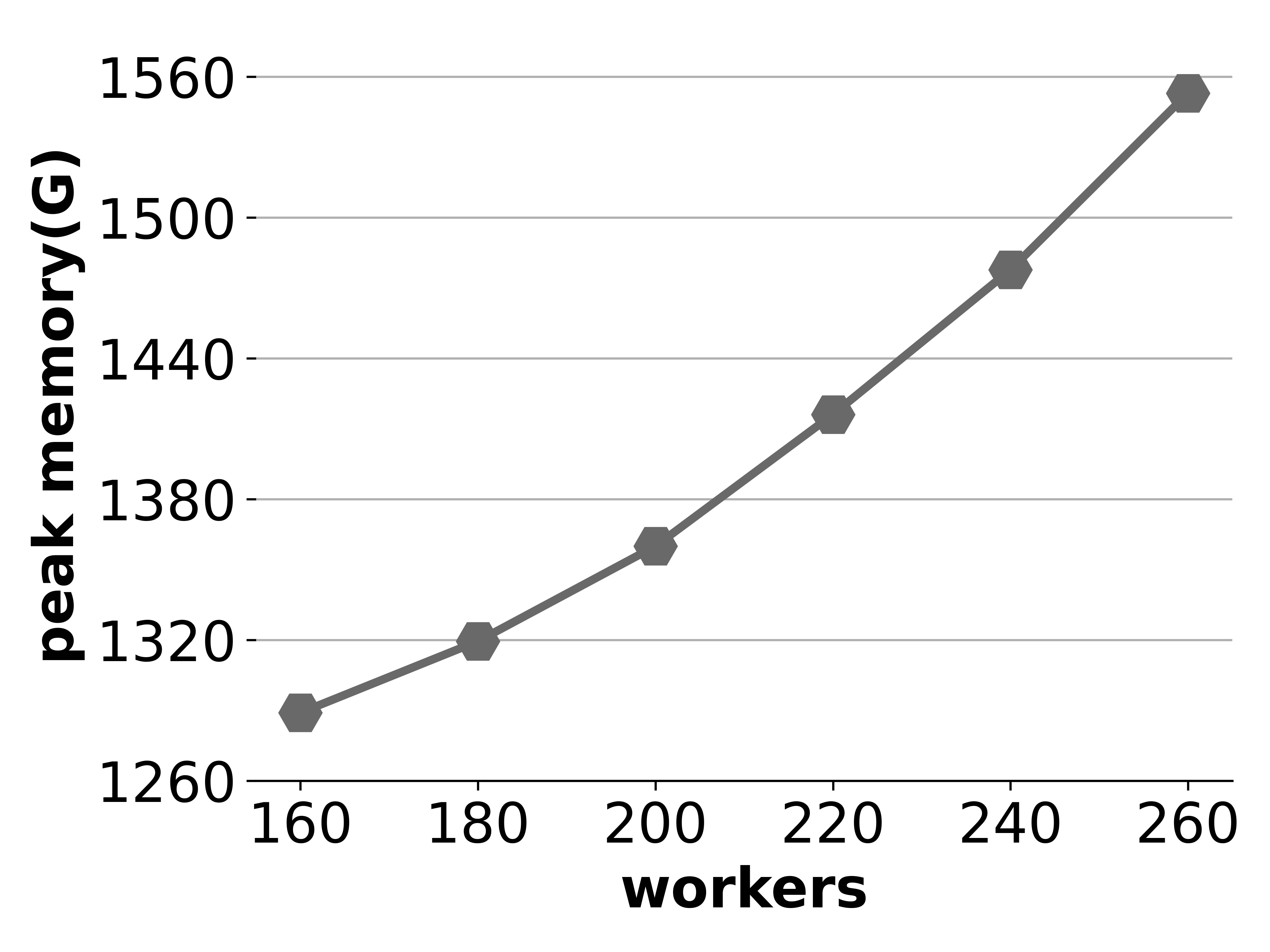}
\end{minipage}
}
\caption{The time (a) and peak memory (b) used for cache analysis on OpenSSL with varying number of workers.\label{fig:cache-sql}}
\end{figure}

\subsubsection{Scalability\label{subsec:scalability}}
To understand the scalability of \tool, we measured the analysis time in seconds and peak memory consumption in gigabytes for both alias analysis and cache analysis given different numbers of workers. 
Figures \ref{fig:alias-sql} and \ref{fig:cache-sql} show the detailed performance results of alias analysis and cache analysis on OpenSSL, respectively, where the x-axis indicates the number of workers, and y-axis represents the time or peak memory used. 
Here only the data of PostgreSQL is reported. 
Other subjects show a similar trend to that of OpenSSL. 

For alias analysis, the time taken by \tool follows a V-bottom pattern shown as Figure \ref{fig:alias-sql-time}. 
When less workers are available (\ie, 40), the total memory capacity just satisfies the analysis need. 
With the number of workers increasing from 40 to 80, increased parallelism is translated to higher performance. Therefore, the overall running time shows a descending trend. 
However, the communication cost among workers is monotonically increased with the growth of workers involved. 
Once the performance benefit of parallelism is no longer superior to the increased communication cost among workers, time climbs steadily.
As such, for the specific analysis, having 80 workers provides the best trade-off between parallelism benefit and communication cost, leading to the shortest running time of all the tested parallel schedules.
It implies that in practice we can seek a sweet spot of parallelism for different subjects according to the tendency of running time as the number of workers changes.
This is particularly meaningful because 1) cloud resources are on demand and charged on actual usage; and 2) performing dataflow analysis on the same large-scale program could be an iterative process as the program evolves constantly.
In terms of the peak memory usage, as more threads/processes consume more memory space, it is not surprising that it shows an ascending trend with the growth of workers in Figure \ref{fig:alias-sql-memory}.
Figure \ref{fig:cache-sql} shows similar trends for cache analysis. 
As can be read from Figure \ref{fig:cache-sql-time}, with 160 workers available, \tool successfully finishes the cache analysis. The best performance is achieved given more workers (\ie, 180). After that point, more analysis time is needed with the increasing number of workers.

\change{
\subsection{Performance of Incremental Analysis\label{subsec:eval-increment}}

In this section, we evaluate the performance of \tool in the mode of incremental analysis.
Due to the limit of evaluation costs, we select two representatives from the subjects in Table \ref{tab:subjects} as the experimental subjects, \ie, the largest Linux and the smallest Httpd. 
Moreover, to mimic the incremental scenario, we produce the analysis tasks by conducting the following steps: 1) we collect the real-world commits from the GitHub repositories of Linux and Httpd 30 days prior to our experiments. We ignored some of the commits which do not involve code changes; 2) we then divide these commits into 3-day and 6-day intervals to generate 10 and 5 code versions, respectively; 3) we construct the CFG of each version and compare the CFGs between two consecutive versions. The resulting CFG with updates would be the input of incremental analysis.

\begin{table}[htb!]
	\caption{Performance of incremental alias analysis. The column \textbf{Subject} shows the incremental versions used. The column class \textbf{Atomic-changes} indicates the percentages of nodes and edges atomically changed against the total number of nodes and edges in the entire CFG. The column class \textbf{Sub-CFG} demonstrates the percentages of nodes and edges in the affected sub-CFG against the total number of nodes and edges in the CFG. The column classes \toolwho and \toolinc show the analysis performance data in terms of the rental cost of cloud resources (\textbf{Cost}), the analysis time (\textbf{Time}), the number of workers used (\textbf{\#Workers}), and the amount of peak memory consumed (\textbf{\#PMem}) for whole-program and incremental analysis, respectively.}
	\label{tab:incre-performance-alias}
	\centering
	\scalebox{0.65}{
		\begin{tabular}{lrr|rr|rrrr|rrrr}
			\toprule
            & \multicolumn{2}{c}{Atomic-changes} & \multicolumn{2}{c}{Sub-CFG} & \multicolumn{4}{c}{\toolwho} & \multicolumn{4}{c}{\toolinc} \\ \cmidrule{2-3} \cmidrule{4-5} \cmidrule{6-9} \cmidrule{10-13}
            Subject & Nodes & Edges & Nodes & Edges & Cost & Time & \#Workers & \#PMem & Cost & Time & Workers & \#PMem \\
		\midrule
            \textsf{Httpd-3days-NU10} & 0.645\% & 0.592\% & 2.130\% & 2.147\% & 
             \$0.073 & 3.2mins & 10 & 120.1G & \$0.018 & 0.8mins & 10 & 40.0G  \\
            \textsf{Httpd-3days-NU9} & 0.000\% & 0.001\% & 1.040\% & 1.095\% & 
            \$0.070 & 3.1mins & 10 & 120.2G & \$0.018 & 0.8mins & 10 & 40.1G \\
            \textsf{Httpd-3days-NU8} & 0.002\% & 0.004\% & 2.233\% & 2.344\% & \$0.073 & 3.2mins & 10 & 120.8G & \$0.018 & 0.8mins & 10 & 41.0G \\
            \textsf{Httpd-3days-NU7} & 0.000\% & 0.001\% & 1.040\% & 1.095\% &  \$0.070 & 3.1mins & 10 & 120.2G & \$0.018 & 0.8mins & 10 & 40.1G \\
            \textsf{Httpd-3days-NU6} & 0.667\% & 0.694\% & 1.944\% & 1.998\% & \$0.070 & 3.1mins & 10 & 117.3G & \$0.016 & 0.7mins & 10 & 40.1G \\
            \textsf{Httpd-3days-NU5} & 0.000\% & 0.001\% & 1.051\% & 1.104\% &  \$0.073 & 3.2mins & 10 & 120.2G & \$0.018 & 0.8mins & 10 & 38.9G \\
            \textsf{Httpd-3days-NU4} & 0.000\% & 0.001\% & 1.051\% & 1.104\% &  \$0.073 & 3.2mins & 10 & 120.2G & \$0.018 & 0.8mins & 10 & 38.9G \\
            \textsf{Httpd-3days-NU3} & 0.000\% & 0.000\% & / & / & / & / & / & / & / & / & / & / \\
            \textsf{Httpd-3days-NU2} & 0.000\% & 0.000\% & / & / & / & / & / & / & / & / & / & / \\
            \textsf{Httpd-3days-NU1} & 0.000\% & 0.000\% & / & / & / & / & / & / & / & / & / & / \\
            \hline
            \textsf{Httpd-6days-NU5} & 0.645\% & 0.592\% & 2.130\% & 2.147\% & \$0.075 & 3.3mins & 10 & 117.3G & \$0.018 & 0.8mins & 10 & 40.0G \\
            \textsf{Httpd-6days-NU4} & 0.002\% & 0.004\% & 2.233\% & 2.344\% & \$0.077 & 3.4mins & 10 & 120.8G & \$0.018 & 0.8mins & 10 & 41.0G \\
            \textsf{Httpd-6days-NU3} & 0.667\% & 0.694\% & 1.944\% & 1.998\% & \$0.073 & 3.2mins & 10 & 117.3G & \$0.016 & 0.7mins & 10 & 40.1G \\
            \textsf{Httpd-6days-NU2} & 0.000\% & 0.001\% & 1.051\% & 1.104\% & \$0.073 & 3.2mins & 10 & 120.2G & \$0.018 & 0.8mins & 10 & 38.9G \\
            \textsf{Httpd-6days-NU1} & 0.000\% & 0.000\% & / & / & / & / & / & / & / & / & / & / \\
            \hline\hline
            \textsf{Linux-3days-NU10} & 0.007\% & 0.006\% & 0.423\% & 0.442\% & \$14.716 & 22.1mins & 350 & 4.3T & \$0.303 & 2.0mins & 80 & 0.3T \\
            \textsf{Linux-3days-NU9} & 6.440\% & 6.261\% & 17.310\% & 20.595\% & \$14.716 & 22.1mins & 350 & 4.3T & \$5.248 & 10.2mins & 270 & 3.3T \\
            \textsf{Linux-3days-NU8} & 2.465\% & 2.689\% & 15.430\% & 18.949\% & \$14.982 & 22.5mins & 350 & 4.3T & \$4.271 & 8.3mins & 270 & 3.3T \\
            \textsf{Linux-3days-NU7} & 2.316\% & 2.524\% & 15.344\% & 18.862\% &  \$14.716 & 22.1mins & 350 & 4.3T & \$4.219 & 8.2mins & 270 & 3.3T \\
            \textsf{Linux-3days-NU6} & 0.088\% & 0.110\% & 15.602\% & 19.152\% & \$14.782 & 22.2mins & 350 & 4.3T & \$4.374 & 8.5mins & 270 & 3.0T \\
            \textsf{Linux-3days-NU5} & 2.653\% & 2.845\% & 16.188\% & 19.720\% & \$14.982 & 22.5mins & 350 & 4.3T & \$4.682 & 9.1mins & 270 & 3.3T \\
            \textsf{Linux-3days-NU4} & 0.302\% & 0.303\% & 16.498\% & 20.070\% & \$15.115 & 22.7mins & 350 & 4.3T & \$4.528 & 8.8mins & 270 & 3.0T \\
            \textsf{Linux-3days-NU3} & 2.514\% & 2.737\% & 15.880\% & 19.415\% & \$15.115 & 22.7mins & 350 & 4.3T & \$4.476 & 8.7mins & 270 & 3.2T \\
            \textsf{Linux-3days-NU2} & 2.811\% & 3.135\% & 42.206\% & 42.355\% & \$14.982 & 22.5mins & 350 & 4.3T & \$6.740 & 13.1mins & 270 & 3.2T \\
            \textsf{Linux-3days-NU1} & 8.886\% & 9.333\% & 45.961\% & 46.179\% & \$14.716 & 22.1mins & 350 & 4.3T & \$7.718 & 15.0mins & 270 & 3.2T \\
            \hline
            \textsf{Linux-6days-NU5} & 6.441\% & 6.261\% & 17.311\% & 20.596\% & \$14.582 & 21.9mins & 350 & 4.3T & \$5.093 & 9.9mins & 270 & 3.3T \\
            \textsf{Linux-6days-NU4} & 2.468\% & 2.690\% & 15.434\% & 18.952\% & \$14.849 & 22.3mins & 350 & 4.3T & \$4.219 & 8.2mins & 270 & 3.2T \\
            \textsf{Linux-6days-NU3} & 2.700\% & 2.902\% & 16.545\% & 20.094\% & \$14.782 & 22.2mins & 350 & 4.3T & \$4.837 & 9.4mins & 270 & 3.2T \\
            \textsf{Linux-6days-NU2} & 2.769\% & 2.978\% & 16.879\% & 20.452\% & \$15.049 & 22.6mins & 350 & 4.3T & \$4.940 & 9.6mins & 270 & 3.3T \\
            \textsf{Linux-6days-NU1} & 9.198\% & 9.663\% & 47.477\% & 47.764\% & \$14.649 & 22.0mins & 350 & 4.3T & \$8.130 & 15.8mins & 270 & 3.2T \\
		\end{tabular}
	}
\end{table}

The first column in Table \ref{tab:incre-performance-alias} and \ref{tab:incre-performance-cache} lists the incremental versions used for alias analysis and instruction cache analysis, respectively. The second and third columns indicate the percentages of nodes and edges atomically changed against the total number of nodes and edges in the CFG. The fourth and fifth columns demonstrate the percentages of nodes and edges in the affected sub-CFG against the total number of nodes and edges in the entire CFG. The column classes \toolwho and \toolinc show the analysis performance data in terms of the rental cost of cloud resources (Cost), the analysis time (Time), the number of workers used (\#Workers), and the amount of peak memory consumed (\#PMem) for whole-program and incremental analysis, respectively.

\begin{table}[htb!]
	\caption{Performance of incremental instruction cache analysis. The column \textbf{Subject} shows the incremental versions used. The column class \textbf{Atomic-changes} indicates the percentages of nodes and edges atomically changed against the total number of nodes and edges in the entire CFG. The column class \textbf{Sub-CFG} demonstrates the percentages of nodes and edges in the affected sub-CFG against the total number of nodes and edges in the CFG. The column classes \toolwho and \toolinc show the analysis performance data in terms of the rental cost of cloud resources (\textbf{Cost}), the analysis time (\textbf{Time}), the number of workers used (\textbf{\#Workers}), and the amount of peak memory consumed (\textbf{\#PMem}) for whole-program and incremental analysis, respectively.}
	\label{tab:incre-performance-cache}
	\centering
	\scalebox{0.65}{
		\begin{tabular}{lrr|rr|rrrr|rrrr}
			\toprule
            & \multicolumn{2}{c}{Atomic-changes} & \multicolumn{2}{c}{Sub-CFG} & \multicolumn{4}{c}{\toolwho} & \multicolumn{4}{c}{\toolinc} \\ \cmidrule{2-3} \cmidrule{4-5} \cmidrule{6-9} \cmidrule{10-13}
            
            Subject & Nodes & Edges & Nodes & Edges & Cost & Time & \#Workers & \#PMem & Cost & Time & Workers & \#PMem \\
		\midrule
            \textsf{Httpd-3days-NU10} & 1.466\% & 1.388\% & 1.622\% & 1.582\% & \$0.643 & 3.4mins & 100 & 626.5G & \$0.038 & 1.0mins & 20 & 90.6G \\
            \textsf{Httpd-3days-NU9} & 0.000\% & 0.000\% & / & / & / & / & / & / & / & / & / & / \\
            \textsf{Httpd-3days-NU8} & 13.518\% & 13.600\% & 6.920\% & 6.960\% & \$0.662 & 3.5mins & 100 & 623.2G & \$0.129 & 1.7mins & 40 & 206.2G \\
            \textsf{Httpd-3days-NU7} & 0.000\% & 0.000\% & / & / & / & / & / & / & / & / & / & / \\
            \textsf{Httpd-3days-NU6} & 0.706\% & 0.704\% & 0.346\% & 0.346\% & \$0.681 & 3.6mins & 100 & 626.5G & \$0.030 & 0.8mins & 20 & 75.6G \\
            \textsf{Httpd-3days-NU5} & 0.000\% & 0.000\% & / & / & / & / & / & / & / & / & / & / \\
            \textsf{Httpd-3days-NU4} & 0.000\% & 0.000\% & / & / & / & / & / & / & / & / & / & / \\
            \textsf{Httpd-3days-NU3} & 0.000\% & 0.000\% & / & / & / & / & / & / & / & / & / & / \\
            \textsf{Httpd-3days-NU2} & 0.000\% & 0.000\% & / & / & / & / & / & / & / & / & / & / \\
            \textsf{Httpd-3days-NU1} & 0.000\% & 0.000\% & / & / & / & / & / & / & / & / & / & / \\
            \hline
            \textsf{Httpd-6days-NU5} & 1.466\% & 1.388\% & 1.622\% & 1.582\% & \$0.643 & 3.4mins & 100 & 626.5G & \$0.038 & 1.0mins & 20 & 90.6G \\
            \textsf{Httpd-6days-NU4} & 13.518\% & 13.600\% & 6.920\% & 6.960\% & \$0.662 & 3.5mins & 100 & 623.2G & \$0.129 & 1.7mins & 40 & 206.2G \\
             \textsf{Httpd-6days-NU3} & 0.706\% & 0.704\% & 0.346\% & 0.346\% & \$0.681 & 3.6mins & 100 & 626.5G & \$0.030 & 0.8mins & 20 & 75.6G \\
             \textsf{Httpd-6days-NU2} & 0.000\% & 0.000\% & / & / & / & / & / & / & / & / & / & / \\
            \textsf{Httpd-6days-NU1} & 0.000\% & 0.000\% & / & / & / & / & / & / & / & / & / & / \\ 
            \hline \hline
            \textsf{Linux-3days-NU10} & 0.001\% & 0.001\% & 0.000\% & 0.000\% & \$42.040 & 46.3mins & 480 & 4.9T & \$0.026 & 0.7mins & 20 & 31.8G \\
            \textsf{Linux-3days-NU9} & 0.031\% & 0.032\% & 0.026\% & 0.027\% & \$42.131 & 46.4mins & 480 & 4.9T & \$0.061 & 1.6mins & 20 & 52.1G \\
            \textsf{Linux-3days-NU8} & 0.197\% & 0.201\% & 0.800\% & 0.814\% & \$42.313 & 46.6mins & 480 & 4.9T & \$0.076 & 2.0mins & 20 & 71.2G \\
            \textsf{Linux-3days-NU7} & 0.237\% & 0.263\% & 0.119\% & 0.132\% & \$41.677 & 45.9mins & 480 & 4.9T & \$0.038 & 1.0mins & 20 & 101.6G \\
            \textsf{Linux-3days-NU6} & 0.000\% & 0.000\% & / & / & / & / & / & / & / & / & / & / \\
            \textsf{Linux-3days-NU5} & 0.987\% & 1.083\% & 1.547\% & 1.622\% & \$38.771 & 42.7mins & 480 & 4.9T & \$0.371 & 4.9mins & 40 & 267.0G \\
            \textsf{Linux-3days-NU4} & 0.871\% & 0.884\% & 0.983\% & 1.032\% & \$38.862 & 42.8mins & 480 & 4.9T & \$0.144 & 1.9mins & 40 & 203.4G \\
            \textsf{Linux-3days-NU3} & 0.317\% & 0.316\% & 0.446\% & 0.477\% & \$38.862 & 42.8mins & 480 & 4.9T & \$0.045 & 1.2mins & 20 & 42.7G \\
            \textsf{Linux-3days-NU2} & 1.073\% & 1.157\% & 0.869\% & 0.944\% & \$38.953 & 42.9mins & 480 & 4.9T & \$0.238 & 2.1mins & 60 & 302.9G \\
            \textsf{Linux-3days-NU1} & 0.346\% & 0.340\% & 0.172\% & 0.169\% & \$39.044 & 43.0mins & 480 & 5.2T & \$0.042 & 1.1mins & 20 & 69.9G \\
            \hline
            \textsf{Linux-6days-NU5} & 0.032\% & 0.033\% & 0.027\% & 0.027\% & \$41.586 & 45.8mins & 480 & 4.9T & \$0.061 & 1.6mins & 20 & 48.9G \\
            \textsf{Linux-6days-NU4} & 0.440\% & 0.471\% & 0.930\% & 0.957\% & \$41.586 & 45.8mins & 480 & 4.9T & \$0.087 & 2.3mins & 20 & 141.1G \\
            \textsf{Linux-6days-NU3} & 0.987\% & 1.083\% & 1.547\% & 1.622\% & \$38.772 & 42.7mins & 480 & 4.9T & \$0.371 & 4.9mins & 40 & 267.0G \\
            \textsf{Linux-6days-NU2} & 0.957\% & 0.971\% & 1.027\% & 1.076\% & \$38.862 & 42.8mins & 480 & 4.9T & \$0.114 & 1.5mins & 40 & 199.5G \\
            \textsf{Linux-6days-NU1} & 1.433\% & 1.511\% & 1.053\% & 1.125\% & \$39.044 & 43.0mins & 480 & 5.2T & \$0.363 & 3.2mins & 60 & 302.9G \\
		\end{tabular}
	}
\end{table}

As can be seen, different versions are of different change scales. Some commits only modify a few lines of code, while others can result in thousands of edges modification.
In alias analysis shown as Table \ref{tab:incre-performance-alias}, for Httpd, three 3-day changed versions (\ie, Httpd-3days-NU1,2,3) and one 6-day changed version (\ie, Httpd-6days-NU1) have zero changes. The change ratios in other versions range from 0.001\% to 0.694\%.
For Linux, the majority of the change ratios are concentrated between 2\% and 10\%.
Low ratios of change are shown in the specific cases (\ie, Linux-3days-NU10) compared to to other cases.
This is because some main modules in Linux do not involve code changes in specific versions, thereby leading to a low change ratio in the whole Linux project.
In cache analysis shown as Table \ref{tab:incre-performance-cache}, for Httpd, seven 3-day changed versions (\ie, Httpd-3days-NU1,2,3,4,5,7,9) and two 6-day changed versions have zero changes (\ie, Httpd-6days-NU1,2). For Linux, one version has no ratio of change (\ie, Linux-3days-NU6) and the other ratios of change are concentrated within 1.6\%.
Moreover, the changed versions from 3-day and 6-day intervals exhibit identical change ratio in some cases. For instance in alias analysis, though Httpd-3days-NU5 and Httpd-3days-NU6 both involve changes, Httpd-6days-NU3 is identical with Httpd-3days-NU6, rather than representing the cumulative effect of two changes. We analyzed the changes caused by these two commits and found that the modified parts in the previous commit are altered again by the second commit. This overlap results in identical change to the CFG, which is reflected in identical change ratio. 
This particular case is also observed in Linux. For example, although Linux-3days-NU1 and Linux-3days-NU2 have different ratios of change, their sub-CFGs are of similar size. This is because that  changes in the Linux/acpi module affect similar parts of the code, and thus have a similar impact on the structure of the CFG. 
Regarding the sub-CFG affected, the scale of the sub-CFG directly determines the computational resource consumption of incremental analysis.
As can be seen, different subjects under different client analyses show different characteristics. 
In general, the size of sub-CFG under alias analysis is apparently larger than that of instruction cache analysis. 
This is because that we adopt a sparse representation of CFG in alias analysis, whereas the traditional CFG representation is utilized in cache analysis. 

Regarding analysis performance, incremental analysis demonstrates superior performance in both alias analysis and instruction cache analysis. 
In alias analysis of Linux, the whole-program analysis (\ie, \toolwho) requires around 22 minutes, 4.3 terabytes of peak memory with 350 workers, with the rental cost around 15 dollars. In contrast, incremental dataflow analysis (\ie, \toolinc) only needs on average 10 minutes, 3.3 TB peak memory with 270 workers. The rental cost is directly cut by two-third. 
For the small subject Httpd, the analysis time and rental cost by incremental analysis is cut by three-fourth.
It lowers peak memory consumption to one-third of what is required by whole-program analysis with the same number of workers.
In the more expensive cache analysis, incremental dataflow analysis shows larger performance advantage against whole-program analysis. The cache analysis of Linux can be completed within 2 minutes in most of the versions. The peak memory consumption drops from 5 terabytes to only hundreds of gigabytes. The number of workers needed is reduced greatly from 480 to 20-60. The rental cost is saved by 2 orders of magnitude.
For the smaller Httpd, incremental dataflow analysis  manages to complete the analysis using just 40 workers within 1 or 2 minutes. Both the number of workers and rental cost are significantly reduced as well. 
These results empirically validate that our incremental analysis offers superior efficiency in terms of both time and cost, which is well-suitable to the modern CI/CD scenarios.

}

\section{Discussion}
\label{sec:discuss} 

\MyPara{Usage Scenarios.}\change{
\tool offers the distributed capability in lifting sophisticated dataflow analysis to large-scale programs.  It's highly valuable for organizations with large codebases to analyze, while often with their own cluster deployed. In such scenarios, \tool readily offers both high-speed and scalable analysis to ultra-large-scale programs.
Moreover, the incremental analysis offers further efficiency by performing incremental analysis on program changes rather than re-analyzing the entire codebase. This makes \toolinc to be well-suited for dynamic environments where frequent updates commonly occurred in industrial codebases. 
}

\MyPara{Soundness.}
Like other analysis frameworks (e.g., Soot and WALA), users implement a particular client analysis by specifying its corresponding merge and transfer functions. It is the analysis developer’s responsibility to ensure the soundness of their analysis; \tool faithfully executes whatever has been implemented by the developers. As for the underlying framework, we adopted the classic worklist algorithm \cite{Kam:1976-iterative}, and directly implemented it based on the vertex-centric computation model. Its soundness stays unchanged.


\section{Related Work}
\label{sec:related}

\subsection{Parallel and Distributed Static Analysis}
Over the past decades, a few attempts have been made to speed up static program analysis by leveraging parallel and distributed computing facilities. 
Lee and Ryder \cite{lee1992comprehensive} exploited algorithmic parallelism to accelerate dataflow analysis. 
Rodriguez et al.\cite{Rodriguez-2011} proposed a parallel algorithm for IFDS-based dataflow analysis \cite{Reps-1995-dataflow} based on the actor model, which requires the transfer function to be distributive over meet operators. 
Nagaraj and Govindarajan \cite{nagaraj2013parallel} utilized Intel Threading Building Blocks to design a parallel flow-sensitive pointer analysis algorithm. 
Su et al. \cite{su2014parallel} proposed parallel CFL-reachability-based flow-insensitive pointer analysis.
Importantly, all the above approaches rely heavily on memory for computation. They can rarely scale to large systems such as Linux kernel.

Following the line of systemizing program analysis, various systems are developed to support scalable interprocedural analysis. Graspan \cite{Wang:2017-graspan,zuo:graspan-tocs20} and BigSpa \cite{Zuo:2018-bigspa-ipdps}  scale the context-sensitive CFL-reachability analysis \cite{Reps-1997-reachability} in a single machine and distributed environment, respectively. 
Unfortunately, many dataflow analyses cannot fall into this category, such as cache analysis and interval analysis.
Chianina \cite{zuo2021chianina} is an out-of-core graph system performing the context- and flow-sensitive analyses in parallel. However, restricted by the limited parallel computing resources in a single machine, it is inefficient when analyzing large-scale systems code. 

For distributed work, Albarghouthi et al. \cite{Albarghouthi:2012-pldi-bolt} took the inspirations from MapReduce paradigm and parallelized the demand-driven top-down analyses, such as verification and software model checking. They failed to support dataflow analysis.
Garbervetsky et al. \cite{Garbervetsky:2017-fse} recently devised a distributed worklist algorithm based on the actor model. However, it does not support the standard dataflow analysis due to the absence of flow ordering between actors.
\textcolor{black}{Christakis et al. \cite{inputsplit-fse-2022} explored input splitting strategies to analyze different code pieces on parallel partitions independently. However, as stated explicitly, the splitting 
causes analysis imprecision due to the information loss across separate partitions.}
Greathouse et al. \cite{Greathouse-cgo-2011} extended dynamic dataflow analyses with a novel sampling system to achieve low runtime overhead. Apparently, they only focused on dynamic analysis rather than static dataflow analysis.

\subsection{Incremental Static Analysis\label{subsec:related-incremental}}



\change{
To efficiently analyze codebases with frequent changes, various studies \cite{ryder1983incremental, pollock1989incremental, marlowe1989efficient, burke1990critical, sreedhar1998new, arzt2014reviser, harianto2024enhancing, szabo2016inca, szabo2021incremental, van2020incremental, wauters2023change} have been proposed to realize incremental analysis over the past decades.
Pollock and Soffa \cite{pollock1989incremental} presented an incremental iterative algorithm for constructing incremental versions of kill-gen problems.
Reviser \cite{arzt2014reviser} presents a revision algorithm for interprocedural distributive environment transformers (IDE) and IFDS-based dataflow analysis in response to incremental program changes.
Harianto et al. \cite{harianto2024enhancing} presented a data structure called SP-graph to enhance incremental dataflow anomaly detection in Integrated Development Environments (IDEs).
IncA \cite{szabo2016inca} is a domain-specific language for the definition of efficient incremental program analyses. It represents program analysis tasks as graph patterns and updates the analysis results incrementally by graph pattern matching.
Laddder \cite{szabo2021incremental} is a novel differential dataflow-based solver for efficient Datalog-based incremental analysis. 
Van and Wauters \cite{van2020incremental, wauters2023change} realized incremental modular analysis by leveraging the modularity of the program and dependencies between modules. 
Besides, there are attempts for incremental pointer and alias analysis \cite{yur1999incremental, saha2005incremental, lu2013incremental, liu2019rethinking, liu2022sharp}. 
However, the above studies mainly focus on sequential computation; none of them supports distributed analysis. 
To the best of our knowledge, \toolinc is the first work which leverages the distributed cloud resources to accelerate the incremental monotone dataflow analysis.
}

\subsection{Vertex-Centric Graph Processing}
Vertex-centric model has been tightly incorporated into distributed processing frameworks to tackle the challenge of large-scale graph processing. 
Based on that, Pregel \cite{Malewicz:2010-pregel} is the pioneering system supporting general graph applications.  
Pregel adopts BSP model to accelerate the intensive computation. 
Following the idea of Pregel, Apache Giraph \cite{avery2011giraph} is implemented in Java as an open source system.  
Following Pregel, more advanced vertex-centric models and variants have been proposed. 
GraphLab \cite{low2012distributed} supports asynchronous vertex computation based on Chandy-Lamport snapshots without halting the entire program. 
GraphX \cite{Gonzalez:2014-graphx} is a graph system based on Resilient Distributed Dataset (a.k.a., RDD) abstraction. 
Note that all the above graph systems are dedicated to the general graph applications. None of them can directly scale the interprocedural dataflow analysis well. As a result, we propose \tool, the first distributed system tailored to dataflow analysis. 

\section{Conclusion}
\label{sec:conclude}
\change{
This paper proposes a distributed dataflow analysis framework called \tool running on a real-world cloud, which supports two modes: whole-program analysis and incremental analysis. By leveraging the large amount of memory and CPU cores in the cloud, \tool greatly improves the scalability and efficiency of dataflow analysis for analyzing large-scale programs. The experiments conducted in a real-world cloud environment validate that \tool not only scales context-sensitive dataflow analysis to million lines of code, but also completes such analysis in a highly efficient manner. It can be expected that we could achieve nearly on-the-fly analysis of industrial-scale codebases by leveraging modern cloud  computing facilities.
 }
	
\section{Data Availability}
\label{sec:data-availability}
\change{
All the source code of \tool is publicly available at: \href{https://github.com/BigDataflow-system}{https://github.com/BigDataflow-system}.
}

\begin{acks}
We would like to thank all the anonymous reviewers for their valuable comments.
This work is partially supported by the National Natural Science Foundation of China (Grant No. 62272217), the Fundamental Research Funds for the Central Universities (Grant No. 020214380104) and Alibaba Group via the Alibaba Innovation Research (AIR) program.
\end{acks}

  \bibliographystyle{ACM-Reference-Format}

	\bibliography{refs}

\end{document}